\documentclass[%
reprint,
nofootinbib,
mathtools,
amssymb,
aps,
pra,
superscriptaddress,
usenames,dvipsnames
]{revtex4-2}

\usepackage[
text={7.05in,10in},centering,
top = .84in, bottom=.84in
]{geometry}

\usepackage{graphicx} 
\usepackage[caption=false]{subfig} 
\usepackage{array} 
\usepackage{multirow}
\usepackage{comment}
\usepackage[shortlabels]{enumitem}

\usepackage{bm}
\usepackage[tbtags]{amsmath}
\usepackage{amsthm}
\usepackage{nicefrac}

\usepackage{tikz}
\usetikzlibrary{quantikz}

\usepackage[linesnumbered,ruled,vlined,commentsnumbered]{algorithm2e} 
\usepackage{complexity} 

\usepackage{xcolor}

\usepackage[colorlinks=true,linktocpage=true]{hyperref}%
\hypersetup{
    allcolors  = {NavyBlue},
}

\theoremstyle{plain}
\newtheorem{theorem}{Theorem}
\newtheorem{corollary}[theorem]{Corollary}
\newtheorem{lemma}[theorem]{Lemma}

\theoremstyle{plain}
\newtheorem*{definition*}{Definition}

\theoremstyle{definition}
\newtheorem{remark}{Remark}

\newcommand{\inlshort}{International Iberian Nanotechnology Laboratory (INL), Braga, Portugal}

\newcommand{\fcupshort}{Departamento de F\'{i}sica e Astronomia, Faculdade de Ci\^{e}ncias, Universidade do Porto, Porto, Portugal}

\newcommand{\uffshort}{Instituto de F\'{i}sica, Universidade Federal Fluminense, Niter\'{o}i, Rio de Janeiro, Brazil}

\newcommand{\ugrshort}{Quantum Thermodynamics and Computation Group, Electromagnetism and Matter Physics Department, University of Granada, Granada, Spain}

\usepackage{soul}

\begin{document}

\title{Reducing depth and measurement weights in Pauli-based computation}

\author{Filipa C. R. Peres}
\email[Corresponding author: ]{fcrperes@onsager.ugr.es}
\affiliation{\inlshort}
\affiliation{\fcupshort}
\affiliation{\ugrshort}

\author{Ernesto F. Galvão}
\affiliation{\inlshort}
\affiliation{\uffshort}

\date{\today}

\begin{abstract}
Pauli-based computation (PBC) is a universal measurement-based quantum computation model steered by an adaptive sequence of independent and compatible Pauli measurements on separable magic-state qubits. Here, we propose several new techniques for reducing the weight of the Pauli measurements and their associated \textsc{cnot} complexity; we also demonstrate how to decrease this model's computational depth. We start by proving new upper bounds on the required weights and computational depth, obtained via a pre-compilation step. We also propose a heuristic algorithm that can contribute reductions of over 30\% to the average weight of Pauli measurements (and associated \textsc{cnot} count) when simulating and compiling Clifford-dominated random quantum circuits with up to 22 $T$ gates and over 20\% for instances with larger $T$ counts. This PBC-compilation scheme, boosted by the heuristic algorithm, outperforms state-of-the-art compilers for the former circuits, reducing the \textsc{cnot} count by 18\% to 96\% compared with the values achieved by other techniques. In contrast, for the latter circuits with larger $T$ counts, it leads to a number of \textsc{cnot}s roughly 30\% larger. Finally, inspired by known state-transfer methods, we introduce incPBC, a universal model for quantum computation requiring a larger number of (now incompatible) Pauli measurements of weight at most 2.
\end{abstract}

\maketitle

\section{Introduction}\label{sec: Introduction}

Quantum computing stands at the forefront of technological innovation, promising unprecedented computational power and transformative potential. However, current hardware remains somewhat limited in its capabilities. Thus, the minimization of quantum resources promises to play a relevant role in current and near-term implementations. This has prompted extensive research into multiple quantum circuit optimization techniques such as pattern matching~\cite{Maslov2008templates, ItenMMSW2022pm}, ZX-calculus~\cite{Duncan2020simplicationzx, KissWet2020reduceT}, and phase polynomials~\cite{AmyMM2014, Heyfron2018, Amy2019}, leading to important reductions in the number of operations and/or circuit depth.

Quantum circuits stand as the dominant framework underlying these works. In the quantum circuit model~\cite{NielsenChuang, Yao1993}, the computation requires three essential steps: (i) the preparation of an input state (typically, the $\ket{0}^{\otimes n}$ state), (ii) its coherent unitary evolution via the sequential application of quantum gates drawn from a universal set, and (iii) the measurement of the final state in the computational basis. The fact that this is an intuitive framework with a classical analog has likely contributed to the prevalence of quantum circuits in the field of quantum computing. Nevertheless, measurement-based models have emerged as intriguing alternatives offering unique insights and potential benefits.

In this paper, we explore Pauli-based computation (PBC)~\cite{BSS2016}, a universal model for quantum computation driven by an adaptive sequence of at most $n$ independent and pairwise commuting (non-destructive) multi-qubit Pauli measurements performed on a set of $n$ qubits initialized in a separable magic state. 
PBC remains a relatively unexplored model compared with other measurement-based schemes. Here, we aim to gain a better understanding of this framework and improve its feasibility. Since stabilizer codes naturally enable a simple fault-tolerant implementation of non-destructive Pauli measurements, PBC constitutes an interesting primitive for fault-tolerant quantum computing. Recent work~\cite{PeresGa2023} also demonstrated that PBC can be useful for compiling some families of quantum circuits dominated by Clifford gates; more specifically, the overall gate counts and depth of many of these quantum circuit instances can be significantly reduced by transforming each of them into a PBC and then translating it into adaptive Clifford circuits with magic-state input. The latter are called PBC-compiled circuits~\cite{PeresGa2023}. Importantly, the overall depth and gate counts of the final circuits are intimately connected to the \emph{weight} of the Pauli measurements in the corresponding PBC. For this reason, this is \emph{the} measure to optimize within the PBC framework to enable not only a more practical (native) implementation of the model itself but also to enhance its efficacy as a circuit compilation tool. While our work is rooted in the goal of acquiring a deeper insight into PBC and improving its native efficiency, this intimate connection to circuit compilation allows direct conclusions about its performance as a compilation tool.

\subsection*{Structure of this work and summary of the main contributions}

\begin{figure*}[t]
    \centering
    \includegraphics[width=0.95\linewidth]{figs/Workflow_impr.png}
    \caption{The usual workflow to perform a Pauli-based computation (PBC). Starting from a Clifford+$T$ quantum circuit with $n$ qubits, $t$ $T$ gates, and $w$ computational-basis readout measurements, step \textcircled{\raisebox{-0.9pt}{1}} consists of transforming the circuit using magic-state injection. The result of such a transformation is an adaptive Clifford circuit on $n+t$ qubits and with $t+w$ measurements. Such a circuit can be easily transformed into a PBC: The classical computer efficiently finds the next Pauli measurement, decides its outcome (classically) when possible, and queries the quantum computer when necessary. The work in Sec.~\ref{sec: Improvements I - connection to 1WC} targets step~\textcircled{\raisebox{-0.9pt}{2}}: A classically efficient pre-compilation step is added before making the PBC translation. This is leveraged to establish new upper bounds for the weight of the Pauli measurements and the computational depth of PBC. The work in Sec.~\ref{sec: Greedy algorithm} addresses step~\textcircled{\raisebox{-0.95pt}{3}}: A new heuristic (classical) algorithm replaces $P_j$ with an equivalent measurement $P_j^{\prime}$ with lower weight, before the quantum hardware is queried.}
    \label{fig: Workflow}
\end{figure*}

This paper is organized as follows: Section~\ref{sec: Background} provides the necessary background for the understanding of this work. We start by recalling the Pauli and Clifford groups and presenting different notions of universality (Secs.~\ref{subsec: Pauli and Clifford groups} and \ref{subsec: Notions of universality}). In Sec.~\ref{subsec: MSI}, we provide a brief description of the magic-state injection model. Section~\ref{subsec: MBQC} presents a detailed review of measurement-based quantum computation (MBQC) focused on one-way quantum computing (1WQC), teleportation-based and state-transfer-based computation, and PBC.

Sections~\ref{sec: Improvements I - connection to 1WC} and~\ref{sec: Greedy algorithm} contain our main results. A pictorial description of these contributions can be seen in Fig.~\ref{fig: Workflow}. In Sec.~\ref{sec: Improvements I - connection to 1WC}, by introducing a pre-compilation step before transforming any adaptive Clifford circuit into a PBC, we achieve significant improvements in the overall complexity of the Pauli measurements. Namely, we derive new upper bounds for the weights of the Pauli operators and demonstrate that the adaptive structure of any PBC allows some of the measurements to be carried out concomitantly, establishing novel bounds for the computational depth of this model. Importantly, the best weight and depth upper bounds found do \emph{not} hold simultaneously. To address this, we show that a weight-depth compromise is possible. We also perform numerical experiments that provide evidence that, other than establishing previously missing upper bounds, the pre-compilation leads to practical benefits: It helps reduce the average weight of the Pauli measurements.

Besides these contributions, in Sec.~\ref{sec: Greedy algorithm}, we propose a novel greedy algorithm capable of further reducing the weight of the Pauli measurements. Numerical results suggest that this heuristic algorithm achieves important reductions that, as explained above, have direct consequences on the practicality of PBC, but also on the depth and gate counts of the PBC-compiled circuits. The algorithm searches over alternative Pauli measurements that project the system onto the same desired eigenspace, selecting the one with the lowest weight. The algorithm's performance depends on the size of the search space. For hidden-shift circuits, reductions of more than 10\% and 12\% to the average weight of the Pauli measurements can be achieved by searching through a number of alternative measurements, respectively, linear and quadratic in the number of qubits of the PBC. These reductions are reported with respect to the average weight of the Pauli measurements obtained without the greedy algorithm. On the other hand, for random quantum circuits, improvements of over 15\% and 20\% to the average weight are consistently achieved by the linear and quadratic orders of the algorithm for the largest instances tested. For smaller random quantum circuits, these improvements exceed 20\% and 30\%, respectively, with respect to the average weight obtained in the absence of the greedy algorithm.

The close connection between the average weight of the Pauli measurements in PBC and the \textsc{cnot} count of the associated PBC-compiled circuits allows a straightforward comparison between our PBC compilation scheme, aided by the greedy algorithm, and other state-of-the-art circuit compilation techniques. In this case, also in Sec.~\ref{sec: Greedy algorithm}, and focusing on the case where the greedy algorithm runs through a quadratic number of alternative measurements, we report that our PBC-compiler is consistently better than other compilers for hidden-shift circuits across different parameter regimes (except for the smaller circuits with only 10 qubits and 14 $T$ gates, where our technique performs worse). More specifically, the reductions in the total number of two-qubit gates range between 9\% and 69\% compared with the corresponding number of such gates obtained with other compilers used for the comparison. In the case of random quantum circuits with 25 qubits and $T$ counts up to 22, our technique provides improvements between 18\% and 96\% to the total number of \textsc{cnot} gates compared with the results of other compilers, while for larger circuits with larger $T$ counts, PBC-compilation underperforms compared with those same compilers. 

In Section~\ref{sec: incPBC}, we wonder whether it is possible to perform universal quantum computation by carrying out constant-weight Pauli measurements on a separable input state. In answer to this question, we introduce a new model for quantum computation that uses only Pauli measurements of weight 1 or 2 at the expense of relaxing two of the defining properties of PBC -- namely, that the number of measurements is at most the number of qubits and that these measurements are all pairwise commuting. We call this new model incompatible, constant-weight Pauli-based computation (incPBC) to distinguish it from the original PBC formulation. We demonstrate the model's universality and investigate the quantum resources needed to perform a given computation within this framework compared with 1WQC and standard PBC.

In Section~\ref{sec: Conclusions}, we comment on important aspects of our main contributions and outline interesting new lines of research.\\

\section{Background}\label{sec: Background}

This work rests on many different concepts, ranging from the stabilizer formalism and the Pauli and Clifford groups to the magic-state injection model and several different measurement-based models.  In this section, we make no pretension of giving a comprehensive description of all of these topics. Instead, we try to strike a balance between being comprehensive enough for an unfamiliar reader to gain a sufficient understanding of all the concepts while avoiding the presentation of excessive details. With that in mind, we point the interested reader toward alternative (more in-depth) references where appropriate.

\subsection{Pauli and Clifford groups}\label{subsec: Pauli and Clifford groups}

Quantum circuits are the most widespread model for describing quantum computations. They work by first preparing an input state which, without loss of generality, can be the $\ket{0}^{\otimes n}$ state. Then, a set of coherent unitary operations, known as quantum gates, are applied in an appropriate sequence, $U$, to the initial state, preparing the state $\ket{\psi_{f}} = U \ket{0}^{\otimes n}$. The final readout is done via computational basis measurements on $\ket{\psi_f}$.

The Pauli and Clifford unitaries are operations that have a pivotal role in quantum computing in general and in our work in particular. An $n$-qubit Pauli operator is constructed by the $n$-fold tensor product of single-qubit Pauli operators ($I$, $X$, $Y$, and $Z$) multiplied by one of four possible phases $\{\pm 1, \pm i\}.$ These operators form a group known as the $n$-qubit Pauli group, or simply, the Pauli group, denoted $\mathcal{P}_n$. We say that an $n$-qubit Pauli operator $P$ has weight $w\leq n$ if it involves $w$ non-identity single-qubit Pauli operators. For instance, the Pauli operator $P = X\otimes I\otimes Y\otimes Z$ is a Pauli operator of weight three. To simplify notation, we often omit the tensor product and identity from the description of multi-qubit Pauli operators and associate with each non-trivial single-qubit Pauli a subscript indicating the qubit it acts on. Using this convention, the operator in the previous example is written simply as $P = X_1Y_3Z_4.$

An $n$-qubit quantum state, $\ket{\phi}$, is said to be a stabilizer state if it is the simultaneous eigenvector of $n$ independent and pairwise commuting Pauli operators with eigenvalue +1: $G_i\ket{\phi} = \ket{\phi},\,\forall i\in \{1,\dots,n\}.$ The operators $G_i$ generate an abelian group called the stabilizer group: $\mathcal{S} = \langle G_1,\dots,G_n \rangle$ which has $\left| \mathcal{S} \right| = 2^n$ elements. Importantly, any stabilizer state can be uniquely defined by the set of generators, $\{G_i\}_{i=1}^n$, of its stabilizer group. Readers who are unfamiliar with the stabilizer formalism are referred to Refs.~\cite[Section~10.5.]{NielsenChuang} and~\cite{Gottesman1998}.

The Clifford group on $n$ qubits, $\mathcal{C}_n$, is the normalizer of the Pauli group, $\mathcal{C}_n \coloneqq \{C\in \mathrm{U}({2^n}): C \mathcal{P}_n C^{\dagger} = \mathcal{P}_n\}$, and is generated by the Hadamard ($H$), phase ($S$), and controlled-\textsc{not} ($CX$) gates~\cite{PhDGottesman, CalderbankRSS1997qec}. The beautiful Gottesman-Knill theorem~\cite{PhDGottesman} states that any quantum circuit with only stabilizer state inputs, Clifford gates, and Pauli measurements is efficiently classically simulable. These circuits are known as stabilizer circuits. Since we expect quantum computers to be strictly more powerful than their classical counterparts, this result indicates that, although stabilizer circuits can be highly entangling, they are \emph{not} universal for quantum computation. 

\subsection{Universality in quantum computation}\label{subsec: Notions of universality}

Since we are interested in the ability to perform \emph{any} quantum computation, it is important to discuss which operations enable universality. However, even before that, the notion of universality itself needs to be clarified.

In Ref.~\cite{NestDMB2007}, the authors define a CQ-universal quantum computer as a device that, taking as input a classical bit string (say $0^n$), is capable of preparing any quantum state $\ket{\psi_f}$. Put differently, given any unitary $U$, a CQ-universal quantum computer can prepare the state $\ket{\psi_f} = U \ket{0}^{\otimes n}$. This corresponds to the strong notion of universality called \textit{strict} universality in Ref.~\cite{Aharonov2003univ}. In the context of quantum computation with circuits, the Clifford unitaries supplemented by any non-stabilizer gate, i.e., any gate outside of the Clifford group, constitute an example of a \emph{strictly universal} gate set~\cite{Boykin1999clifford+t}. In this sense, alongside entanglement, non-stabilizerness (colloquially known as magic) is regarded as a necessary resource for enhanced computational power. A common choice for the additional gate is the $T$ gate, $T \coloneqq \mathrm{diag}(1,e^{i\pi/4}).$ Another important gate set known to be strictly universal is the set of all single-qubit rotations together with the \textsc{cnot} gate~\cite{Barenco+1995}.

A weaker notion of universality is CC-universality~\cite{NestDMB2007} (or \textit{computational} universality~\cite{Aharonov2003univ}). Given any unitary operation $U$, a device is said to be CC-universal if it can reproduce the statistics of computational basis measurements in any state $U\ket{0}^{\otimes n}$. In other words, a CC-universal quantum computer can reproduce the output probability distribution of any quantum circuit. Clearly, strict universality implies computational universality, i.e., a strictly universal gate set is also computationally universal, but the reverse is not necessarily true. An example of a computationally universal gate set that is not strictly universal is the set generated by the Toffoli and Hadamard gates~\cite{Aharonov2003univ,shi2002toffoli}.

It is also important to comment on the difference between \emph{exact} and \emph{approximate} universality. Discrete gate sets (such as the Clifford+$T$ and the Toffoli+$H$ gate set) can only approximate arbitrary unitaries or output distributions up to a desired precision. In this sense, they are only approximately (strictly or computationally) universal. In contrast, the set comprised of the \textsc{cnot} together with all single-qubit rotations is exactly (strictly) universal.

For other models of quantum computation, different sets of operations can be similarly used to enable strict or computational, and exact or approximate, universality. This work explores the transformation of Clifford+$T$ quantum circuits into different measurement-based schemes. Hence, our use of the term \emph{universality} always implies \emph{approximate universality}.

\subsection{Magic-state injection}\label{subsec: MSI}

As we have seen, Clifford+$T$ quantum circuits are strictly universal for quantum computation. Here, we briefly describe another universal model of quantum computation known as the magic-state injection model~\cite{BravKit2005}.

In the context of fault-tolerant computation, error-correcting codes play a crucial role in protecting quantum information against errors. Stabilizer codes constitute a particular class of error-correcting codes wherein the encoding of information can always be accomplished using only Clifford gates. In commonly used codes, like the Steane code or the surface code, Clifford gates can be easily implemented; specifically, to apply a Clifford gate $C$ to an encoded qubit, one needs only apply that gate to all corresponding physical qubits. Gates that allow such a simple implementation are said to be transversal. Unfortunately, it is not possible for a universal gate set to have all its gates implemented transversally~\cite {EastinKnill2009}. As such, it comes as no surprise that in the aforementioned codes, the $T$ gate is non-transversal.
\begin{figure}[t]
    \centering
    \begin{tikzpicture}
      \node[scale=.9] {
        \tikzset{
            my label/.append style={above right,xshift=0.13cm,yshift=0.cm}
        }
        \begin{quantikz}[thin lines]
           \lstick{\ket{\psi}} & \qw & \ctrl{1} & \gate{S^m} & \qw \rstick{$T$ \ket{\psi}}\\
           \lstick{\ket{T}} & \qw & \targ{}  & \meter{$m$}\vcw{-1}
        \end{quantikz}      
      };
    \end{tikzpicture}
    \caption{Fault-tolerant implementation of the $T$ gate via the well-known $T$-gadget, using only stabilizer operations and classical feedforward.}
    \label{fig:T_gadget}
\end{figure}
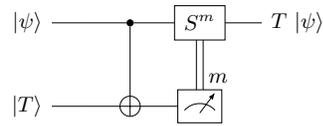

To circumvent this difficulty in fault-tolerantly implementing the $T$ gate, Bravyi and Kitaev~\cite{BravKit2005} proposed a way of producing low-noise copies of the magic state $\ket{T} \coloneqq (\ket{0} + e^{i\pi/4} \ket{1})/\sqrt{2}$ from several noisier copies. Importantly, this magic-state distillation procedure uses only (fault-tolerant) Clifford operations. Once we have access to these low-noise $\ket{T}$ states, we can implement any $T$ gate by the so-called $T$ gadget (Fig.~\ref{fig:T_gadget}), which uses only stabilizer operations and classical feedforward of measurement outcomes.

Using the $T$ gadget, any universally general $n$-qubit Clifford$+T$ quantum circuit with $t$ $T$ gates can be transformed into an $(n+t)$-qubit adaptive Clifford circuit~[Fig.~\ref{fig: Workflow}, step \textcircled{\raisebox{-0.9pt}{1}}], that can be fault-tolerantly implemented using a suitable error-correcting code. The price to pay for this is the need for the (offline) preparation of $t$ copies of the $\ket{T}$ state, as well as feedforward and adaptivity.

\subsection{Measurement-based quantum computation}\label{subsec: MBQC}

Measurement-based quantum computation (MBQC) comes in an impressive variety of flavors. Here, we try to center the discussion around the models that are more directly related to our work.

\subsubsection{One-way quantum computing}\label{subsubsec: 1WC}

In 2001, Raussendorf and Briegel proposed the one-way quantum computer~\cite{RaussBrie2001} where the computation can be separated into two stages. First, the offline preparation of an entangled resource state known as the \emph{cluster} state. Second, the processing stage, wherein the qubits from the cluster state are measured in suitable bases to implement the desired computation. We now focus in a bit more detail on each of these two stages.

Let $\mathcal{G}=\left(V,E\right)$ be an undirected simple graph where $V$ are its vertices and $E$ its edges. For any such graph, a corresponding graph state $\left|\mathcal{G}\right>$ can be constructed as follows. First, a single qubit in state $\left|+\right> \coloneqq (\left|0\right> +\left|1\right> )/\sqrt{2}$ is assigned to each vertex $i\in V$ of $\mathcal{G}$; then, for each edge $e\in E$ connecting two vertices $i$ and $j$, the corresponding qubits are entangled using the (Clifford) gate $CZ \coloneqq\mathrm{diag}(1,1,1,-1)$. 

Graph states are stabilizer states and therefore admit a simple representation using the stabilizer formalism. That is, an $n$-qubit graph state can be uniquely described by writing down the set of $n$ generators of its stabilizer. These can be chosen to have the following form~\cite{HeinEB2004graphstates}: For each qubit $i$, a generator $G_i$ is given by the tensor product of a Pauli operator $X$ acting on the considered qubit and $Z$ operators applied to each of its neighbors:
\begin{equation}\label{eq: generators of a graph state}
    G_i = X_{i}\prod_{j\in{\cal N}(i)}Z_{j}\,,
\end{equation}
where $\mathcal{N}(i)$ denotes the set of qubits neighboring (i.e., connected to) qubit $i$.

The cluster state is no more than a particular kind of graph state where the underlying graph consists of a two-dimensional square grid. After preparing this state, following the prescription described above, the processing step can be carried out according to whichever algorithm we want to implement.

Raussendorf and Briegel showed that single-qubit rotations and the \textsc{cnot} gate can be deterministically implemented in the cluster state by performing suitable sequences of single-qubit measurements in the equator of the Bloch sphere, together with computational basis measurements~\cite{RaussBrie2001}. Since these gates constitute a strictly universal set, this endows the one-way computer with strictly universal capabilities.

Since the seminal work of Ref.~\cite{RaussBrie2001}, other combinations of resource states and sets of single-qubit measurements have been shown to be strictly universal~\cite{Nest2006, BroadbentFK2009blind, Mhalla2012XZuni, Mantri2017xycluster}. Crucially, it is not necessary to allow a continuous range of measurement bases and discrete sets suffice for (approximate) strict universality~\cite{BroadbentFK2009blind, Mantri2017xycluster}. The interested reader is pointed to Ref.~\cite[Table~1]{Takeuchi2019hypergraph} for a quick overview of the different resource states and measurement bases that can be used for universal quantum computation.

One of the greatest advantages of 1WQC is that all multi-qubit operations (i.e., the entangling $CZ$ gates) can be done offline before the processing stage. Since these unitaries are often more prone to errors than single-qubit ones, their isolation in the state preparation stage helps to mitigate the nefarious effects of those errors. Furthermore, if we have a device capable of preparing a universal resource state (of the proper size), we can perform any quantum computation by using the appropriate sequence of measurements in the processing stage. However, the production of a large quantum resource is not without its challenges. For a circuit with $n$ qubits and logical depth $d_L$, the corresponding cluster state needs to have size $O(nd_L)$, which is extremely demanding for near- and intermediate-scale quantum hardware.

One way of circumventing this problem is to explore a peculiar feature of 1WQC: Clifford unitaries can be implemented via non-adaptive measurement patterns consisting exclusively of Pauli measurements. Because of this, all Clifford operations can be performed at once at the very beginning of the computation, regardless of their placement in the corresponding quantum circuit. Thus, 1WQC provides an intrinsically quantum way of parallelizing quantum computations, cutting across the strict temporal ordering of the quantum circuit model. If we consider an $n$-qubit Clifford$+T$ quantum circuit with $t$ non-Clifford $T$ gates, carrying out all the Pauli measurements in the corresponding cluster state will leave us with an $(n+t)$-qubit computation-specific resource state that is equivalent to a graph state up to local Clifford transformations; $t$ of the qubits of this state are what we call the \emph{computational qubits} which need to be (adaptively) measured in rotated angles of $\pm \pi/4$ along the equator of the Bloch sphere, while the remaining $n$ qubits are the so-called \emph{output qubits} which will hold the final state output by the one-way computation~\cite{BroadbentFK2009blind, Mantri2017xycluster}.

Removing all Pauli measurements and determining the computation-specific input state can be performed efficiently on a classical computer~\cite{RaussBrBr2003}. Hence, rather than needing an $O(nd_L)$ (entirely general) cluster state, we require only an $(n+t)$-qubit computation-specific resource state. The problem with this approach is that, while we substantially save on the number of qubits needed, the quantum state to be prepared can have a significantly more intricate connectivity structure, which might be (more) challenging to prepare on actual quantum hardware.
\begin{remark}[Removing output qubits]\label{remark: CC universality is sufficient}
    Because our work focuses mostly on improving PBC and since this model of quantum computation is only computationally universal (see Sec.~\ref{subsubsec: PBC} and Remark~\ref{remark: PBC is CC universal} therein), we can content ourselves with having only a (weaker) CC-universal one-way quantum computer. That is, we are concerned only with simulating the output statistics of a certain quantum circuit. This means that further simplification to the computation-specific input state is possible. Since the output qubits are measured in Pauli bases, their measurements can be classically processed together with the measurements associated with the Clifford gates, leading to a computation-specific input state that has only $t$ qubits, which need to be measured along the $\pm \pi/4$ directions of the equator of the Bloch sphere. As before, this input state is local-Clifford equivalent to a graph state.
\end{remark}

Recently, some results have started to arise concerning the realization of 1WQC using hypergraph states~\cite{MillerMiyake2016hypergraph, Takeuchi2019hypergraph, takeuchi2023catalytic}. These are not stabilizer states; in fact, they possess both entanglement and non-stabilizerness. It has been demonstrated that hypergraph states together with single-qubit Pauli measurements are sufficient for strict universality~\cite{takeuchi2023catalytic}.

\subsubsection{Teleportation and state transfer}\label{subsubsec: teleportation}

While 1WQC is considered a measurement-based model because its processing stage consists solely of single-qubit measurements, it still requires the coherent preparation of the resource state. Alternatively, in 2003, Nielsen proposed a universal scheme that requires no coherent unitary dynamics~\cite{Nielsen2003}; instead, it relies only on the preparation of qubits in the $\ket{0}$ state and projective measurements of up to four qubits. Later, this approach to quantum computing was improved by Leung~\cite{Leung2004} in several ways. Notably, she removed the need for the recursive procedure to deal with the Pauli corrections while showing that a discrete set of at most two-qubit measurements suffices to guarantee universality.
\begin{figure}[t]
    \centering
    \begin{tikzpicture} \node[scale=1.0]{ \tikzset{operator/.append style={fill=gray!15, rounded corners}} \begin{quantikz}[thin lines] \lstick[wires=1]{\ket{\psi}} &\qw &\gate[wires=2]{Z \otimes Z} &\gate{X} &\qw \\ &\gate{X} &\qw &\qw\rstick{$P\ket{\psi}$} \end{quantikz} }; \end{tikzpicture}
    \caption{Implementation of state transfer via single- and two-qubit Pauli measurements, up to a Pauli correction, $P$, that depends on all three measurement outcomes. Grey boxes with rounded edges are used throughout to represent projective measurements; the outcome of each of these measurements is stored in memory and accessible for future use (if needed).}
    \label{fig: state transfer (basic)}
\end{figure}
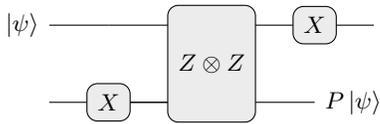
\begin{figure}[t]
    \centering
    \begin{tikzpicture} \node[scale=0.95]{ \tikzset{operator/.append style={fill=gray!15, rounded corners}} \begin{quantikz}[thin lines] \lstick[wires=2]{\ket{\psi}} &\qw &\gate[wires=3]{Z \otimes I \otimes X} &\qw &\qw\rstick[wires=2]{$P\, CX \ket{\psi}$} \\ &\qw &\qw &\gate[wires=2]{X\otimes Z} &\qw \\ &\gate{Z} &\qw &\qw &\gate{X} &\qw \end{quantikz} }; \end{tikzpicture}
    \caption{Implementation of the \textsc{cnot} gate using state transfer; the gate is applied to the (arbitrary) two-qubit input state $\ket{\psi}$ up to a two-qubit Pauli operator $P$ which depends on the outcomes of the measurements. Note that the second measurement involves only the first and last qubits, as it has an identity on the second qubit, that is, it consists of the measurement $Z_1X_3$ (having only weight 2).}
    \label{fig: State transfer CNOT}
\end{figure}
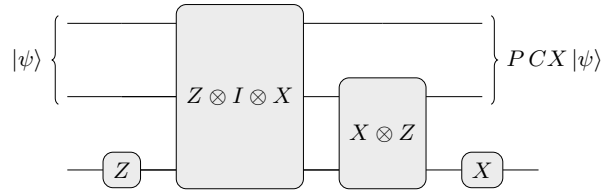
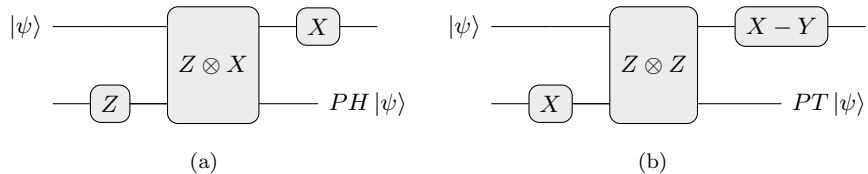
\begin{figure*}
    \centering
    \subfloat[\label{subfig: State transfer: H}]{\begin{tikzpicture} \node[scale=1.0]{ \tikzset{operator/.append style={fill=gray!15, rounded corners}} \begin{quantikz}[thin lines] \lstick{\ket{\psi}} &\qw &\gate[wires=2]{Z \otimes X} &\gate{X} &\qw \\ &\gate{Z} &\qw &\qw\rstick{$P H\ket{\psi}$} \end{quantikz} }; \end{tikzpicture}}\quad\quad\quad\quad  
    \subfloat[\label{subfig: State transfer: T}]{\begin{tikzpicture} \node[scale=1.0]{ \tikzset{operator/.append style={fill=gray!15, rounded corners}} \begin{quantikz}[thin lines] \lstick[wires=1]{\ket{\psi}} &\qw &\gate[wires=2]{Z \otimes Z} &\gate{X-Y} &\qw \\ &\gate{X} &\qw &\qw\rstick{$P T \ket{\psi}$} \end{quantikz} }; \end{tikzpicture}}
    \caption{Implementation of (a) the Hadamard gate and (b) the $T$ gate using state transfer. Note that the unitary transformations are applied to the (arbitrary) input state $\ket{\psi}$ up to a single-qubit Pauli operator $P$ which depends on the measurement outcomes.}
    \label{fig: State transfer: H and T}
\end{figure*}

These two schemes share as an underlying primitive the use of teleportation to implement universal gate sets~\cite{GottChuang1999}. By using state transfer (Fig.~\ref{fig: state transfer (basic)}), Perdrix managed to propose an even simpler scheme, where a single two-qubit measurement is sufficient to realize a universal gate set and, thus, universal quantum computation~\cite{Perdrix2005}.

Here, we succinctly describe the scheme by Perdrix, as this will be directly useful to the work presented in Sec.~\ref{sec: incPBC}. Since the gate set generated by Hadamard, $T$ gate, and \textsc{cnot} is strictly universal for quantum computing, realizing each of these gates via state transfer is sufficient to prove the universality of the model. Figures~\ref{fig: State transfer CNOT} and~\ref{fig: State transfer: H and T} depict how one can apply these unitary transformations (up to a Pauli operator) to an arbitrary quantum state $\ket{\psi}$. Because the implementation is non-deterministic, in that the unitary is implemented up to a correction that depends on the outcomes of the measurements, we need to understand how these corrections can be dealt with. Note that standard state transfer, as depicted in Fig.~\ref{fig: state transfer (basic)}, can be seen as an attempt to implement an identity transformation which will be done, in this scheme, up to a Pauli operation. This is the way to handle the Pauli corrections: Whenever a Pauli operator $P$ arises from implementing a certain unitary, we perform regular state transfer until the measurement outcomes combine in such a way that any extra Pauli factors are canceled out. This resembles the recursive procedure in Nielsen's scheme~\cite{Nielsen2003}. On average, four iterations suffice to correct for one single-qubit Pauli operator $P$; thus, on average, this recursive procedure incurs an overhead that is linear in the total number of gates in the quantum circuit.

Finally, we remark that although at first 1WQC and these teleportation and state-transfer schemes were considered separately, several independent works~\cite{Leung2004unify, Perdrix2005unify, ChildsLeung2005unify} established important connections between them, thus succeeding in presenting a unified view of MBQC.

\subsubsection{Pauli-based computation}\label{subsubsec: PBC}

Compared with other measurement-based models, PBC had a fairly late appearance, being proposed in 2016 in the seminal paper by Bravyi, Smith, and Smolin~\cite{BSS2016}. Possibly owing to this, but also due to practical implementation challenges, PBC remains a relatively understudied model of quantum computation. Here, we provide a brief review of this scheme. For more detailed discussions, the interested reader is referred to Refs.~\cite{BSS2016, YogaJS2019, PeresGa2023}; additionally, Ref.~\cite[Sec. IIA]{Peres2023} provides a state-of-the-art review of recent works exploring PBC in different contexts. 

In a PBC, a separable non-stabilizer input state is prepared offline; the computation is then driven by an adaptive sequence of independent and pairwise commuting multi-qubit Pauli measurements performed on the qubits of the input state. In Ref.~\cite{BSS2016}, the authors showed that any (universally general) Clifford+$T$ quantum circuit with $n$ qubits and $t$ $T$ gates can be simulated by a PBC on $t$ qubits and at most $t$ $t$-qubit Pauli measurements. The ability to simulate any quantum circuit makes PBC a universal model for quantum computation.

The proof of universality goes as follows. First, take the (non-adaptive) $n$-qubit Clifford+$T$ circuit with $t$ $T$ gates and transform it into the magic-state injection model (Sec.~\ref{subsec: MSI}); this means we are left with an adaptive Clifford circuit with input $\ket{0}^{\otimes n}\ket{T}^{\otimes t}$. Note that the stabilizer register is stabilized by $\mathcal{L}= \{ Z_1, Z_2,\dots, Z_n \}$. Since all operations are now Clifford unitaries, we can efficiently back-propagate every measurement to the beginning of the circuit~\cite{Gottesman1998}. Due to the adaptive nature of the circuit, the measurements need to be dealt with in the appropriate order. The intermediate measurements stemming from the gadgets need to be handled first (starting with the one from the first gadget and working our way successively until the last) and only afterward can the final readout measurements be pushed to the beginning of the circuit.

\begin{figure*}[t]
    \centering
    \begin{tikzpicture} \node[scale=0.9]{ \begin{quantikz}[thin lines] \lstick[wires=4]{\ket{\Psi_{\mathrm{in}}}} &\gate[wires=4,style={fill=gray!15, rounded corners}]{X\otimes I\otimes Y\otimes Z} &\qw\rstick[wires=4]{\ket{\Psi_{\mathrm{out}}}} \\ &\qw &\qw \\ &\qw &\qw \\ &\qw &\qw \end{quantikz} $\rightarrow$ \begin{quantikz}[thin lines] \lstick[wires=4]{\ket{\Psi_{\mathrm{in}}}} &\qw\gategroup[wires=5,steps=6,style={rounded corners,fill=gray!15, inner xsep=2pt},background]{{$P = X\otimes I\otimes Y\otimes Z$}} &\targ{} &\qw &\qw &\qw &\qw&\qw\rstick[wires=4]{\ket{\Psi_{\mathrm{out}}}} \\ &\qw &\qw &\qw &\qw &\qw &\qw&\qw \\ &\gate{S^{\dagger}} &\qw &\targ{} &\qw &\gate{S} &\qw&\qw \\ &\gate{H} &\qw &\qw &\targ{} &\gate{H} &\qw&\qw \\ \lstick{$\ket{0}_{\mathrm{aux}}$} &\gate{H} &\ctrl{-4} &\ctrl{-2} &\ctrl{-1} &\gate{H} &\meter{$\sigma_P$} \end{quantikz} }; \end{tikzpicture}
    \caption{Illustration of a quantum circuit implementation of the Pauli measurement $P=X_1Y_3Z_4$. This simple example demonstrates the close connection between the weight of the Pauli operator to be measured and the number of \textsc{cnot} gates in the corresponding circuit implementation: The Pauli has weight 3, thus three \textsc{cnot} gates are needed to implement the measurement using this scheme.}
    \label{fig: Translation PBC -> circuits}
\end{figure*}

Once a Pauli measurement, $P$, arrives at the beginning of the circuit, it may fall into one of three categories:
\begin{itemize}
    \item[(i)] \textit{$P$ anti-commutes with (at least) one of the operators, $Q$, in $\mathcal{L}$.}---In this case, the outcome of the corresponding Pauli measurement, $\sigma_P\in \{0,1\}$, can be decided classically using a coin toss; $P$ is then dropped from the quantum circuit and replaced by the Clifford unitary:
    \begin{equation}\label{eq: V unitaries}
        V(\sigma_P,\sigma_Q) = \frac{(-1)^{\sigma_P} P +  (-1)^{\sigma_Q}Q}{\sqrt{2}},
    \end{equation}
    where $\sigma_Q$ denotes the outcome associated with the (prior) measurement of $Q.$ The outcome $\sigma_P$ is stored in a list containing all the outcomes. 
    \item[(ii)] \textit{$P$ commutes will all operators in $\mathcal{L}$ and depends on a subset of them.}---In this case, the Pauli measurement can again be dealt with classically, as its outcome, $\sigma_P$, can be efficiently inferred from the outcomes of the Pauli operators it depends on. The outcome obtained is stored in the list with all the other outcomes. 
    \item[(iii)] \textit{$P$ commutes with all operators in $\mathcal{L}$ and is independent of them.}---This is when $P$ needs to be measured in the quantum computer. Note that case~(i) ensures that any such $P$ must act trivially on the stabilizer register so that we can measure only its $t$-qubit non-stabilizer-register component. The Pauli operator $P$ is then added to $\mathcal{L}$ and its outcome is saved in the list with all the outcomes.
\end{itemize}

The procedure described above ensures that the original quantum circuit is simulated by an adaptive sequence of independent and compatible Pauli measurements performed only on the $t$-qubit non-stabilizer register, i.e., the circuit is simulated by a PBC. Since there are at most $t$ independent and pairwise commuting Pauli operators of $t$ qubits, the number of measurements will be at most $t$. Potentially, these operators could all be of weight $t$, although our numerical results in Ref.~\cite{PeresGa2023} indicate that often the weight is lower than this trivial upper bound. In Section~\ref{subsec: Improved weights}, we prove that better weights can be achieved by carrying out a pre-compilation step before proceeding to the PBC framework. More specifically, we take the input quantum circuit and transform it into a 1WQC procedure with an associated $t$-qubit computation-specific input state and adaptive measurement pattern, and use the latter as the starting point for performing PBC.
\begin{remark}[Universality of PBC]\label{remark: PBC is CC universal}
    The attentive reader will note that PBC simulates the original quantum circuit by producing samples drawn from the same output distribution. Therefore, Pauli-based quantum computers are devoid of strict universality and are instead computationally universal. This is in striking contrast to the other measurement-based models presented thus far, which can prepare the same output state as the corresponding quantum circuit. Put differently, PBC does not concern itself with state preparation, but rather with simulating the output statistics of the corresponding quantum circuit.
\end{remark}
\begin{figure*}[t]
    \centering
    \def\myvdots{\ \vdots\ } \def\myddots{\ \ddots\ } \begin{tikzpicture} \node[scale=0.8] { \tikzset{ my label/.append style={above right,xshift=0.35cm,yshift=-0.25cm} } \begin{quantikz}[thin lines] \lstick[wires=4]{$\ket{\mathcal{G}}$} \qw &\gate{S^{\dagger}} &\gate{T} &\gate{H} &\meter{$s_1$} &\cwbend{1} &\cw &\cw &\cw &\cw &\cwbend{1}\\ \qw &\qw &\qw &\qw &\qw &\gate{S^{s_1}} &\gate{S^{\dagger}} &\gate{T} &\gate{H} &\meter{$s_2$} &\cwbend{1}\\ \myvdots & & & & & & &\myddots & & &\myvdots \\ \qw &\qw &\qw &\qw &\qw &\qw &\qw &\qw &\qw &\qw &\gate{S^{f_t}} &\gate{S^{\dagger}} &\gate{T} &\gate{H} &\meter{$s_t$} \end{quantikz} }; \end{tikzpicture}
    \caption{Circuit corresponding to a given one-way quantum computation on a $t$-qubit graph state $\ket{\mathcal{G}}.$ The fact that the measurements can be restricted to the bases along the angles of $\pm \pi/4$ on the equator of the Bloch sphere means that we do not need arbitrary $Z$-rotation gates and that the $T$ gate is sufficient. The measurement outcomes $s_i$ in a given layer influence the measurement bases in subsequent layers. This dependence is encoded in the Boolean functions $f_j$ and depicted by the classical wires seen in the picture. The value of each $f_j$ depends on the set of outcomes $s_i$ (with $i<j$) that influence the measurement basis of the $j$th computational qubit and can be efficiently calculated in an assisting classical machine.}
    \label{fig: 1WQC-circuit}
\end{figure*}

In our work in Ref.~\cite{PeresGa2023}, we demonstrated that PBC can be regarded as a circuit compilation tool that allows us to trade (affordable) classical computation for (expensive) quantum resources. Specifically, we numerically demonstrated how, for certain families of quantum circuits dominated by Clifford gates, PBC can be leveraged to obtain quantum circuits that often require fewer gate counts and/or depth than the original ones. These results were achieved by providing various circuit implementations of Pauli measurements, which allow a translation from PBC back to quantum circuits. An example of one such scheme can be seen in Fig.~\ref{fig: Translation PBC -> circuits}. It clearly demonstrates the intimate connection between the weight of the Pauli measurements and the \textsc{cnot} complexity of the PBC-compiled circuits: one-to-one in this particular case. That is, the number of \textsc{cnot} gates exactly corresponds to the weight of the measured Pauli operator. Hence, any reduction of the weight of these operators with respect to our prior work in~\cite{PeresGa2023} translates directly into an improvement of the results presented therein. Ref.~\cite[Secs.~3.1.2 and 3.1.3]{PeresGa2023} presents alternative circuit translations. In all of them, there is a tight relation between the weight of the Pauli measurements and the total number of \textsc{cnot} gates in the final PBC-compiled quantum circuits.
\begin{figure*}[t]
    \centering
    \def\myvdots{\ \vdots\ } \def\myddots{\ \ddots\ } \begin{tikzpicture} \node[scale=0.8] { \tikzset{ my label/.append style={above right,xshift=0.35cm,yshift=-0.25cm} } \begin{quantikz}[thin lines] \lstick[wires=4]{$\ket{\mathcal{G}}$} \qw &\qw &\gate{S^{\dagger}} &\ctrl{4} &\gate{S^{m_1}} &\gate{H} &\meter{$s_1$} &\cwbend{1} &\cw &\cw &\cw &\cw &\cwbend{1} \\ \qw &\qw &\qw &\qw &\qw &\qw &\qw &\gate{(S^{\dagger})^{s_1\oplus 1}} &\ctrl{4} &\gate{S^{m_2}} &\gate{H} &\meter{$s_2$} &\cwbend{1}\\ \myvdots & & & & & & & & & & &\myddots &\myvdots \\ \qw &\qw &\qw &\qw &\qw &\qw &\qw &\qw &\qw &\qw &\qw &\qw &\gate{(S^{\dagger})^{f_t \oplus 1}} &\ctrl{4} &\gate{S^{m_t}} &\gate{H} &\meter{$s_t$}\\ \lstick[wires=4]{$\ket{T}^{\otimes t}$} \qw &\qw &\qw &\targ{} &\meter{$m_1$}\vcw{-4}\\ \qw &\qw &\qw &\qw &\qw &\qw &\qw &\qw &\targ{} &\meter{$m_2$}\vcw{-4}\\ \myvdots & & & & & & & & &\myvdots & &\myddots\\ \qw &\qw &\qw &\qw &\qw &\qw &\qw &\qw &\qw &\qw &\qw &\qw &\qw &\targ{} &\meter{$m_t$}\vcw{-4} \end{quantikz} }; \end{tikzpicture}
    \caption{Transformation of the circuit in Fig.~\ref{fig: 1WQC-circuit} into an adaptive Clifford circuit by replacing each $T$ gate with the $T$ gadget depicted in Fig.~\ref{fig:T_gadget}. Note that the outcomes of measurements of computational qubits are denoted by $s_i$, while for the outcomes of gadget measurements, $m_i$ is used.}
    \label{fig: 1WQC-circuit after T gadget}
\end{figure*}

\section{Pre-compilation as a way to improve Pauli-based computation }\label{sec: Improvements I - connection to 1WC}

In the previous section, we described PBC and noted that its computational steps are (at most) $t$ independent and pairwise commuting Pauli measurements that can potentially involve all of the $t$ qubits of the system. On the other hand, we have also seen that constant-weight projective measurements are sufficient for universal quantum computation~\cite{Nielsen2003, Leung2004, Perdrix2005}. The obvious follow-up questions are: Can we find alternative formulations of PBC with improved weights? Are constant weights sufficient for PBC? Here, we provide partial answers to these questions. Additionally, in line with the overarching goal of improving the feasibility of PBC, we supply important new results on how to reduce the depth of quantum computations carried out within this model.

\subsection{The pre-compilation step}\label{subsec: The pre-compilation step}

Throughout this entire section, we consider that we want to simulate a universal, non-adaptive quantum circuit $\mathcal{U}$ acting on $n$ qubits and with gates drawn from the Clifford+$T$ set. The circuit has logical depth $d_L,$ $t$ $T$ gates, and $w$ readout computational basis measurements. We have seen that, if we are concerned only with computational universality, $\mathcal{U}$ can be simulated by a PBC requiring the separable input state $\ket{T}^{\otimes t}$ and an adaptive sequence of at most $t$ $t$-qubit independent and compatible Pauli measurements.

However, $\mathcal{U}$ can also be simulated by a one-way computation involving a $t$-qubit computation-specific resource state, $\ket{\mathcal{R}}$, and adaptive single-qubit measurements along the $\pm \pi/4$ directions in the equator of the Bloch sphere. The state $\ket{\mathcal{R}}$ is local-Clifford equivalent to a graph state $\ket{\mathcal{G}}$ with stabilizers $\{G_i\}_{i=1}^t$ described by Eq.~\eqref{eq: generators of a graph state}. Our approach is to leverage this observation and use 1WQC as a stepping stone (or, put differently, as a pre-compilation step) before finding the corresponding PBC (recall Fig.~\ref{fig: Workflow}).

We consider the specific case where the input state for the one-way computation is a graph state $\ket{\mathcal{G}}$ (see Sec.~\ref{subsec: Generality} for comments on generalizability). The computation is then driven by a sequence of measurements that are broken into layers, with outcomes from one layer determining the measurement bases in subsequent layers. We can represent this procedure in the form of a circuit as depicted in Fig.~\ref{fig: 1WQC-circuit}. Importantly, we note that we are exploring the fact that a discrete set of measurement bases is sufficient for universality~\cite{BroadbentFK2009blind}, notably the measurement bases along the $\pm \pi/4$ directions on the equator of the Bloch sphere suffice. We are assuming that any Pauli measurements, including the readout measurements, have already been removed, originating the computation-specific graph state $\ket{\mathcal{G}}$ (recall Remark~\ref{remark: CC universality is sufficient}). Additionally, we note that the labeling of the qubits is such that no measurement $M_i$ depends on the outcome of measurement $M_j$ with $j>i$. This means there is a time-ordering to the measurement pattern so that: $M_i \prec M_j \Longrightarrow i < j$. This ordering is assumed throughout, including in Figs.~\ref{fig: 1WQC-circuit} and \ref{fig: 1WQC-circuit after T gadget}, where the potential classical influence of a measurement outcome on subsequent measurements is depicted by the classical wires.

To translate the one-way computation depicted in Fig.~\ref{fig: 1WQC-circuit} into a PBC, we transform each $T$ gate into a $T$ gadget, originating the adaptive Clifford circuit, $\mathcal{C}$, shown in Fig.~\ref{fig: 1WQC-circuit after T gadget}. This is the starting point to the proofs of all of our main results.

In the remainder of the paper, we use the following notation: Measurements associated with the computational qubits (i.e., the qubits of the graph state) are referred to as ``\emph{computational measurements}'' and denoted CM, while measurements performed on the qubits of the auxiliary register $\ket{T}^{\otimes t}$ are termed ``\emph{gadget measurements}'' and denoted GM. We choose to use different letters to denote the measurement outcomes of computational qubits, $\{s_i\}_{i=1}^t$, and those of the auxiliary qubits introduced by the $T$ gadget, $\{m_i\}_{i=1}^t$, (cf. Fig.~\ref{fig: 1WQC-circuit after T gadget}). Additionally, we also differentiate the Pauli operators that stem from these two types of measurements; Pauli operators resulting from measurements on computational qubits are denoted by $P_i$, while, for those originating from gadget measurements, we use $Q_i$. The generators of the stabilizer of the graph state, $\mathcal{S}$, are represented by $G_i$ and constructed as prescribed in Eq.~\eqref{eq: generators of a graph state}.

\subsection{Overview of the results}\label{subsec: Overview of results and techniques}
\begin{table*}
    \begin{tabular}{>{\centering\arraybackslash}p{4.0cm} >{\centering\arraybackslash}p{9.5cm} >{\centering\arraybackslash}p{2.8cm}}
        \hline\hline
        Improvement & Processing order & Formal statement \\
        \hline
        Better weights & $\mathcal{O}_1:\,\, \mathrm{GM}_1 \prec \mathrm{CM}_1 \prec \mathrm{GM}_2 \prec \mathrm{CM}_2 \prec\dots\prec \mathrm{GM}_t \prec \mathrm{CM}_t$ & Theorem~\ref{theorem: improved weights} \\
        Better depth & $\mathcal{O}_2:\,\, \mathrm{GMs} \prec \mathrm{CM}_{\ell_1} \prec \mathrm{CM}_{\ell_2} \prec \dots\prec \mathrm{CM}_{\ell_{d_{\mathrm{1W}}}}$ & Thereom~\ref{theorem: improved depth} \\
        Weight-depth trade-off & $\mathcal{O}_3:\,\, \mathrm{GM}_{\ell_1} \prec \mathrm{CM}_{\ell_1} \prec \mathrm{GM}_{\ell_2} \prec \mathrm{CM}_{\ell_2} \prec\dots\prec \mathrm{GM}_{\ell_{d_{\mathrm{1W}}}} \prec \mathrm{CM}_{\ell_{d_{\mathrm{1W}}}}$ & Theorem~\ref{theorem: improved depth and weights} \\
        \hline\hline
    \end{tabular}
    \caption{Main theoretical results of this section obtained by choosing to back-propagate the Pauli measurements in different ways. As explained in the text, $\mathrm{GM}_i$ ($\mathrm{CM}_i$) is used to denote the gadget (computational) measurement performed on the $i$th auxiliary (computational) qubit, while $\mathrm{CM}_{\ell_i}$ denotes the set of all measurements of computational qubits belonging to layer $\ell_i$ and $\mathrm{GM}_{\ell_i}$ the set of corresponding gadget measurements. The notation $A\prec B$ is used to indicate that $A$ precedes $B$.}
    \label{tab: Summary of main results}
\end{table*}

Before delving into the detailed and formal statement of our results, we would like to provide a high-level picture and punctuate a key observation underlying our three main theorems. If the starting point to the PBC procedure is the quantum circuit $\mathcal{U}$, as described in Sec.~\ref{subsubsec: PBC}, one processes a total of $(t+w)$ measurements by first back-propagating the $t$ gadget measurements and only afterward the $w$ final measurements. On the other hand, by starting from the corresponding 1WQC procedure, we need to deal with a total of $2t$ Pauli operators, $t$ gadget measurements on auxiliary qubits, and $t$ computational measurements performed on the corresponding computational qubits. Importantly, there is some freedom in the way we handle these measurements. One option is to take a processing order wherein we back-propagate first the measurement associated with the first gadget and immediately afterward the measurement of the corresponding computational qubit; we proceed in the same manner for the remaining qubits, always starting with the gadget measurement followed by the associated computational measurement. We can write this explicitly as 
\begin{equation}\label{eq: ordering for the weights}
        \mathcal{O}_1:\quad \mathrm{GM}_1 \prec \mathrm{CM}_1 \prec \dots\prec \mathrm{GM}_t \prec \mathrm{CM}_t\,,
\end{equation}
where $\mathrm{GM}_i$ ($\mathrm{CM}_i$) denotes the gadget (computational) measurement performed on the $i$th auxiliary (computational) qubit and the notation $A\prec B$ is used to indicate that $A$ precedes $B$. This scheduling gives rise to new upper bounds for the weight of the successive Pauli measurements, as stated in Theorem~\ref{theorem: improved weights}. Summarily, the theorem asserts that the weight of the Pauli measurements in the PBC sequence increases monotonically with the number of operators that were already processed. To provide some intuition, this result arises from the structure of the causal dependencies, which gradually extend through the auxiliary (magic) qubits as successive measurements are back-propagated.

Alternatively, we can consider the measurement pattern underlying the one-way computation, where the computation is usually broken into layers with measurements in the same layer, $\ell_i$, being performed simultaneously and influencing the bases of measurements in subsequent layers $\ell_j$ with $j>i$. This gives us the option to propagate all gadget measurements first, followed by the layered propagation of the measurements on the computational qubits:
\begin{equation}\label{eq: ordering for another depth result -> GMs -> CMs}
    \mathcal{O}_2:\,\, \mathrm{GMs} \prec \mathrm{CM}_{\ell_1} \prec \mathrm{CM}_{\ell_2} \prec \dots\prec \mathrm{CM}_{\ell_{d_{\mathrm{1W}}}}\,.
    \end{equation}
Here, $\mathrm{GMs}$ denotes all $t$ gadget measurements, $\mathrm{CM}_{\ell_i}$ the set of all measurements of computational qubits belonging to layer $\ell_i\,$, and $d_{\mathrm{1W}}$ the depth of the underlying one-way computation. This processing order is used to upper bound the overall depth of the PBC by $d_{\mathrm{1W}}$ (see formal statement in Theorem~\ref{theorem: improved depth}). Even without further technical details, this result seems natural, as it rests directly on the structure of the underlying one-way quantum computation.

The fact that two different processing orders are used to arrive at Theorems~\ref{theorem: improved weights} and~\ref{theorem: improved depth} means that they do \emph{not} hold simultaneously, so we can either guarantee a reduced weight \emph{or} a reduced depth, respectively. This led us to wonder whether a trade-off could be achieved to guarantee concurrent improvements of both of these parameters. This is the essence of Theorem~\ref{theorem: improved depth and weights}, which is achieved through the layered propagation of both gadget and computational measurements so that:
\begin{equation}\label{eq: ordering for the depth}
    \mathcal{O}_3:\,\, \mathrm{GM}_{\ell_1} \prec \mathrm{CM}_{\ell_1} \prec \dots\prec \mathrm{GM}_{\ell_{d_{\mathrm{1W}}}} \prec \mathrm{CM}_{\ell_{d_{\mathrm{1W}}}}\,.
\end{equation}
As before, $\mathrm{CM}_{\ell_i}$ denotes the set of all measurements of computational qubits belonging to layer $\ell_i$, while $\mathrm{GM}_{\ell_i}$ refers to the set of corresponding gadget measurements. 

Table~\ref{tab: Summary of main results} summarizes our main results and the associated processing orders. This is meant to provide a global view of our results and be used as a reference to guide the reader through the upcoming subsections. The fine details underlying both the discussion above and Table~\ref{tab: Summary of main results} are better understood by looking at the proofs of our main results (found in Appendices~\ref{app: Proof of the weight theorem} and \ref{app: Proof of depth and depth-weight trade-off}).

\subsection{Improved weights}\label{subsec: Improved weights}

Looking carefully at the PBC procedure (Sec.~\ref{subsubsec: PBC}), we recognize that there are two mechanisms that may cause the Pauli measurements being back-propagated to the beginning of the circuit to spread (unboundedly) across the qubits of the system. The first is the \textsc{cnot} gates, which may increase (or decrease) the weight of Pauli operators by 1. The second is the Clifford unitaries $V(\sigma_P,\sigma_Q)$ that need to be introduced into the quantum circuit whenever an anti-commuting Pauli is detected in (and removed from) the measurement sequence. The latter may increase (or decrease) the weight of Pauli measurements propagated across them by an unconstrained amount. Because of this, when starting from quantum circuits, it is hard to bound the weights of the $t$ independent and pairwise commuting Pauli operators that need to be measured in the quantum computer. Therefore, the upper bound known thus far for the average weight of the measurements in a PBC, $\bar{w}$, consists only of the trivial one~\cite{BSS2016, PeresGa2023}. Specifically, since each Pauli operator may have weight $t$, the average weight is upper-bounded by $t$: $\bar{w} \leq t$.

Here, we show that starting from 1WQC allows us to establish better (non-trivial) upper bounds for the weights of the Pauli measurements in the corresponding PBC.
\begin{theorem}[Improved weights]\label{theorem: improved weights}
    Consider a one-way computation to be carried out on a $t$-qubit, computation-specific graph state $\ket{\mathcal{G}}$ with a measurement pattern requiring only measurements along the $\pm\pi/4$ directions on the equator of the Bloch sphere. By taking on the processing order $\mathcal{O}_1$ defined in Eq.~\eqref{eq: ordering for the weights}, the (magic-register) weights of the $2t$ Pauli operators in the (complete) PBC procedure are upper-bounded by \{1,\,1,\,2,\,2,\dots, t-1,\,t-1,\,t,\,t\}.
\end{theorem}
We defer the proof of this result to Appendix~\ref{app: Proof of the weight theorem} and use the remainder of the section to illustrate and discuss some of its practical consequences. 
\begin{figure*}[t]
    \centering
    \subfloat[\label{subfig: RQCs Fam.1}]{\includegraphics[width=8.cm]{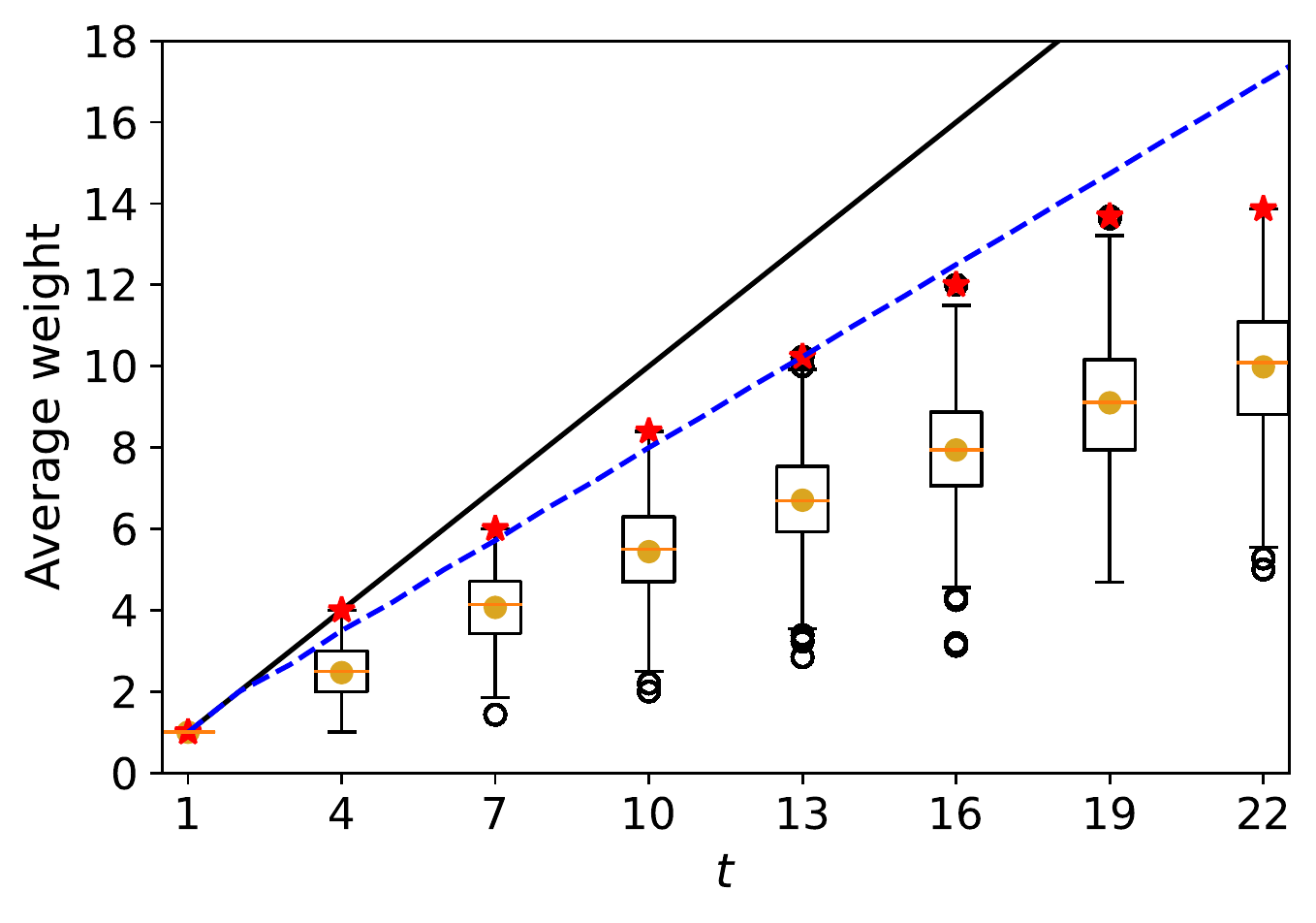}}\hspace{1cm}
    \subfloat[\label{subfig: RQCs Fam.2}]{\includegraphics[width=8.cm]{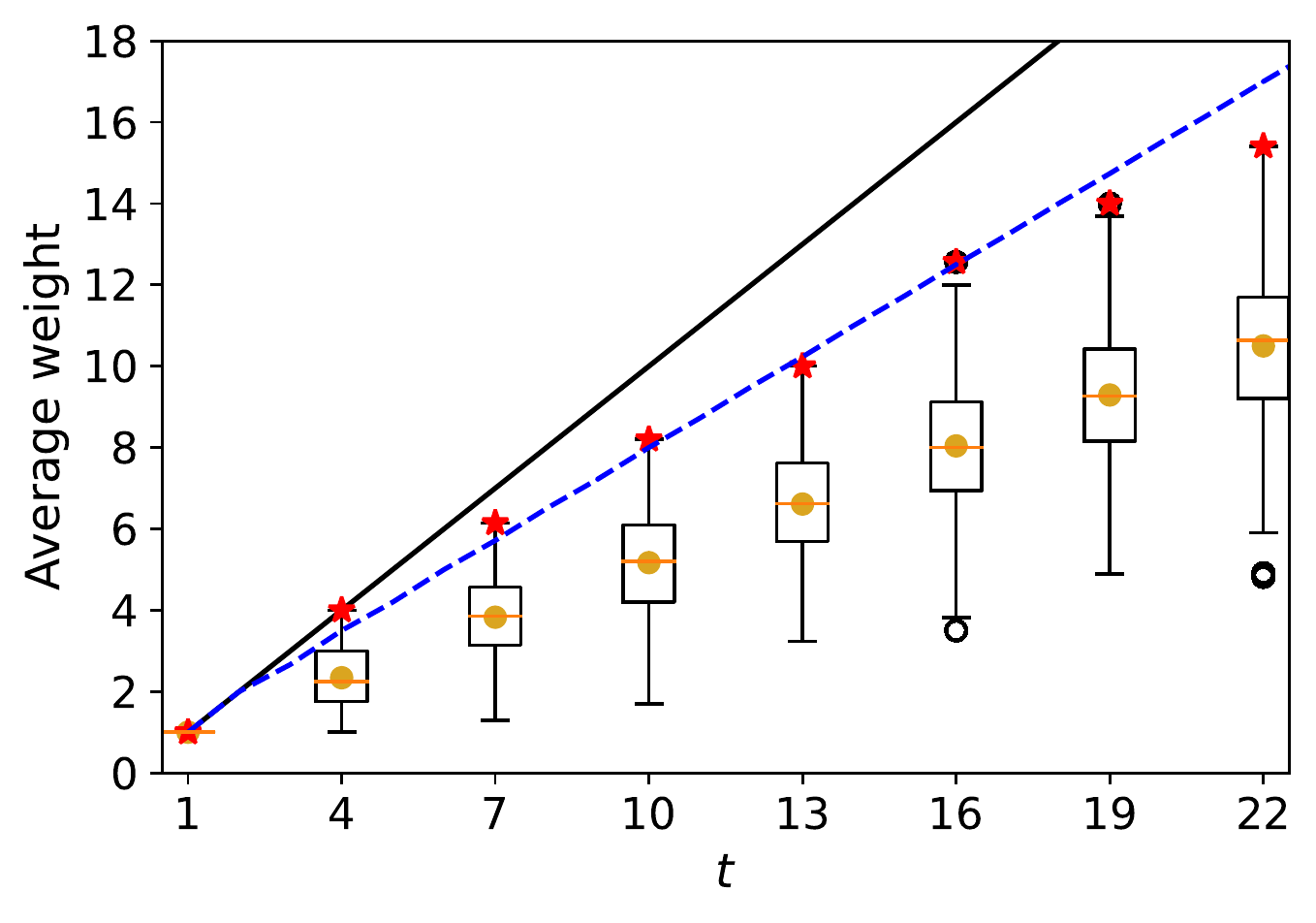}}
    \caption{Numerical data resulting from the compilation of two distinct families of random quantum circuits using our code written for Ref.~\cite{PeresGa2023} and openly available in Ref.~\cite{GitRepoPBC}. (a) Circuits generated as in Ref.~\cite{PeresGa2023} with a well-defined and ordered entanglement structure, and (b) circuits with arbitrary entanglement structure wherein gates are drawn at random from the gate set $\{H,\,S,\,CX,\,T\}.$ The numerical results are depicted using boxplots; the box extends from the first quartile ($q_1$) to the third ($q_3$) and its width, ($q_3 - q_1$), is known as the interquartile range. The orange line inside the box represents the median; the arms extending below and above the box have a length given by 1.5 times the interquartile range; any outliers are depicted as empty black-lined circles. The red stars signal the maximum average weight obtained for a PBC for each $t$, while the yellow circles identify its mean. The solid black line denotes the trivial average weight upper bound of $t$ known prior to this work and the dashed blue line gives the new upper bound for this quantity (as stated in Corollary~\ref{corollary: maximum weight upper bound}).}
    \label{fig: Numerics RQCs}
\end{figure*}

We know that the number of independent and pairwise commuting Pauli operators on $t$ qubits is $t$. Therefore, of the list of $2t$ Paulis, at most $t$ of those correspond to measurements that will have to be performed in the quantum computer. Assuming the worst-case scenario where the Pauli operators to be measured are those with larger weight, i.e., the last $t$ Paulis in the sequence, leads straightforwardly to the following corollary.
\begin{corollary}[Average weight upper bound]\label{corollary: maximum weight upper bound}
    The average weight of the Pauli operators that need to be measured in the quantum hardware is upper bounded by $\bar{w} \leq 3t/4 + 1/2.$
\end{corollary}

Corollary~\ref{corollary: maximum weight upper bound} is easy to prove, following straightforwardly from Theorem~\ref{theorem: improved weights}. It yields an improvement of roughly 25\% over the trivial upper bound previously known for PBC.

We took two distinct families of random quantum circuits (RQCs) as testbeds for our results and their usefulness. The first family is the one we used in Ref.~\cite{PeresGa2023} wherein the circuits have a specific entangling structure implemented by sequences of eight different entanglement layers. (For details, see Sec.~4.2.2 of~\cite{PeresGa2023}.) The second family is less structured in that gates were drawn from the set $\{H,\, S,\, CX,\, T\}$ with probabilities $\{(1-p)/3,(1-p)/3,(1-p)/3,p\}$, without enforcing any specific pattern to the entangling gates. The probability, $p$, of drawing $T$ was adjusted to facilitate creating quantum circuits with the desired $T$ count; the remaining gates were equiprobably drawn. Just like in Ref.~\cite{PeresGa2023}, we consider circuits with $25$ qubits and an adjustable number of $T$ gates $t=\{1,\,4,\,7,\,10,\,13,\,16,\,19,\,22\}.$ For each family and each $T$ count, we generated 500 different RQCs that were simulated by PBC with a total of 1024 shots. We used our code, companion to Ref.~\cite{PeresGa2023} and openly available in Ref.~\cite{GitRepoPBC}, selecting the dummy simulator option, where measurement outcomes are drawn from a uniform distribution rather than from the actual hard-to-simulate distribution. This choice is strongly supported by our previous work~\cite{PeresGa2023}.

Our results are depicted in Figs.~\ref{subfig: RQCs Fam.1} and \ref{subfig: RQCs Fam.2} respectively for the first and second families. Given the large amount of data, we chose to use a boxplot~\cite{Dutoit1986statistics}; the box extends from the first ($q_1$) to the third ($q_3$) quartile with the orange line inside identifying the median. The arms of the plot (extending from the box) have a length, $l$,  of 1.5 times the interquartile range: $l = 1.5(q_3 - q_1)$; any outliers are depicted as empty, black-lined circles. Red stars and yellow circles identify, respectively, the maximum and the mean average weight obtained for each $t$. 

We note that the numerical results associated with the two distinct families are very similar. We also see that while all data points comply with the trivial upper bounds (as must be), several instances violate the average weight bound set by Corollary~\ref{corollary: maximum weight upper bound}. In particular, we see this for both families at $t$ equals 4, 7, and 10, but also at $t=16$ for the second family of RQCs. This happens because the code carries out PBC in the usual way, taking as a starting point the non-adaptive Clifford$+T$ quantum circuits. On the other hand, the upper bound set by Corollary~\ref{corollary: maximum weight upper bound} is obtained assuming the intermediate step of transforming the circuit into a one-way computation and the subsequent suitable processing of the measurements as described by Theorem~\ref{theorem: improved weights}. Hence, the violation of this upper bound indicates that taking 1WQC as an intermediate step serves not only the theoretical purpose of finding better upper bounds but also provides practical advantages in achieving reduced weights. In this sense, 1WQC can be regarded as a pre-compilation step that one might be interested in performing before running the actual PBC procedure.

\subsection{Improved parallelizability}\label{subsec: Improved depth}

Theorem~\ref{theorem: improved weights} (or, more precisely, its proof) disregards the fact that the depth of a $t$-qubit one-way computation is often smaller than $t$, since some of the qubits can be measured simultaneously. Thus, one may now wonder: How does this knowledge influence the corresponding PBC procedure? Can some measurements in PBC also be performed simultaneously? If we were to discover this to be the case, it would translate into important improvements to the overall computational depth.

Consider a PBC on $t$ qubits involving the measurement of $r\leq t$ (adaptively chosen) Pauli operators. Since the PBC procedure ensures that all of these operators are compatible, it is the adaptive structure of the PBC that determines whether some of these measurements may or may not be performed simultaneously. When starting from the quantum circuit model, it is hard to extract any information that allows us to guarantee a depth better than $d=r.$ The adaptive structure of the 1WQC measurement pattern is determined solely by the underlying graph and the so-called generalized flow (or gflow) conditions~\cite{DanosKashefi2006, BrowneKashefi+2007}; in general, it ensures that the computation is broken down into layers of measurements that can be performed simultaneously. When taking 1WQC as a starting point, we expect to mirror this property to PBC; that is, we anticipate that some of the measurements of the PBC can similarly be grouped and performed concomitantly as they do not influence one another, constituting a single layer of the PBC. To the total number of layers formed in this way, we call ``\textit{depth of the PBC,}'' denoted $d_{\mathrm{PBC}}$.

The main results of this section are the following two theorems.
\begin{theorem}[Improved depth]\label{theorem: improved depth}
    Consider a one-way computation to be carried out on a $t$-qubit, computation-specific graph state $\ket{\mathcal{G}}$ with a measurement pattern requiring only measurements along the $\pm\pi/4$ directions on the equator of the Bloch sphere. By taking up the processing order $\mathcal{O}_2$ in Eq.~\eqref{eq: ordering for another depth result -> GMs -> CMs}, the depth of PBC coincides with the depth of the corresponding one-way quantum computation, $d_{\mathrm{1W}}$.
\end{theorem}

This result illustrates how we can reduce the depth of PBC from $t$ to $d_{\mathrm{1W}}$. As is apparent from the proof of the theorem, this comes at the expense of losing the better weight upper bounds derived in the previous section. Hence, we wondered if an intermediate result could be provided so that the depth was still better than $t$ while, at the same time, allowing for an upper bound on the weights better than the trivial one. Such a scenario is subsumed in Theorem~\ref{theorem: improved depth and weights}. Recall that we denote by $\mathrm{CM}_{\ell_i}$ the set of all measurements of computational qubits belonging to layer $\ell_i$ and $\mathrm{GM}_{\ell_i}$ the set of corresponding gadget measurements.
\begin{theorem}[Weight-depth trade-off]\label{theorem: improved depth and weights}
    Consider a one-way computation with logical depth $d_{\mathrm{1W}}$ and layering so that the number of computational qubits in layer $\ell_i$ is $\kappa_i$: $\sum_{i=1}^{d_{\mathrm{1W}}} \kappa_i = t.$ The computation is to be carried out on a $t$-qubit, computation-specific graph state $\ket{\mathcal{G}}$ with a measurement pattern requiring only measurements along the $\pm\pi/4$ directions on the equator of the Bloch sphere.
    By back-propagating the measurements following $\mathcal{O}_3$ in Eq.~\eqref{eq: ordering for the depth}, the depth of the corresponding PBC is upper-bounded by $\min \{2d_{\mathrm{1W}} - 1,t\}.$ Moreover, the weight of the $2\kappa_i$ Pauli operators stemming from $\mathrm{GM}_{\ell_i}$ and $\mathrm{CM}_{\ell_i}$ measurements is upper bounded by $\sum_{j=1}^{i} \kappa_j\,.$ 
\end{theorem}

\begin{figure}
    \centering
    \def\myvdots{\ \vdots\ } \begin{tikzpicture} \node[scale=0.9] { \tikzset{ my label/.append style={above right,xshift=0.35cm,yshift=-0.25cm} } \begin{quantikz}[thin lines] \lstick[wires=4]{$\ket{\mathcal{G}}$} \qw &\gate{S^{\dagger}} &\gate{T} &\gate{H} &\meter{$s_1$}\\ \qw &\gate{S^{\dagger}} &\gate{T} &\gate{H} &\meter{$s_2$}\\ \myvdots & &\myvdots & &\myvdots\\ \qw &\gate{S^{\dagger}} &\gate{T} &\gate{H} &\meter{$s_t$} \end{quantikz} }; \end{tikzpicture}
    \caption{A single-layer one-way computation depicted as a quantum circuit. }
    \label{fig: single-layer 1WQC}
\end{figure}

In this theorem, we see that there is still a monotonic increase of the upper bound on the weights as the computation moves along. However, while in Theorem~\ref{theorem: improved weights} the weight upper bound increases by one every other operator, now the upper bound evolves in steps of (variable) width $\kappa_i$, increasing from one step to the other by an amount given by the number of qubits in the associated layer of the underlying one-way computation.

We now consider the simple case of quantum computations with depth 1. This is the only point of the main text where we go through an explicit demonstration. The reason for this is two-fold. First, it gives the reader a glimpse of the proof strategies used throughout this work via the simplest possible example. Secondly, this demonstration equips us with useful insights that are exploited in the more general proofs of Theorems~\ref{theorem: improved depth} and~\ref{theorem: improved depth and weights}, which are deferred to Appendix~\ref{app: Proof of depth and depth-weight trade-off}.
\begin{lemma}[Single-layer computation]\label{lemma: single-layer result}
    A single-layer one-way computation on a computation-specific graph state $\ket{\mathcal{G}}$ with non-adaptive measurements performed along the $\pi/4$ direction on the equator of the Bloch sphere has an associated single-layer PBC.
\end{lemma}
\begin{proof}
    The starting point for the proof is the quantum circuit in Fig.~\ref{fig: single-layer 1WQC} which depicts the one-way computation in question. The first step consists of transforming the $T$ gates into $T$ gadgets.
    
    Since we are considering that we have $d_{\mathrm{1W}} = 1$, we note that $\mathcal{O}_2 \equiv \mathcal{O}_3:\quad \mathrm{GM}_{\ell_1} \prec \mathrm{CM}_{\ell_1}$. Thus, we take all gadget measurements and process them first, before dealing with the computational measurements. In doing this, we note that each measurement stemming from the $T$ gadget arrives at the beginning of the circuit $\mathcal{C}$ as
    $Q_i = Z_i Z_{t+i}$ which anti-commutes with $G_i\,.$
    
    We know that, in general, the order in which these operators are processed is relevant. Here, we show that, in this particular case, they can all be processed simultaneously. Take first $Q_1 = Z_1Z_{t+1}.$ Since it anti-commutes with $G_1$ it is dealt with by drawing $m_1$ uniformly at random from $\{0,1\}$ and then replacing $Q_1$ with the Clifford unitary
    \begin{equation*}
        V(G_1,Q_1,m_1) = \frac{G_1 + (-1)^{m_1}Q_1}{\sqrt{2}}.
    \end{equation*}
    To simplify notation, we denote this and all upcoming $V$ unitaries solely by $V(m_i)$; the reader should keep in mind that they also depend on $G_i$ and $Q_i$.
    
    Next comes $Q_2.$ This operator arrives at the beginning of the quantum circuit as $Q_2 = Z_2Z_{t+2}.$ But, after processing $Q_1$, we now have the unitary $V(m_1)$ through which $Q_2$ also needs to be back-propagated. Crucially, $[Q_2,G_1] = [Q_2,Q_1] = 0$ which means that it is pushed through $V(m_1)$ without being altered. It is then processed much like $Q_1$ by making a coin toss to decide $m_2$ and introducing the corresponding Clifford unitary $V(m_2).$ In the same fashion, each subsequent $Q_j$ is passed through previously added $V(m_i)$ unitaries, with $i<j$, without being changed.

    This means that all gadget measurements can be processed simultaneously in the classical machine by drawing a $t$-bit string $\boldsymbol{m}\in \{0,1\}^t$ uniformly at random and then adding the unitaries
    \begin{equation}\label{eq: V(mi) single-layer}
        V(m_i) = \frac{G_i + (-1)^{m_i} Q_i}{\sqrt{2}}
    \end{equation}
    to the beginning of the circuit. For convenience, let us denote $V=\prod_{i=1}^t V(m_i).$
    
    Next, we need to handle the computational measurements. These can similarly all be processed simultaneously. They are back-propagated through the adaptive Clifford circuit $\mathcal{C}$ and $V$, leading to Pauli operators with the following form:
    \begin{equation}\label{eq: Final operators to be measured in single-layer}
        Z_i \xlongrightarrow{\mathcal{C}V}
        \begin{cases}
                \text{if }m_i = 0:\quad P_i = R_i Y_{t+i}\\
                \text{otherwise}:\quad P_i = G_i R_i X_{t+i}
            \end{cases},
    \end{equation}
    with $R_i = (-1)^{\sum_{a\in \mathcal{N}(i)} m_a} \left( \prod_{b\in \mathcal{N}(i)} G_b\right)\left( \prod_{c\in \mathcal{N}(i)} Z_{t+c} \right)$.
    Since the calculations are somewhat extensive, we leave the details for Appendix~\ref{app: proof for of single-layer Paulis}. 
    
    It can be shown that, for any $\boldsymbol{m}$, $[P_i,\, P_j] = 0,\,\forall i,j$ so that all $t$ operators are independent and compatible. Additionally, they are also compatible with the stabilizers of the graph state. Thus, these operators are identified as Paulis that need to be measured in the quantum hardware. Importantly, they can be processed and, therefore, measured simultaneously, comprising a PBC with a single layer. This concludes the proof.
\end{proof}

Note that since the processing order is \emph{not} $\mathcal{O}_1$, Theorem~\ref{theorem: improved weights} is \emph{not} guaranteed to hold. Thus, we managed to perform a single-layer computation at the expense of potentially larger Pauli weights. 

Let us look more carefully into the weights of the measurements in this particular case of single-layer computation. From Eq.~\eqref{eq: Final operators to be measured in single-layer}, it is easy to see that the weight of each Pauli operator $P_i$ (in the magic register) is $w_i = N_i + 1,$ where $N_i$ denotes the number of neighbors of the $i$th qubit of the graph state. For a fully connected graph, this means that each operator has weight $w_i = t$. On the other hand, if the degree of each vertex of the graph is $O(1)$, each Pauli measurement has weight $O(1)$ and the average weight is constant (rather than linear).

The work of Markov and Shi~\cite{MarkovShi2008} shows that 1WQC can be classically simulated in time exponential in the treewidth of the underlying graph. This means that any graph having a treewidth logarithmic in its size will lead to a computation that can be efficiently classically simulated. The maximal degree of a graph and its treewidth are independent properties. However, if a graph has maximal degree 2, its treewidth can be at most two, leading to efficient classical simulation. Hence, we have the following corollary.
\begin{corollary}
    No single-layer PBC with an average weight smaller than or equal to three can lead to quantum advantage.
\end{corollary}

A similar result holds for the adaptive, multi-layered case of Theorem~\ref{theorem: improved depth}. For details, see Appendix~\ref{app: Proof of depth and depth-weight trade-off}.

\subsection{Remarks on generality}\label{subsec: Generality}

Theorems~\ref{theorem: improved weights}, \ref{theorem: improved depth}, and \ref{theorem: improved depth and weights}, as well as Lemma~\ref{lemma: single-layer result}, take as starting point a one-way computation on a computation-specific graph state $\ket{\mathcal{G}}$, with a measurement pattern requiring only measurements along the $\pm\pi/4$ directions on the equator of the Bloch sphere. In all generality, a Clifford+$T$ quantum circuit with $t$ $T$ gates may be transformed into a one-way computation with a more general underlying resource state $\ket{\mathcal{R}}$ that is local-Clifford equivalent to $\ket{\mathcal{G}}$. One may wonder how the inclusion of such local Clifford transformations may influence our results. The answer is simple. Theorem~\ref{theorem: improved weights} is not altered by considering any general resource state $\ket{\mathcal{R}}$. On the other hand, Theorems~\ref{theorem: improved depth} and \ref{theorem: improved depth and weights}, and Lemma~\ref{lemma: single-layer result} are modified by simply adding a constant factor of 1 to the computational depth upper bounds stated.

Proving these generalized versions of the results can be done by considering every possible local Clifford transformation that can be applied to the generic graph state $\ket{\mathcal{G}}$ and appropriately generalizing the stabilizers in Eq.~\eqref{eq: generators of a graph state} to include every such possibility. While straightforward, this completely general examination is quite laborious. For this reason, we chose to formulate our results in terms of an input graph state and provide the corresponding proofs, as it makes the presentation shorter and simpler, and is more elucidating than the completely general proofs.

\section{Greedy algorithm for smaller-weight measurements}\label{sec: Greedy algorithm}

In this section, we present a heuristic algorithm that reduces the weight of the Pauli measurements to be performed in a PBC. To understand the idea behind this algorithm, consider a PBC such that a sequence of Pauli operators $L_P = \{P_1,..., P_{r-1}\}$ have already been measured. The state of the system after such a sequence is given by
\begin{equation*}
    \ket{\psi_{r-1}} \propto \prod_{i=1}^{r-1}\frac{I^{\otimes t} + (-1)^{\sigma_i}P_i}{2} \ket{T}^{\otimes t}\,,
\end{equation*}
where $\sigma_i$ is the outcome of the corresponding Pauli measurement $P_i$; these outcomes are stored in a list denoted $L_{\sigma}$. Continuing the PBC procedure, the ensuing Pauli measurement, $P_{r}$, is discovered. Performing this Pauli measurement causes the state to evolve such that:
\begin{equation*}
    \ket{\psi_r} \propto \frac{I^{\otimes t} + (-1)^{\sigma_{r}}P_{r}}{2} \ket{\psi_{r-1}}\,.
\end{equation*}

Since the quantum state $\ket{\psi_{r-1}}$ is stabilized by $\left< (-1)^{\sigma_1} P_1,\, \dots,\, (-1)^{\sigma_{r-1}} P_{r-1} \right>$, it is not hard to show that $\ket{\psi_{r}}$ can equivalently be obtained if, rather than measuring $P_{r}$, one decides to measure any Pauli operator of the form $P_{r}\prod_{j\in \mathcal{W}} (-1)^{\sigma_j}P_j$. Here, $\mathcal{W}$ denotes one of the $2^{r-1} - 1$ possible subsets of $\{1,\dots,r-1\}$ (excluding the empty set $\emptyset$, which corresponds to considering $P_r$ itself). This implies that, whenever we find a Pauli operator $P_{r}$ at the $r$th time step, we can measure exactly that Pauli or any one of the $2^{r-1}-1$ Pauli operators that perform an equivalent state transformation. Thus, we are free to choose whichever operator has the lowest weight.

The discussion above shows that, for the last time step $t$, the total number of equivalent Pauli measurements is $2^{t-1}$. Thus, if we try to analyze all possibilities, we incur an exponential classical processing overhead of $O(2^t)$. To preserve the efficiency of the PBC procedure, the search for smaller-weight Pauli operators should be restricted to a polynomial-sized subset of equivalent Pauli operators. Put differently, the number of subsets $\{\mathcal{W}_i\}_{i=1}^N$ considered  at each time step must be such that $N=O(\poly(t))$. An option that incurs only a constant overhead to the processing of each Pauli operator (and, therefore, an overhead of $O(t)$ to the overall procedure) is to consider, at each time step $r$, only one set: $\mathcal{W}=\{1,\dots,r-1\}$. This means that all one needs to do is compare the weight of the Pauli operator $P_{r}$ with that of $P_{r}\prod_{j=1}^{r-1} (-1)^{\sigma_j}P_j$ and choose the one with the smallest weight.
\begin{algorithm}[t]
    \SetKwData{comb}{comb}
    \SetKwFunction{Combinations}{Combinations}
    \SetKwFunction{Length}{Length}
    \SetKwFunction{FindWeight}{FindWeight}
    \SetKwInOut{Input}{Input}
    \SetKwInOut{Output}{Output}
    \Input{($L_P$, $L_{\sigma}$, $P_r$, $\mathsf{go}$)}
    \Output{$P_r^{\prime}$, the Pauli operator to be measured.}
    \BlankLine
    $r \leftarrow$ \Length{$L_P$} + 1 \tcc*{Current time step}
    $w \leftarrow$ \FindWeight{$P_r$} \tcc*{Weight to beat}
    $P_r^{\prime} \leftarrow P_r$\;
    \For{$a\leftarrow 0$ \KwTo $\mathsf{go}$}{
        \ForEach{$\mathcal{W}$ in \Combinations{$\{1,\dots,r-1\}$, $a$} $\cup$ \Combinations{$\{1,\dots,r-1\}$, $r-1-a$}}{ 
            $P_{\text{new}} \leftarrow P_r \prod_{j\in \mathcal{W}} (-1)^{L_{\sigma}[j]}L_P[j]$\;
            $w_{\text{new}} \leftarrow$ \FindWeight{$P_{\mathrm{new}}$}\;
            \If{$w_{\mathrm{new}} < w$}{
                $P_r^{\prime} \leftarrow P_{\text{new}}$\;
                $w\leftarrow w_{\text{new}}$\;
            }
        }
    }
    \Return $P_r^{\prime}$
    \caption{Greedy algorithm}\label{alg: Greedy algorithm}
\end{algorithm}

Alternatively, if we restrict ourselves to subsets $\mathcal{W}$ of size either 1 or $r-2$, the total number of possible subsets is $2r,\forall r\geq 4$. This incurs a linear overhead to the classical processing of each Pauli operator since we need to compare the weight of $P_{r}$ to that of a total of $O(r)$ other Pauli operators. In total, this leads to a contribution of $O(t^2)$ to the complete PBC procedure. If we allow the subsets $\mathcal{W}$ to have size 2 or $r-3$, then the total number of equivalent Pauli measurements to be tested against $P_{r}$ is $O(r^2)$, i.e., we incur a quadratic overhead in the classical processing of each Pauli, leading to a total overhead of $O(t^3)$ to the complete PBC procedure.

These ideas allow us to construct a greedy algorithm that, at each time step $r$, given the Pauli operator to be measured, $P_{r}$, searches for an alternative, equivalent Pauli measurement with a smaller weight. The efficiency of the algorithm is controlled by a parameter that we call ``greedy order'', denoted $\mathsf{go}\geq 0$. Once we set a specific value for $\mathsf{go}$, the algorithm searches for Pauli measurements with better weight among the (sub)sets $\mathcal{W} \subseteq \{1,\,\dots,\,r-1\}$ with size $r-1-a$ and $a$, with $0\leq a\leq \mathsf{go} \leq (r-1)/2$. This means that the algorithm incurs a time overhead given by:
\begin{equation*}
    \tau_{\mathrm{greedy}} = \sum_{r=1}^t \left( 2\sum_{a=0}^{\mathsf{go}} \frac{(r-1)!}{(r-1-a)!\,a!} -1 \right)\,.
\end{equation*}
Thus, $\mathsf{go}=0$ implies a linear overhead to the entire PBC procedure, $\mathsf{go} = 1$ a quadratic overhead, $\mathsf{go} = 2$ a cubic one, and so on.

Pseudocode for this greedy algorithm is outlined in Algorithm~\ref{alg: Greedy algorithm}. We included this algorithm in the code found at~\cite{GitRepoPBC} and used it to perform PBC compilation on the same RQCs and hidden-shift circuits (HSCs) studied in our previous work~\cite{PeresGa2023}. The latter are particularly useful since they have a known deterministic outcome; thus, they can be used to verify that the greedy algorithm is working correctly. In the results that follow, the RQCs with $t$ ranging from 4 to 22 and the HSCs with $t=14$ were simulated using an actual (Schrödinger-type) classical simulator. On the other hand, the circuits with larger $T$ counts were simulated using the dummy simulator (recall Sec.~\ref{subsec: Improved weights} and see Ref.~\cite{PeresGa2023} for details).
\begin{figure}
    \centering
    \includegraphics[width=0.9\columnwidth]{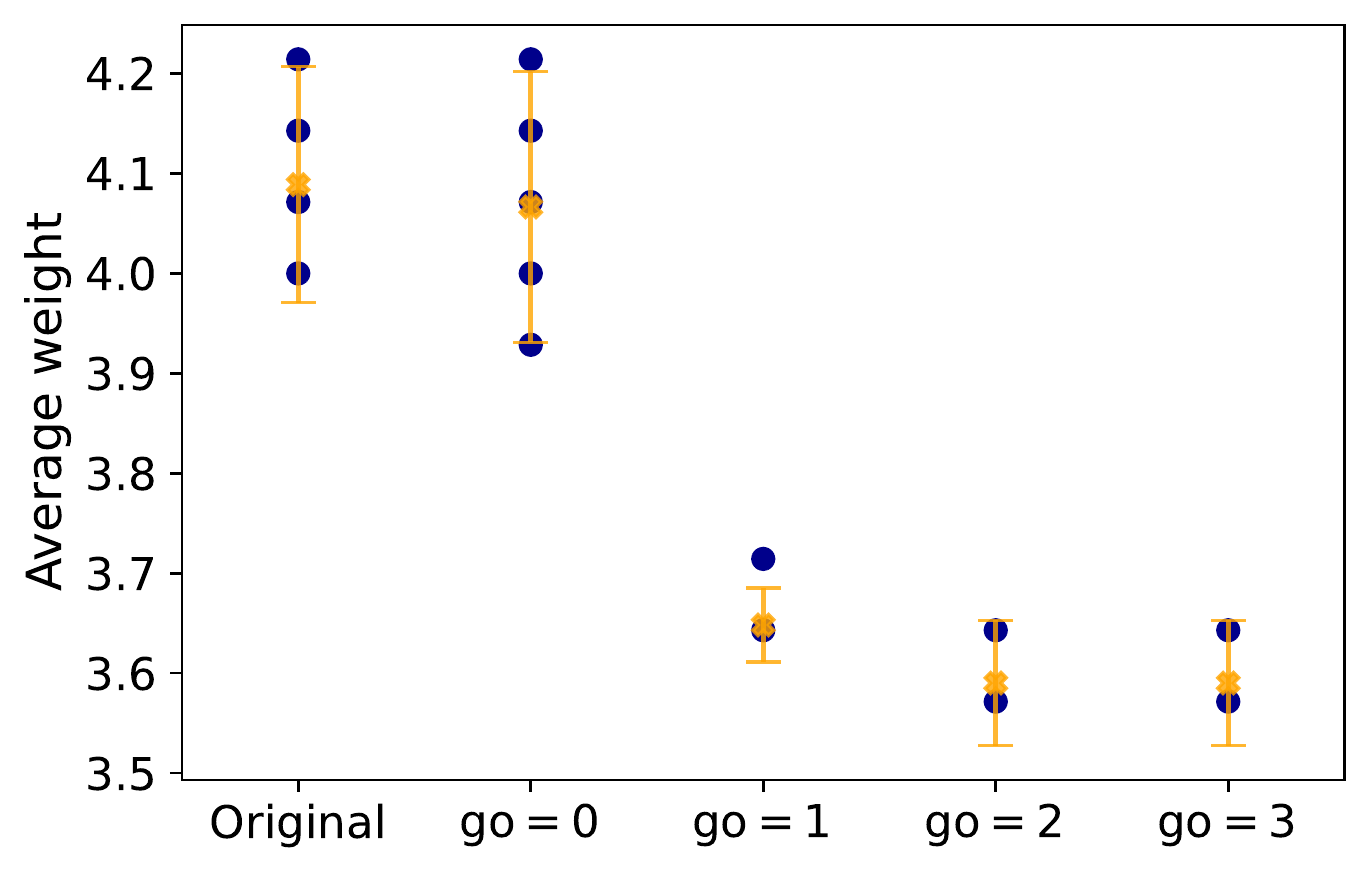}
    \caption{Effect of the greedy algorithm on the possible average weight of the Pauli measurements in the PBCs obtained from hidden-shift circuits with $t=14$ and varying $n=\{10, 14, 18, 22, 28\}$. The results are independent of $n$. We register reductions of 0.7\%, 10.8\%, and 12.3\% to the mean average weight (yellow crosses) respectively for $\mathsf{go}=0,1,\text{ and }2$. The dark blue dots represent the different possible average weights obtained for the PBCs that simulate the hidden-shift circuits under consideration and the error bars depict two standard deviations.}
    \label{fig: Greedy -- small HSCs}
\end{figure}

For the HSCs with $t=14$, the effect of the greedy algorithm is represented in Fig.~\ref{fig: Greedy -- small HSCs}. The results are independent of the number $n$ of qubits in the original circuit. We see that the original circuits were compiled to a total of four different possible average weights (equivalently, a total of four possible \textsc{cnot} counts). Applying the greedy algorithm with $\mathsf{go}=0$ has only a small effect of adding a fifth possible average weight lower than the original four, reducing the mean average weight only by $\sim 0.7\%$. With $\mathsf{go}=1$ the impact of the greedy algorithm is more expressive. In particular, the considered HSCs lead to PBCs with one out of two possible distinct average weights, both of which are lower than the possible average weights of the PBCs obtained without the greedy algorithm or with $\mathsf{go}=0$. A similar effect is seen for $\mathsf{go}=2$, where only two possible distinct average weights occur, the largest of which corresponds to the smallest value obtained with $\mathsf{go}=1$. Focusing on the mean average weight associated with these results, $\mathsf{go}=1$ and $\mathsf{go}=2$ lead to improvements of $10.8\%$ and $12.3\%$, respectively. Importantly, the reader should recall that any improvements registered to the average weight of PBC correspond directly to the reduction in the number of \textsc{cnot} gates of the PBC-compiled circuits. For these (small) circuits, running $\mathsf{go}=3$ leads to similar results as obtained with $\mathsf{go}=2$, suggesting that the performance of the algorithm might stagnate after that.

One may wonder how the performance of the greedy algorithm changes for larger circuits. To understand this, we picked the same 50 larger HSCs with $n=t=42$ as in our previous work~\cite{PeresGa2023}. In this case, the effect of $\mathsf{go}=0$ is negligible, while $\mathsf{go}=1$ and $\mathsf{go}=2$ achieve improvements to the average weight of $13.7\%$ and $16.7\%$, respectively. Interestingly, the performance of our algorithm seems to improve slightly for these larger circuits.
\begin{figure}
    \centering
    \includegraphics[width=0.9\columnwidth]{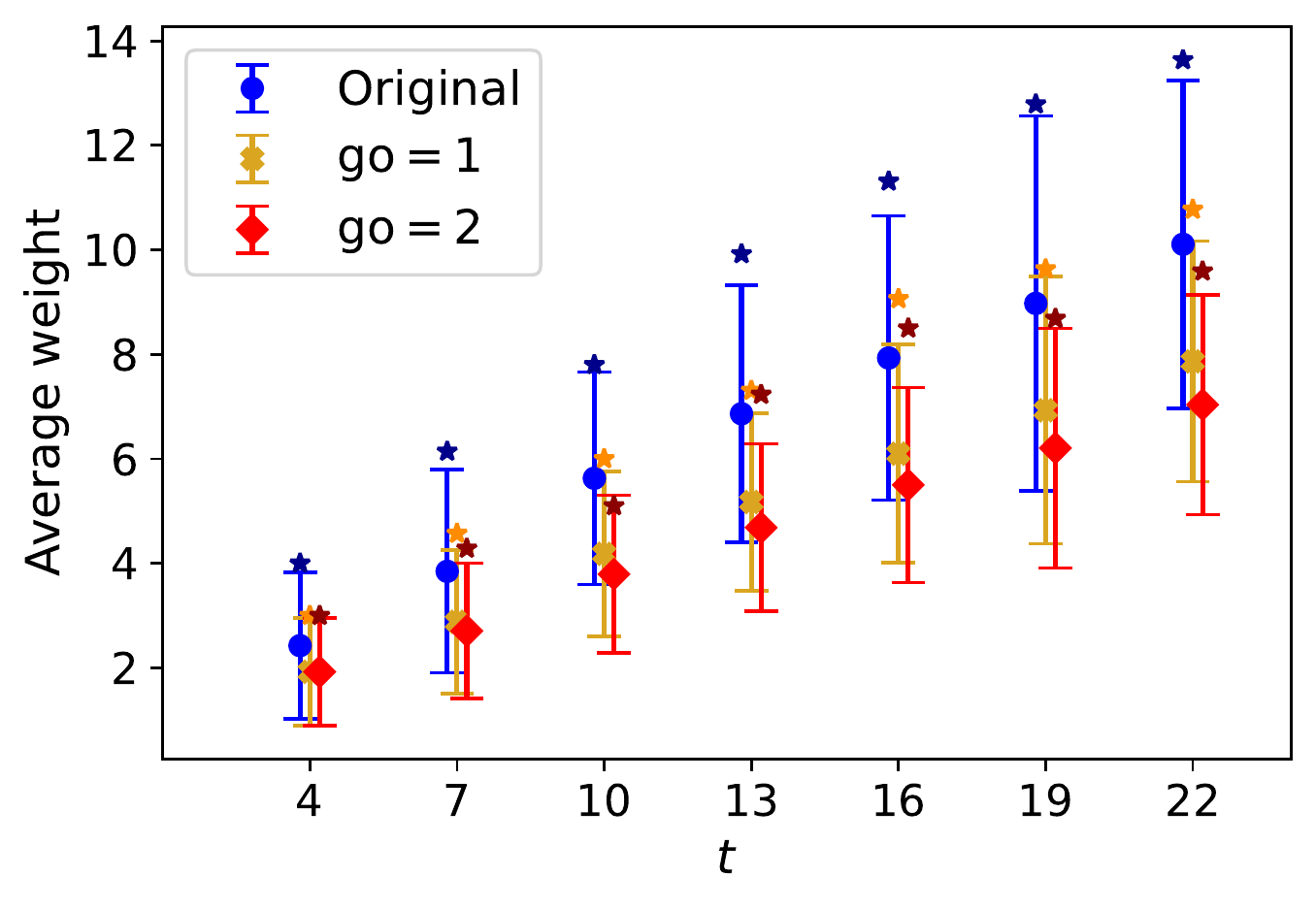}
    \caption{Performance of the greedy algorithm for random quantum circuits with varying $t=\{4, 7, 10, 13, 16, 19, 22\}$ for $\mathsf{go}=1$ (yellow crosses) and $\mathsf{go}=2$ (red diamonds) as compared with the results in the absence of a greedy algorithm (blue circles). The error bars depict two standard deviations from the mean average weight and the stars represent the maximum average weight. The impact of the different orders of the greedy algorithm is noticeable. In particular, for $t=4$, $\mathsf{go}=1$ and $\mathsf{go}=2$ yield the same performance, corresponding to an improvement of roughly $20.6\%$ to the mean average weight. For $t>4$, the improvement to the mean average weight varies between 22.2\% and 25.7\% with $\mathsf{go}=1$ and between 29.7\% and 32.6\% with $\mathsf{go}=2$.}
    \label{fig: Greedy -- small RQCs}
\end{figure}

The results for the RQCs are depicted in Fig.~\ref{fig: Greedy -- small RQCs}. To make this figure more readable, we omitted the results for $\mathsf{go}=0$ even though the improvements associated with this order of the greedy algorithm are more appreciable than for the HSCs, varying between $6.0\%$ and $12.2\%$. For $t=4$, $\mathsf{go}=1$ and $\mathsf{go}=2$ yield the same performance, corresponding to an improvement of roughly $20.6\%$ to the mean average weight. For $t>4$, the improvement to the mean average weight varies between 22.2\% and 25.7\% with $\mathsf{go}=1$ and between 29.7\% and 32.6\% with $\mathsf{go}=2$.

As before, it is interesting to understand how the algorithm performs for larger instances. To that end, we investigated 50 RQCs with $n=t=49$. For these circuits, the mean average weight was reduced by $5\%$, $18.2\%$, and 25.4\%, respectively, for $\mathsf{go}=0$, $\mathsf{go}=1$, and $\mathsf{go}=2$. We see that, contrary to what was observed in the case of larger HSCs, the performance of the greedy algorithm seems to worsen for the larger RQCs. 

The improvements to the average weight (and associated \textsc{cnot} complexity of the final PBC-compiled circuits) described in the previous paragraphs are reported with respect to the results obtained with standard PBC without the greedy algorithm. One may further wonder how the \textsc{cnot} count of PBC-compiled circuits using the greedy algorithm compares to that of circuits compiled using other tools. Here, we focus our comparison against the best-performing ZX calculus compiling algorithm for each case~\cite{KissWet2020reduceT, Wetering2020zxcalculus, Staudacher2023}; for some instances, this corresponds to the basic optimization algorithm, while for others it corresponds to a recent heuristic algorithm for optimizing the two-qubit gate count~\cite{Staudacher2023}. The results are summarized in Tables~\ref{tab: ZXcomp - HSCs} and \ref{tab: ZXcomp - RQCs}. From the analysis of the data, we see that while the performance of the PBC compiler depends only on $t$ irrespective of the number of qubits $n$, the inverse behavior is observed with the ZX calculus algorithms. One detail that is important to point out is that while our PBC compilation allows for the \textsc{cnot} as the only 2-qubit gate, the same is not true in the case of the ZX algorithms, which make use of other 2-qubit gates like the $CZ$ and \textsc{swap} gates. This makes the comparison of the results somewhat unfair in favor of the ZX calculus framework since, for instance, each \textsc{swap} gate requires three \textsc{cnot}s to synthesize (this would increase the gate counts registered in the fourth column of Tables~\ref{tab: ZXcomp - HSCs} and \ref{tab: ZXcomp - RQCs}). Even so, we see that PBC compilation significantly outperforms the best ZX calculus compiler for all of the studied HSCs, save for the circuits with 10 qubits and 14 $T$ gates for which our \textsc{cnot} count is $42.9\%$ larger than the total number of 2-qubit gates in the circuits compiled with the ZX calculus. Similarly, our compiler outperforms those of the ZX calculus for all the RQCs with $t\leq 22$ but has a worse performance for the larger circuits with $n=t=49$. Altogether, these results show the relevance of our work to the current state-of-the-art landscape of quantum circuit compilation.

\begin{table}[]
    \centering
    \begin{tabular}{ccccc}
        \hline\hline
        $n$ & $t$ & $\#\textsc{cnot}$ [PBC, $\mathsf{go}=2$] & $\#\text{2-qubit gates}$ [ZX] & $\Delta$ \\
        \hline
        10 & \multirow{6}{*}{14} & \multirow{6}{*}{$50\pm 1$} & $35\pm 6$ & $+42.9\%$\\
        14 & & & $55\pm 8$  & $-9.1\%$ \\
        18 & & & $74\pm 10$ & $-32.4\%$ \\
        22 & & & $98\pm 13$ & $-49.0\%$ \\
        28 & & & $133\pm 15$ & $-62.4\%$\\
        32 & & & $159\pm 18$ & $-68.6\%$\\
        \hline
        42 & 42 & $156\pm 5$ & $311\pm 28$ & $-49.8\%$\\
         \hline\hline
    \end{tabular}
    \caption{Performance of PBC-compilation with $\mathsf{go}=2$ compared with that of the best-performing ZX calculus circuit-compilation algorithm for the case of HSCs with different parameters $n$ and $t$. The third column registers the average number of \textsc{cnot} gates in the PBC-compiled circuits and the fourth column registers the average number of 2-qubit gate counts in the circuit compiled by the ZX calculus algorithm. The last column registers the percentual difference between these two results: $\Delta = \left( \#\textsc{cnot}_{\mathrm{PBC},\mathsf{go}=2} - \#\text{2-qubit gates}_{\mathrm{ZX}} \right) / \#\text{2-qubit gates}_{\mathrm{ZX}}$.}
    \label{tab: ZXcomp - HSCs}
\end{table}

\begin{table}[]
    \centering
    \begin{tabular}{ccccc}
        \hline\hline
        $n$ & $t$ & $\#\textsc{cnot}$ [PBC, $\mathsf{go}=2$] & $\#\text{2-qubit gates}$ [ZX] & $\Delta$ \\
        \hline
        \multirow{7}{*}{25} & 4 & $8\pm 4$ & $191\pm 12$ & $-95.8\%$ \\
         &  7 & $19\pm 9$ & $190\pm 12$ &   $-90.0\%$\\
         & 10 & $38\pm 15$ & $190\pm 13$ &  $-80.0\%$\\
         & 13 & $61\pm 21$ & $191\pm 12$ &  $-68.1\%$\\
         & 16 & $88\pm 30$ & $189\pm 12$ &  $-53.4\%$\\
         & 19 & $118\pm 44$ & $191\pm 13$ & $-38.2\%$\\
         & 22 & $155\pm 46$ & $190\pm 13$ & $-18.4\%$\\
        \hline
        49 & 49 & $1021\pm 187$ & $795\pm 26$ & $+28.4\%$\\
         \hline\hline
    \end{tabular}
    \caption{Performance of PBC-compilation with $\mathsf{go}=2$ compared with that of the best-performing ZX calculus circuit-compilation algorithm for the case of RQCs with different parameters $n$ and $t$. The third column registers the average number of \textsc{cnot} gates in the PBC-compiled circuits and the fourth column registers the average number of 2-qubit gate counts in the circuit compiled by the ZX calculus algorithm. The last column registers the percentual difference between these two results: $\Delta = \left( \#\textsc{cnot}_{\mathrm{PBC},\mathsf{go}=2} - \#\text{2-qubit gates}_{\mathrm{ZX}} \right) / \#\text{2-qubit gates}_{\mathrm{ZX}}$.}
    \label{tab: ZXcomp - RQCs}
\end{table}

We leave further analyses of the performance of the greedy algorithm to Appendix~\ref{app: More results on greedy}. We also discuss a variation of the algorithm and compare its performance to that of Algorithm~\ref{alg: Greedy algorithm}.

Since in noisy intermediate-scale quantum devices, two-qubit gates are prone to more errors, reducing the number of such gates is desirable as it often has a positive impact on the overall fidelity of the circuit. Considering that the mean average weight has a direct interpretation as the mean number of \textsc{cnot} gates needed in the PBC-compiled quantum circuits, the improvements registered in this section are striking and should have important impacts on near-term quantum computing solutions.

\section{Incompatible and constant-weight Pauli-based computation}\label{sec: incPBC}

In Sec.~\ref{subsec: Improved weights}, we improved the weights of the Pauli measurements in a PBC via Theorem~\ref{theorem: improved weights}. However, the average weight of these measurements is still linear in $t$ (Corollary~\ref{corollary: maximum weight upper bound}). Complementarily, in Sec.~\ref{subsec: Improved depth}, we discussed certain conditions under which a PBC is guaranteed to be driven by constant-weight measurements only. These results provide partial answers to the questions formulated at the beginning of Sec.~\ref{sec: Improvements I - connection to 1WC}. Here, we complement these results by engendering a universal model for quantum computation that requires only constant-weight Pauli measurements performed on a separable input state. We dub this model incompatible, constant-weight Pauli-based computation (incPBC).

We use Perdrix's scheme~\cite{Perdrix2005} as a starting point and modify it so that it requires only Pauli measurements. To that end, we relax the condition that the computation needs to be performed using measurements only and, as in PBC, allow the preparation of $\ket{T}$ magic states.

We note that the \textsc{cnot} and Hadamard gates, implemented in Perdrix's scheme as depicted in Figs.~\ref{fig: State transfer CNOT} and \ref{subfig: State transfer: H}, respectively, require only Pauli measurements and are, therefore, already in the form we are interested in. Universality further requires a construction for implementing the $T$ gate which avoids the non-stabilizer measurement $(X-Y)/\sqrt{2}$ used in Ref.~\cite{Perdrix2005} (Fig.~\ref{subfig: State transfer: T}) or any other such non-Pauli measurements. This can be done provided we have access to copies of the magic state $\ket{T}$. In that case, the $T$ gate can be implemented deterministically, up to a \emph{Clifford} unitary,  using the construction in Fig.~\ref{fig: T gadget as in BSS2016.}~\cite{BSS2016}.

An important detail is that half of the time, the $T$ gate is implemented up to a unitary that involves the $S$ gate. Thus, we also need to find a construction that deterministically carries out this unitary. It is not hard to show that the implementation in Fig.~\ref{fig: Implementation of the S gate using state transfer} performs the phase gate up to a Pauli operator $P$ that depends on the measurement outcomes. Finally, the Pauli corrections can be dealt with similarly to Perdrix's scheme.

Suppose we have an $n$-qubit Clifford+$T$ quantum circuit with $t$ $T$ gates, $c_1$ phase and Hadamard gates, $c_2$ \textsc{cnot}s and $w$ final readout measurements. What are the quantum resources needed to implement the same computation using the model defined in the previous paragraphs? The answer is simple. Since the state-transfer-based implementation of each unitary requires exactly one auxiliary qubit, in each computational layer, we may need up to $n$ auxiliary qubits. Thus, we need a quantum processor possessing at least $2n$ qubits. The measurements involved are all Pauli measurements of weight, either 1 or 2. Assuming the highly unlikely scenario wherein the measurement outcomes are such that no corrections are needed, a total of $(w + 2t + 3c_1 + 4c_2)$ measurements need to be performed to carry out the desired computation in this model. In practice, though, on average, an additional overhead linear in the total number of gates will be needed to deal with the tiresome Pauli and Clifford corrections that may arise from state transfer. 
\begin{figure}
    \centering
    \begin{tikzpicture} \node[scale=1.0]{ \tikzset{operator/.append style={fill=gray!15, rounded corners}} \begin{quantikz}[thin lines] \lstick[wires=1]{\ket{\psi}} &\gate[wires=2]{Z \otimes Z} &\qw\rstick{$C T \ket{\psi}$} \\ \lstick[wires=1]{\ket{T}} &\qw &\gate{X} &\qw \end{quantikz} }; \end{tikzpicture}
    \caption{Implementation of the $T$ gate using only Pauli measurements~\cite{BSS2016}. We note that, unlike what happens in Ref.~\cite{Perdrix2005}, the $T$ gate is implemented up to a \emph{Clifford} unitary ($I$, $Z$, $S$, or $ZS$) depending on the measurement outcomes.}
    \label{fig: T gadget as in BSS2016.}
\end{figure}
\begin{figure}
    \centering
    \begin{tikzpicture} \node[scale=1.0]{ \tikzset{operator/.append style={fill=gray!15, rounded corners}} \begin{quantikz}[thin lines] \lstick[wires=1]{\ket{\psi}} &\qw &\gate[wires=2]{Z \otimes Z} &\gate{-Y} &\qw \\ &\gate{X} &\qw &\qw\rstick{$P S \ket{\psi}$} \end{quantikz} }; \end{tikzpicture}
    \caption{Implementation of the $S$ gate, up to a Pauli correction $P$, using only Pauli measurements. $P$ depends on the measurement outcomes.}
    \label{fig: Implementation of the S gate using state transfer}
\end{figure} 

Fortunately, we can do better. Since PBC is only computationally universal, we similarly allow this scheme to comply only with this weaker form of universality. This has immediate consequences in simplifying the framework. Notably, we do not need to worry about any Pauli corrections that may arise, as they can be pushed past subsequent Pauli measurements. We can also avoid the implementation of any single-qubit Clifford gates (i.e., $H$ and $S$) via state-transfer constructions and similarly deal with these by propagating them until the end of the quantum circuit (past all Pauli measurements), with the consequence of changing the nature of the Pauli measurements they get pushed through, but never their weight. This is in the same spirit as PBC.
\begin{table*}[t]
    \centering
    \begin{tabular}{>{\centering\arraybackslash}p{1.2cm} >{\centering\arraybackslash}p{2.9cm} >{\centering\arraybackslash}p{5.9cm} >{\centering\arraybackslash}p{3.9cm} >{\centering\arraybackslash}p{1.5cm} >{\centering\arraybackslash}p{1.5cm}}
        \hline\hline
        & \multicolumn{4}{c}{Quantum resources} \\
        \hline
       Model & Quantum memory & Measurement type & Number of measurements & Weight & Depth \\
        1WQC    & $t$  & Compatible non-stabilizer measurements & $t$ & 1 & $d_{\mathrm{1W}}$ \\
        PBC     & $t$  & Compatible Pauli measurements & $\leq t$ & $\leq t$ & $\leq t$ \\
        incPBC  & $2n$ & Incompatible Pauli measurements & $w + 2t + 3c_2$ & $\{1,2\}$ & $3d_L$ \\
        \hline\hline
    \end{tabular}
    \caption{Comparison of the quantum resources needed across different measurement-based models of quantum computation for simulating a Clifford+$T$ quantum circuit with $n$ qubits, $t$ $T$ gates, $c_2$ $\textsc{cnot}$s, and $w$ readout computational basis measurements; $d_{\mathrm{1W}}$ denotes the number of layers in the adaptive measurement sequence of 1WQC which is often smaller than the logical depth $d_L$ of the corresponding quantum circuit~\cite{BroadbentK2009parallelizing, HobanWA2014mbcc}. The term ``quantum memory'' refers to the number of qubits that need to be online at any given point of the computation. As explained in the main text, qubit reinitialization and reuse are allowed in the case of incPBC.}
    \label{tab: resource comparison}
\end{table*}

\begin{figure*}[t]
    \centering
    \begin{tikzpicture} \node[scale=0.77]{ \begin{quantikz}[thin lines] \lstick[wires=2]{\ket{\psi}} &\gate{H} &\ctrl{1} &\gate{T} &\gate{H} &\meter{} \\ &\qw &\targ{} &\qw &\qw \end{quantikz}\hspace{0.5cm}$\longrightarrow$\hspace{0.3cm} \begin{quantikz}[thin lines] & & & & &\lstick{\ket{T}} &\gate[wires=2, style={fill=gray!15, rounded corners}]{Z \otimes Z}\gategroup[wires=2,steps=3,style={rounded corners,fill=teal!20, inner xsep=2pt},background]{{$T$}} &\gate[style={fill=gray!15, rounded corners}]{X} & & &\\ \lstick[wires=2]{\ket{\psi}} &\gate{H} &\gate[wires=3, style={fill=gray!15, rounded corners}]{Z \otimes I \otimes X}\gategroup[wires=3,steps=4,style={rounded corners,fill=blue!20, inner xsep=2pt},background]{\textsc{cnot}} &\qw &\qw &\gate[wires=2, label style=red]{P_{\textsc{cnot}}} & &\qw &\gate[label style=red]{C_T} &\gate{H} &\gate[style={fill=gray!15, rounded corners}]{Z}\\ &\qw & &\gate[wires=2, style={fill=gray!15, rounded corners}]{X\otimes Z} &\qw & &\qw &\qw &\qw &\qw &\qw\\ &\lstick{\ket{0}} & & &\gate[style={fill=gray!15, rounded corners}]{X} & & & & & & \end{quantikz} }; \end{tikzpicture} \begin{tikzpicture} \node[scale=0.75]{ \hspace{3.8cm}$\longrightarrow$\hspace{0.3cm} \begin{quantikz}[thin lines] & & & &\lstick{\ket{T}} &\gate[wires=2, style={fill=gray!15, rounded corners}]{Z \otimes X}\gategroup[wires=2,steps=2,style={rounded corners,fill=teal!20, inner xsep=2pt},background]{} &\gate[style={fill=gray!15, rounded corners}]{X} &\\ \lstick[wires=2]{\ket{\psi}} &\gate[wires=3, style={fill=gray!15, rounded corners}]{X \otimes I \otimes X}\gategroup[wires=3,steps=3,style={rounded corners,fill=blue!20, inner xsep=2pt},background]{} &\qw &\qw &\qw & &\qw &\gate[style={fill=gray!15, rounded corners}]{HC_T^{\dagger}XC_T H}\\ & &\gate[wires=2, style={fill=gray!15, rounded corners}]{X\otimes Z} &\qw &\qw &\qw &\qw &\qw\\ \lstick{\ket{0}} & & &\gate[style={fill=gray!15, rounded corners}]{X} & & & & \end{quantikz} }; \end{tikzpicture}
    \caption{Illustration of how the incPBC model can weakly simulate any Clifford+$T$ quantum circuit by performing only Pauli measurements of weight 1 or 2. The outcomes of the measurement $HC_T^{\dagger} X C_T H$ are guaranteed to obey the same probability distribution as those of the readout measurement of the quantum circuit. $P_{\textsc{cnot}}$ and $C_T$ denote, respectively, the Pauli and Clifford corrections associated with the implementation of the \textsc{cnot} and $T$ gates following Figs.~\ref{fig: State transfer CNOT} and \ref{fig: T gadget as in BSS2016.}. As with all other figures, white boxes with sharp edges depict quantum gates, while grey boxes with rounded edges represent projective measurements.}
    \label{fig: Illustration of incPBC}
\end{figure*}

In Sec.~\ref{subsec: Improved weights}, we explicitly noted that there are two mechanisms responsible for increasing the weight of the Pauli measurements in a PBC. To ensure a model that relies only on constant-weight Pauli measurements, both of these mechanisms must be bypassed. Avoiding the $V(\sigma_P,\sigma_Q)$ unitaries means allowing measurement incompatibility in the model. Preventing the increase of the weights promoted by propagation through \textsc{cnot}s is done by implementing this gate using the state transfer construction by Perdrix~\cite{Perdrix2005}, depicted in Fig.~\ref{fig: State transfer CNOT}. 

Therefore, the incPBC model we propose here works as follows. The \textsc{cnot} gate and the $T$ gate are implemented using the measurement-based constructions proposed, respectively, in Refs.~\cite{Perdrix2005} and~\cite{BSS2016} and depicted in Figs.~\ref{fig: State transfer CNOT} and \ref{fig: T gadget as in BSS2016.}. The gates are thus implemented deterministically up to Pauli or Clifford corrections that are irrelevant since they can be pushed through any Pauli measurements, potentially changing their nature but without altering their weight. The Hadamard and $S$ gates native to the circuit can similarly be pushed through any Pauli measurements that come from the implementation of the $T$ and \textsc{cnot} gates, again (possibly) changing them but never altering their weight.

To simulate the output distribution of an $n$-qubit Clifford+$T$ quantum circuit with $t$ $T$ gates, $c_2$ \textsc{cnot}s, logical depth $d_L$, and $w$ readout measurements, incPBC requires (up to) $2n$ qubits, the ability to prepare them either in the $\ket{0}$ or the $\ket{T}$ state, and a total of $(w + 2t + 3c_2)$ Pauli measurements. We remark that not only have we managed to remove the dependence on the number of single-qubit Clifford gates in the quantum circuit, but, additionally, by removing the need for the recursive (repeat-until-success) procedure in Perdrix's scheme, the number of measurements is no longer a best-case scenario requiring a convenient and highly unlikely sequence of measurement outcomes. Assuming that each layer of the circuit has at least one \textsc{cnot}, the depth of this model--understood as the total number of measurement layers--is $3d_L.$

In Table~\ref{tab: resource comparison}, the resources needed in this scheme are compared with those needed to perform the same task within 1WQC and the original PBC. We note that the price to pay for having only constant-weight Pauli measurements is two-fold. First, measurement incompatibility needs to be allowed. Second, the number of measurements that need to be performed is much larger: $(w+2t+3c_2)$ instead of $t$. Figure~\ref{fig: Illustration of incPBC} provides a visual depiction of our incPBC model via a simple example.

\subsection*{Relation to other works}\label{subsec: Connection to FBQC}

Other than the obvious tie to Perdrix's work, the astute reader may note a connection to the work by Bartolucci \textit{et al.}~\cite{FBQC2023}, which also makes use of measurements of constant weight. In that paper, the authors construct a scheme for fault-tolerant photonic quantum computation relying on the so-called fusions~\cite{BrowneRudolph2005}; they suitably dub this model fusion-based quantum computation~(FBQC).

One way to look at FBQC is as a practical architectural proposal of how to construct sufficiently large graph states to enable universal quantum computation. Ref.~\cite{FBQC2023} notes that any graph state can be generated (up to local Clifford unitaries) provided one has access to 3-qubit Greenberger--Horne--Zeilinger (GHZ) states and is capable of doing Bell measurements. In the specific context of linear optics, the latter requirement corresponds to the ability to perform Bell fusions, also known as type-II fusions~\cite{BrowneRudolph2005}. 

We now spare some remarks on what makes incPBC and FBQC similar and what sets them apart. First of all, the description above highlights that FBQC can be regarded as a framework for universal quantum computation that makes use of few-qubit measurements. Specifically, Bell fusions consist of two-qubit measurements; additionally, since the scheme also requires access to 3-qubit GHZ states, and since multi-qubit states are not allowed as a resource in Perdrix's scheme or incPBC, a fair comparison requires understanding how such a resource could be generated by projective measurements. Since a GHZ state is stabilized by $\mathcal{S}_{\text{GHZ}} = \left< XXX,\, ZZI,\, IZZ \right>$, we understand that in total FBQC requires at most measurements of weight 3. This contrasts with our scheme, where measurements of at most weight 2 are needed.

FBQC is more general than incPBC in that the Bell fusions used to grow the resource state are assumed to be probabilistic. That is, the desired Bell measurement $(XZ, ZX)$ is performed with probability $1-p_{\text{fail}}$ and it fails with probability $p_{\text{fail}}$, carrying out the separable single-qubit measurements $ZI$ and $IZ$. In contrast, our scheme assumes the measurements always succeed.

Finally, we note that, so far, all operations mentioned concerning FBQC are stabilizer measurements. Therefore, it remains to explain how universality can be achieved therein. The authors provide three different options: (i) applying modified fusion operations (akin to the works by Nielsen~\cite{Nielsen2003} and Leung~\cite{Leung2004}), (ii) making single-qubit non-stabilizer measurements (as in~\cite{Perdrix2005}), or (iii) replacing the resource state with a suitable magic resource state. The latter option is close in spirit to incPBC.

\section{Discussion and outlook}\label{sec: Conclusions}

PBC is substantially different from other methods existing in the literature for saving quantum resources. Compilation techniques often apply and lead to non-adaptive quantum circuits and strategies for resource optimization for adaptive quantum computations are somewhat lacking. PBC addresses this on its own and our work drives this even further, presenting a wide set of results that collectively improve the practical feasibility of this computational model. 

In Sec.~\ref{sec: Improvements I - connection to 1WC}, we formulated three theorems that guarantee distinct non-trivial upper bounds for the average weight of the Pauli measurements in a PBC (a measure directly related to the number of \textsc{cnot}s of the PBC-compiled quantum circuits) and the depth of the PBC (defined as the number of measurement layers). We complemented these more formal results by providing numerical simulations of random quantum circuits. The results indicate that while the theorems promise ``only'' new upper bounds, the pre-compilation technique underlying them has an impact that goes beyond that, leading to PBCs that actually have smaller average weights.

The greedy algorithm proposed in Sec.~\ref{sec: Greedy algorithm} further improves the natural resource savings achieved by the PBC model by providing substantial reductions to the average weight of the Pauli measurements. Importantly, this algorithm provides the option of ``distributing the hardness'' of the computation as one sees fit. That is, the overhead incurred by the greedy algorithm is entirely classical. Hence, one can push the classical machine by increasing the order of the greedy algorithm if one has as a priority reducing the demands on the quantum hardware. The suitable choice will depend on the (classical and quantum) resources available to the user. 

In Sec.~\ref{sec: incPBC}, we defined a new version of PBC, which we dubbed incPBC, that uses only constant-weight Pauli measurements at the expense of (i) allowing measurement incompatibility and (ii) utilizing a number of measurements greater than that used by standard PBC. We also commented on similarities and differences between this model and interesting work done on fusion-based quantum computation.

Put together, the techniques and contributions presented in this work significantly improve the state of the art of the PBC model of quantum computation, making it more amenable to practical implementation. A question that is left open is whether PBC can be formulated with constant weights while retaining measurement compatibility and a maximum number of measurements equal to that of qubits. As a partial answer to this question, we learned that weights of at most three are insufficient for universal quantum computation within the PBC framework whenever $d_{\text{PBC}} \leq d_{\text{1W}}$.

Another interesting line of research is inspired by the greedy algorithm. As expressed in Algorithm~\ref{alg: Greedy algorithm}, the greedy algorithm attempts to find the lowest-weight Pauli operator at a given step. Alternatively, a different optimization criterion can be chosen. For instance, in certain quantum hardware, a gate might exist that is noisier than the \textsc{cnot} gate, so optimizing the PBC sequence to reduce, for instance, the number of $Z$ operators might be more beneficial. One can also think of more sophisticated algorithms that, rather than trying to find the best solution at each step (within the allowed number of tests), try to optimize things globally. In doing this, the algorithm might avoid reducing the weight at one specific step to reap a better reward at later stages of the computation. This might be accomplished by an algorithm with a global (rather than local) reward system, such as seen, for instance, in reinforcement learning algorithms.

We conclude by remarking that the greedy algorithm is looking for lower-weight representations of the generators of an abelian subgroup of the Pauli group. If an algorithm exists that gives concrete performance guarantees for such a task, this could be very impactful in characterizing the experimental feasibility of different algorithms within the PBC framework. The main technical difficulty lies in finding the generators sequentially (i.e. adaptively), as the full sequence of Pauli measurements is unknown \emph{a priori}.

\section*{Acknowledgements}
F.C.R.P. was supported by the Portuguese institution FCT -- Funda\c{c}\~{a}o para a Ci\^{e}ncia e a Tecnologia (Portugal) via the Ph.D. Research Scholarship 2020.07245.BD during the initial stages of this work and later funded by Ayuda Consolidación CNS2023-145392 funded by MICIU/AEI/10.13039/501100011033 and European Union NextGenerationEU/PRTR. E.F.G. acknowledges support from FCT via project CEECINST/00062/2018. This work was supported by the Digital Horizon Europe project FoQaCiA, GA no.101070558, funded by the European Union, NSERC (Canada), and UKRI (U.K.). Numerical simulations were made possible by INCD funded by FCT and FEDER under the project 01/SAICT/2016 nº 022153 and also by the Search-ON2: Revitalization of HPC infrastructure of UMinho (NORTE-07-0162-FEDER-000086), co-funded by the North Portugal Regional Operational Programme (ON.2 -- O Novo Norte), under the National Strategic Reference Framework (NSRF), through the European Regional Development Fund (ERDF). The authors would also like to thank Selman Ipek and Miriam Backens for calling attention to some parts of the manuscript lacking in clarity. Their comments led to modifications that improved the overall readability.

\appendix


\section{Proof of Theorem~\ref{theorem: improved weights}}\label{app: Proof of the weight theorem}

Here, we prove Theorem~\ref{theorem: improved weights}.

\theoremstyle{plain}\newtheorem*{theorem: improved weights}{Theorem~\ref*{theorem: improved weights} \normalfont{(Improved weights)}}
\begin{theorem: improved weights}
Consider a one-way computation to be carried out on a $t$-qubit, computation-specific graph state $\ket{\mathcal{G}}$ with a measurement pattern requiring only measurements along the $\pm\pi/4$ directions on the equator of the Bloch sphere. By taking on the processing order $\mathcal{O}_1$ defined in Eq.~\eqref{eq: ordering for the weights}, the (magic-register) weights of the $2t$ Pauli operators in the (complete) PBC procedure are upper-bounded by \{1,\,1,\,2,\,2,\dots, t-1,\,t-1,\,t,\,t\}.
\end{theorem: improved weights}
\begin{proof}
    The starting point for the proof is the adaptive Clifford circuit depicted in Fig.~\ref{fig: 1WQC-circuit after T gadget}. As stated in the formulation of the theorem, we consider the processing order $\mathcal{O}_1$ given by Eq.~\eqref{eq: ordering for the weights}; hence, we begin with the measurement of the first auxiliary qubit. Back-propagating this measurement leads to: $Z_{t+1} \longrightarrow Q_1 = Z_1 Z_{t+1}.$ It is clear that this operator anti-commutes with $G_1.$ Thus, the standard PBC procedure informs us that its outcome $m_1$ is determined in the classical computer by making a coin toss; $Q_1$ is then dropped from the quantum circuit and replaced by the Clifford operator:
    \begin{equation*}
        V(G_1, Q_1, m_1) = \frac{G_1 + (-1)^{m_1}Q_1}{\sqrt{2}}.
    \end{equation*}
    To alleviate notation, we label all $V$ unitaries by the outcome associated with the Pauli operator that originated it so that $V(G_1, Q_1, m_1) \equiv V(m_1)$. We note that this operator entangles the first qubit of the auxiliary register with the first computational qubit and its neighbors on the graph. Finding $m_1$ decides the presence or absence of the gate $S^{m_1}$ acting on the first data qubit. This means that we are capable of back-propagating the measurement on that qubit until it reaches the beginning of the quantum circuit depicted in Fig.~\ref{fig: 1WQC-circuit after T gadget}. Doing so leads to
    \begin{equation*}
        \begin{split}
            Z_1 \xlongrightarrow{\mathcal{C}}
            \begin{cases}
                \text{if }m_1 = 0:\quad P_1^{\prime} = Y_1X_{t+1}\\
                \text{otherwise}:\quad P_1^{\prime} = X_1X_{t+1}
            \end{cases}.
        \end{split}
    \end{equation*}
    However, we need to remember that the Clifford unitary $V(m_1)$ is now present and we need to propagate $P_1^{\prime}$ through it. For a generic $V$ operator as given by Eq.~\eqref{eq: V unitaries}, an arbitrary Pauli operator $R$ is back-propagated through $V$ in the following manner:
    \begin{equation}\label{eq: explicit back-propagation through V}
        \begin{split}
            R \xlongrightarrow{V}
            \begin{cases}
                \text{if }[R,P] = [R,Q] = 0: & R \coloneqq R\\
                \text{if }[R,P] = \{R,Q\} = 0: & R \coloneqq \alpha QPR\\
                \text{if }\{R,P\} = \{R,Q\} = 0: & R \coloneqq - R
            \end{cases}\,,
        \end{split}
    \end{equation}
    with $\alpha = (-1)^{\sigma_P + \sigma_Q}.$ Applying these update rules to the present case, we obtain
    \begin{equation*}
        \begin{split}
            P_1^{\prime} \xlongrightarrow{V(m_1)}
            \begin{cases}
                \text{if }m_1 = 0:\quad P_1 = \prod_{j\in \mathcal{N}(1)} Z_j Y_{t+1}\\
                \text{otherwise}:\quad P_1 = X_1X_{t+1}
            \end{cases},
        \end{split}
    \end{equation*}
    where the product runs over all qubits neighboring the computational qubit 1. In both cases, $P_1$ is recognized as an anti-commuting Pauli whose outcome $s_i$ is determined via coin toss and that originates the Clifford unitary:
    \begin{equation*}
        V(s_1) = \frac{G_{j\in\mathcal{N}(1)} + (-1)^{s_1}P_1}{\sqrt{2}}\,,
    \end{equation*}
    which needs to be placed at the beginning of the quantum circuit. Once again, we observe that this unitary establishes a connection between the first auxiliary qubit and a subset of computational qubits. The attentive reader will note that $V(s_1)$ also depends on $m_1$, as $P_1 \equiv P_1(m_1)$, although the chosen notation does not make this dependence explicit.

    That neither $Q_1$ nor $P_1$ is an actual quantum measurement that needs to be performed in the QPU is something that we can make sense of qualitatively. Recall that, in PBC, whenever a Pauli operator is recognized as an operator to be measured in the actual quantum hardware, the measurement is reduced to its magic-register component. Now, both $Q_1$ and $P_1$ have weight 1 in the magic register and we do not expect that single-qubit Pauli measurements on a product state $\left| T \right>^{\otimes t}$ lead to extra computational power beyond that of classical computation (as entanglement is lacking from this scenario). The existence of Pauli measurements of weight 1 at the beginning of the computation would mean that the corresponding (magic) qubit serves the sole purpose of being measured in the $X$, $Y$, or $Z$ basis (and could then be removed from the computation). Together, these two statements mean that the same computation could be performed with fewer qubits, thus avoiding the need for weight-1 measurements at the start of the computation. Note that the same reasoning does not apply for measurements in the middle of the PBC procedure because entanglement starts to arise and a single-qubit measurement performed on an entangled state has a potentially non-trivial influence (suffices to think of 1WQC).
    
    Reasoning in a similar way as illustrated for $Q_1$ and $P_1$, we need to process the remaining $2t - 2$ Pauli operators. The structure of the one-way computation ensures that each $Z_{t+k}$ is pushed to the beginning of the quantum circuit in Fig.~\ref{fig: 1WQC-circuit after T gadget}, leading to
    \begin{equation*}
        Z_{t+k} \xlongrightarrow{\mathcal{\mathcal{C}}} Q_k^{\prime} = Z_{k}Z_{t+k}\,,
    \end{equation*}
    while each $Z_k$ leads to
    \begin{equation*}
        \begin{split}
            Z_k \xlongrightarrow{\mathcal{C}}
            \begin{cases}
                \text{if } f_k \oplus m_k = 1:\quad\quad\,\,\,\, P_k^{\prime} = X_kX_{t+k}\\
                \text{if } m_k=0 \wedge f_k = 0:\quad P_k^{\prime} = Y_kX_{t+k}\\
                \text{if } m_k=1 \wedge f_k = 1:\quad P_k^{\prime} = -Y_kX_{t+k}
            \end{cases},
        \end{split}
    \end{equation*}
    where $f_k \equiv f_k (s_j \in \mathcal{I}_k)$ is a Boolean function whose value depends on the set of outcomes $s_j\in \mathcal{I}_k$ (with $j<k$) that influence the measurement basis of the $k$th computational qubit.
    
    For a specific Pauli operator $Q_k^{\prime}$, a certain number of Clifford operators $V(m_i)$ and $V(s_i)$ resulting from previous gadget and computational measurements might be present at the beginning of the quantum circuit $\mathcal{C}$. This means that
    \begin{equation*}
        Q_k^{\prime} \xlongrightarrow{V} Q_k = V^{\dagger} Q_k^{\prime} V \,,
    \end{equation*}
    where $V$ encompasses all Clifford unitaries introduced by the processing of previous measurements. The structure of these unitaries guarantees that $Q_k$ cannot have a non-trivial presence on the qubits of the auxiliary magic register with $i>k$ (even if it can be non-trivial in all qubits of the stabilizer register). On the other hand, potential $V$ unitaries added by the processing of previous Pauli operators may lead to non-trivial contributions in auxiliary qubits labeled $i\leq k.$ Thus, this Pauli measurement can take the following form (on the $t$-qubit magic register): $Q_k = R_k \otimes I^{\otimes t - k},$ with $R_k$ denoting a Pauli operator on $k$ qubits with weight $1\leq w\leq k$. An identical reasoning applies to any $P_k$.
    
    This guarantees that the $2t$ Pauli operators processed in this way have maximum weights (in the magic register) given by $\{1,1,2,2,3,3,\dots,t,t\},$ as stated in the theorem.
    
    As a concluding remark, we remind the reader that not all of these Pauli operators will be measured. Since there are at most $t$ independent and pairwise commuting Pauli operators on $t$ qubits, the maximum number of quantum measurements is still $t$, as explained in the main text.
\end{proof}

\section{Proofs of Theorems~\ref{theorem: improved depth} and \ref{theorem: improved depth and weights}}\label{app: Proof of depth and depth-weight trade-off}

In Sec.~\ref{subsec: Improved depth}, we demonstrated how the PBC associated with any one-way computation carried out on a graph state $\ket{\mathcal{G}}$ and with an associated single-layer measurement pattern is also single-layered, that is, all of the $r\leq t$ Pauli measurements can be performed simultaneously. In this appendix, we consider the more general situation where the underlying one-way computation has $d_{\mathrm{1W}}>1$. There are two different ways to approach this scenario, which lead to Theorems~\ref{theorem: improved depth} and \ref{theorem: improved depth and weights}. The proof of the former consists of a straightforward generalization of the proof of Lemma~\ref{lemma: single-layer result}. 
\newtheorem*{theorem: improved depth}{Theorem~\ref*{theorem: improved depth} \normalfont{(Improved depth)}}
\begin{theorem: improved depth}
Consider a one-way computation to be carried out on a $t$-qubit, computation-specific graph state $\ket{\mathcal{G}}$ with a measurement pattern requiring only measurements along the $\pm\pi/4$ directions on the equator of the Bloch sphere. By taking up the processing order $\mathcal{O}_2$ in Eq.~\eqref{eq: ordering for another depth result -> GMs -> CMs}, the depth of PBC coincides with the depth of the corresponding one-way quantum computation, $d_{\mathrm{1W}}$.
\end{theorem: improved depth}
\begin{proof}
    Assume that we take the processing order $\mathcal{O}_2$, given by Eq.~\eqref{eq: ordering for another depth result -> GMs -> CMs}. That is, all of the gadget measurements are dealt with first, followed by the layered propagation of the computational measurements. 

    From Fig.~\ref{fig: 1WQC-circuit after T gadget}, it may look like this order cannot be realized; for instance, it may seem that we need to know the Clifford correction $(S^{\dagger})^{s_1\oplus 1},$ before propagating the second gadget measurement. In reality, that is not the case, since $Z_{t+2}$ is transformed into $Z_2Z_{t+2}$ via back-propagation through the \textsc{cnot} gate, and $Z_2Z_{t+2}$ commutes with $(S^{\dagger})^{s_1 \oplus 1}$ regardless of the value of $s_1$.
    (This is strikingly different from what happens with the back-propagation of computational measurements, which \textit{do} require the knowledge of $S^{m_i}$ and, therefore, the prior determination of the gadget outcome $m_i$.)

    By choosing the order $\mathcal{O}_2,$ we see that each Pauli $Q_i$ stemming from a gadget measurement will take the form $Q_i = Z_iZ_{t+i}$ and we will have the same sequence of Clifford unitaries, $V = \prod_j V(m_j)$, as in the single-layer case.
    
    Similarly, each Pauli $P_i$ stemming from a computational measurement will take the form given in Eq.~\eqref{eq: Final operators to be measured in single-layer}. The only difference to the single-layer case is that now these $P_i$ measurements are grouped into layers; operators in the same layer can be processed and measured simultaneously, but only after operators in prior layers have been measured (to fix the Clifford correction factor $(S^{\dagger})^{f_i \oplus 1}$ determined by measurements in previous layers):
    \begin{equation}\label{eq: Final operators to be measured -- best depth}
        Z_i \xlongrightarrow{\mathcal{C}V}
        \begin{cases}
                \text{if }m_i \oplus f_i = 0:\,\, P_i = R_i Y_{t+i}\\
                \text{otherwise}:\quad\quad\, P_i = G_i R_i X_{t+i}
            \end{cases},
    \end{equation}
    with $R_i = (-1)^{\sum_{a\in \mathcal{N}(i)} m_a} \left( \prod_{b\in \mathcal{N}(i)} G_b\right)\left( \prod_{c\in \mathcal{N}(i)} Z_{t+c} \right)$.
    
    This means that we end up with a PBC with $d_{\mathrm{1W}}$ layers with the Pauli operators to be measured given by Eq.~\eqref{eq: Final operators to be measured -- best depth}. Because of the block $V$, these Pauli measurements can potentially have weight $t$. However, the same observations made at the end of Sec.~\ref{subsec: Improved depth} apply.
\end{proof}

Theorem~\ref{theorem: improved depth and weights} provides a result wherein the depth of the PBC may be lower than $t$ while, at the same time, establishing better weight upper bounds for the Pauli measurements than the trivial value of $t$. In this sense, this result supplies us with an intermediate approach between Theorems~\ref{theorem: improved weights} and \ref{theorem: improved depth}.
\newtheorem*{theorem: improved depth and weights}{Theorem~\ref*{theorem: improved depth and weights} \normalfont{(Weight-depth trade-off)}}
\begin{theorem: improved depth and weights}
Consider a one-way computation with logical depth $d_{\mathrm{1W}}$ and layering so that the number of computational qubits in layer $\ell_i$ is $\kappa_i$: $\sum_{i=1}^{d_{\mathrm{1W}}} \kappa_i = t.$ The computation is to be carried out on a $t$-qubit, computation-specific graph state $\ket{\mathcal{G}}$ with a measurement pattern requiring only measurements along the $\pm\pi/4$ directions on the equator of the Bloch sphere.
By back-propagating the measurements following $\mathcal{O}_3$ in Eq.~\eqref{eq: ordering for the depth}, the depth of the corresponding PBC is upper-bounded by $\min \{2d_{\mathrm{1W}} - 1,t\}.$ Moreover, the weight of the $2\kappa_i$ Pauli operators stemming from $\mathrm{GM}_{\ell_i}$ and $\mathrm{CM}_{\ell_i}$ measurements is upper bounded by $\sum_{j=1}^{i} \kappa_j\,.$ 
\end{theorem: improved depth and weights}
\begin{proof}
    We approach the back-propagation of measurements following the order $\mathcal{O}_3$ in Eq.~\eqref{eq: ordering for the depth}, rather than $\mathcal{O}_1$ or $\mathcal{O}_2$ which were used to prove, respectively, Theorems~\ref{theorem: improved weights} and~\ref{theorem: improved depth}. 

    Again, Lemma~\ref{lemma: single-layer result} and Appendix~\ref{app: proof for of single-layer Paulis} provide us with important insights that we can use for this proof. It is clear that the $\kappa_1$ Pauli operators stemming from $\mathrm{GM}_{\ell_1}$ will be back-propagated through the circuit to give: $Q_i = Z_iZ_{t+i}$, each of which anti-commutes with $G_i$. From the proof of Lemma~\ref{lemma: single-layer result}, we know that these operators can all be processed simultaneously in the classical computer adding to the beginning of the circuit $\kappa_1$ unitaries $V(m_i)$ of the form given in Eq.~\eqref{eq: V(mi) single-layer}. Here, we denote by $V$ the unitary comprised of all of these: $V=\prod_{i=1}^{\kappa_1} V(m_i)\,.$ All of the Clifford unitaries $V(m_i)$ commute, so that they can be added in any order. So far, things completely resemble the single-layer scenario described in full in the main text.

    Next comes the back-propagation of the $\kappa_1$ measurements associated with the computational qubits of the first layer. After being propagated through the adaptive Clifford circuit and the $\kappa_1$ $V(m_i)$ unitaries introduced in the previous layer, these will take the form:
    \begin{equation}\label{eq: Pi's in the multi-layer}
        Z_i \xlongrightarrow{\mathcal{C}V}
        \begin{cases}
                \text{if }m_i = 0:\,\, P_i = R_i \left( \prod_{c\in \mathcal{N}(i) \setminus \mathcal{A}} Z_c \right) Y_{t+i}\\
                \text{otherwise}:\,\, P_i = R_i \left( \prod_{c\in \mathcal{N}(i) \cap \mathcal{A}} Z_c \right) X_i X_{t+i}
            \end{cases}\hspace{-0.2cm},
    \end{equation}
    with $R_i = (-1)^{\sum_{a\in \mathcal{N}(i)\cap\mathcal{A}} m_a} \left( \prod_{b\in \mathcal{N}(i)\cap \mathcal{A}} G_b Z_{t+b} \right),$ where $\mathcal{A}$ denotes the set of indices labeling gadget measurements that have been identified as anti-commuting (i.e., which originated $V(m_i)$ unitaries). In the present computational layer, $\mathcal{A} = \{1,\dots,\,\kappa_1\};$ but in future steps, this may change as some of the gadget measurements may lead to Pauli operators that need to be measured in the quantum hardware and that, therefore, do not create a unitary $V(m_i)$.  To understand where Eq.~\eqref{eq: Pi's in the multi-layer} comes from, check Appendix~\ref{app: proof for of single-layer Paulis} and, in particular, Remark~\ref{remark: From single-layer to multi-layer} therein.

    The next step is to assess the $P_i$ operators given in Eq.~\eqref{eq: Pi's in the multi-layer}. Since no other operators have been measured, what matters is whether they commute or anti-commute with the generators of the graph state $\{G_j\}_{j=1}^t\,.$ For $m_i=0$, if $\mathcal{N}(i)\setminus \mathcal{A} \neq \emptyset$, each $P_i$ will be identified as an anti-commuting Pauli; otherwise, $P_i$ commutes with every generator of the graph state. The statement $\mathcal{N}(i)\setminus \mathcal{A} = \emptyset$ is equivalent to saying that the $i$th computational qubit does not have neighbors in upcoming layers. In the first layer, we expect that each qubit has neighbors in ensuing layers. Thus, for $m_i=0$, $P_i$ will likely be an anti-commuting Pauli operator. A similar reasoning holds for $m_i=1$. Hence, \textit{a priori}, there is nothing enforcing these Pauli operators to be identified as Pauli measurements to be performed in quantum hardware.

    Can the inclusion of $V(s_i)$ unitaries stemming from this same layer change the nature of the other $P_j$ operators within the layer? Let us suppose that $P_1$ anti-commutes with a generator $G_k$. This will lead to the inclusion of the Clifford unitary:
    \begin{equation*}
        V(s_1) = \frac{G_k + (-1)^{s_1}P_1}{\sqrt{2}}.
    \end{equation*}
    Two situations can happen. First, an upcoming $P_i$ operator ($i\neq 1$) may commute with $G_k$ in which case it is pushed through $V(s_1)$ without being altered and thus preserves its nature. Second, it may anti-commute with $G_k$ in which case it will be modified after back-propagation through $V(s_1)$ following the rule in the second line of Eq.~\eqref{eq: explicit back-propagation through V}. 

    This highlights how taking ordering $\mathcal{O}_3$ substantially complicates things. Operators in a given layer can originate $V(m_i)$ or $V(s_i)$ Clifford unitaries that may influence other operators in that same layer, potentially even changing their nature. This is considerably more involved than the single-layer case or the multi-layer scenario using the $\mathcal{O}_2$ ordering.
    
    To achieve the results stated in the theorem, we take the following approach. Suppose that all gadget measurements of an arbitrary layer $\ell_j$ have been propagated to the beginning of the quantum circuit, leading to the following sequence of operators: $\{Q_i, Q_{i+1},\, \dots,\, Q_{i+\kappa_j}\}.$ Importantly, all of these Pauli operators are compatible. We consider them in increasing order of their indices (but any other order could be used instead). Taking $Q_i$, if it is a Pauli that is recognized as independent and pairwise commuting from all previous measurements performed in the quantum hardware, we store the information about that Pauli, but do \emph{not} perform the measurement immediately. If, instead, the Pauli operator anti-commutes with a generator of the graph state or some previously performed measurement, $W$, we add its corresponding $V(m_i)=[(-1)^{\sigma_{W}}W + (-1)^{m_i}Q_i]/\sqrt{2}$ unitary to the circuit.
    
    All the remaining Pauli operators in the layer $\{Q_{i+1},\,\dots,\,Q_{i+\kappa_j}\}$ are processed exactly in the same way. Importantly, upcoming operators $Q_j$ (with $j>i$) interact with $V(m_i)$ either by being propagated without alteration (if they commute with $W$) or by transforming into $Q_j^{\prime} = (-1)^{m_i}WQ_iQ_j$. After the propagation is completed, this operator is evaluated and processed accordingly, either contributing with a new unitary $V(m_j)$ or being saved for future measurement. After the entire procedure has been completed for a given layer, we are left with a list of operators that have been recognized as Pauli measurements to be performed in the hardware. They can then be measured simultaneously in the quantum hardware, originating a single PBC layer.

    For computational measurements, the same procedure as described in the previous paragraph can be applied. This means that, in total, the PBC can have at most $2d_{\mathrm{1W}}-1$ computational layers. Understanding that the weights of these measured Paulis are those in the theorem requires understanding that the Clifford $V(m_i)$ and $V(s_i)$ unitaries are the only drivers of the weight increase. In each layer $\ell_i$, they ensure that the Pauli operators can only act non-trivially on the first $\sum_{j=1}^i \kappa_j$ qubits of the magic register. This concludes the proof of the theorem.
\end{proof}

\section{Proof of Equation~\eqref{eq: Final operators to be measured in single-layer}}\label{app: proof for of single-layer Paulis}

Here, we provide the explicit proof for the form of the Pauli operators $P_i$ presented in Eq.~\eqref{eq: Final operators to be measured in single-layer}. It is straightforward to see that when the $Z_i$ measurements on the computational qubits are pushed through the Clifford circuit $\mathcal{C}$ (after all corrections $S^{m_i}$ have been fixed), they are transformed so that
\begin{equation*}
    \begin{split}
        Z_i \xlongrightarrow{\mathcal{C}}
        \begin{cases}
            \text{if }m_i = 0:\quad P_i^{\prime} = Y_iX_{t+i}\\
            \text{otherwise}:\quad P_i^{\prime} = X_iX_{t+i}
        \end{cases}.
    \end{split}
\end{equation*}

Next, we need to understand how each of these operators is back-propagated through the Clifford unitary $V$ introduced by the Pauli operators $Q_j$ stemming from the gadget measurements. Recall that $V = \prod_{j=1}^t V(m_j)$ where each $V(m_j)$ takes the form in Eq.~\eqref{eq: V(mi) single-layer}. We note that $[V(m_j),\,V(m_k)] = 0,\,\forall j,k\,.$ This is helpful because it allows us to shuffle these unitaries around at our convenience.

Let us take the Pauli operator $P_i^{\prime},$ and see what happens as it is back-propagated through $V.$ We start by assuming that $m_i=0$ so that $P_i^{\prime} = Y_iX_{t+i}.$ In this case, the propagation of $P_i^{\prime}$ through the unitaries $V(m_j)$ in $V$ can be broken down into three different cases [recall Eq.~\eqref{eq: explicit back-propagation through V}].

\noindent\textbf{Case 1} $(j\neq i \wedge j\notin \mathcal{N}(i))$\textbf{.} This means that $[P_i^{\prime},\,G_j] = [P_i^{\prime},\,Q_j] = 0 \Longrightarrow [P_i^{\prime},\,V(m_j)] = 0\,.$ Thus, $P_i^{\prime}$ passes unchanged through all such $V(m_j)$ unitaries.

\noindent\textbf{Case 2} $(j\neq i \wedge j\in \mathcal{N}(i))$\textbf{.} In this case, $P_i^{\prime}$ still commutes with $Q_j$ but it anti-commutes with $G_j.$ This means that the back-propagation of $P_i^{\prime}$ through all unitaries $V(m_j)$ falling into this category will happen in the following manner:
\begin{equation}\label{eq: aux case 2}
\begin{split}
    P_i^{\prime} & \xlongrightarrow[j_1 \in \mathcal{N}(i)]{V(m_{j_1})} (-1)^{m_{j_1}}G_{j_1}Q_{j_1}P_i^{\prime} \\
    & \xlongrightarrow[j_2\in \mathcal{N}(i)]{V(m_{j_2})} (-1)^{m_{j_1}+m_{j_2}}G_{j_1}Q_{j_1}G_{j_2}Q_{j_2}P_i^{\prime} \\
    & \xlongrightarrow{\hspace{1cm}} \dots \\
    & \xlongrightarrow{\hspace{1cm}} P_i^{\prime\prime} = (-1)^{\sum_{a\in \mathcal{N}(i)} m_a} \prod_{b\in \mathcal{N}(i)} (G_b Q_b) P_i^{\prime}.
\end{split}
\end{equation}

\noindent\textbf{Case 3} $(j = i)$\textbf{.} Finally, we can push the $P_i^{\prime \prime}$ obtained from the previous step through $V(m_i)$ itself, which means that we have the same commutation properties and in case 2, that is, $P_i^{\prime \prime}$ commutes with $Q_i$ but it anti-commutes with $G_i.$ This leads to
\begin{widetext}
\begin{equation}\label{eq: first Pi double prime -- explicit long proof mi=0}
\begin{split}
    P_i^{\prime\prime} \xlongrightarrow{V(m_i=0)} P_i & = (-1)^{\sum_{a\in \mathcal{N}(i)} m_a} G_i Q_ i \prod_{b\in \mathcal{N}(i)} (G_b Q_b) P_i^{\prime} \\
    & = (-1)^{\sum_{a\in \mathcal{N}(i)} m_a} \left( \prod_{b\in \mathcal{N}(i)} G_b \right) G_i \left( \prod_{c\in \mathcal{N}(i)} Z_c Z_{t+c} \right) Z_i Z_{t+i} Y_i X_{t+i}\,.
\end{split}
\end{equation}
\end{widetext}
The important observation now is that $G_i = X_i \prod_{d\in \mathcal{N}(i)} Z_d$, leading to the final form for $P_i$ given by
\begin{equation}\label{eq: aux mi=0 proof}
    P_i = (-1)^{\sum_{a\in \mathcal{N}(i)} m_a} \left( \prod_{b\in \mathcal{N}(i)} G_b \right) \left( \prod_{c\in \mathcal{N}(i)} Z_{t+c} \right) Y_{t+i}\,.
\end{equation}

Next, we need to do the same for the case when $m_i=1$, which means $P_i^{\prime} = X_i X_{t+1}.$ The propagation of this operator through the sequence of unitaries $V(m_j)$ can also be split into three cases. For convenience, we now consider them in a different order. We can easily do this because, as we have seen, the operators $V(m_j)$ all commute and therefore we can shuffle them around at will.

\noindent\textbf{Case 1} $(j\neq i \wedge j\notin \mathcal{N}(i))$\textbf{.} This means that $[P_i^{\prime},\,G_j] = [P_i^{\prime},\,Q_j] = 0 \Longrightarrow [P_i^{\prime},\,V(m_j)] = 0\,.$ Thus, the operator passes unchanged through all such $V(m_j)$ unitaries.

\noindent\textbf{Case 2} $(j = i)$\textbf{.} In this case, $P_i^{\prime}$ commutes with both $G_i$ and $Q_i$ which means that again the operator remains unaffected.

\noindent\textbf{Case 3} $(j\neq i \wedge j\in \mathcal{N}(i))$\textbf{.} Here, we have the same situation as case 2 of the prior scenario, meaning that $P_i$ assumes the form given by Eq.~\eqref{eq: aux case 2}. We can re-write it in the following form:
\begin{equation}\label{eq: aux mi=1 proof}
    P_i = (-1)^{\sum_{a\in \mathcal{N}(i)} m_a} \left( \prod_{b\in \mathcal{N}(i)} G_b \right) \left( \prod_{c\in \mathcal{N}(i)} Z_{t+c} \right) G_iX_{t+i}.
\end{equation}

\begin{table*}[t]
\begin{tabular}{>{\centering\arraybackslash}p{1.3cm} >{\centering\arraybackslash}p{2.0cm} >{\centering\arraybackslash}p{2.5cm} >{\centering\arraybackslash}p{2.5cm} >{\centering\arraybackslash}p{2.5cm} >{\centering\arraybackslash}p{2.5cm} >{\centering\arraybackslash}p{2.5cm}}
\hline\hline
$T$ count & Original & $\mathsf{go}=0$ & $\mathsf{go}=1$ & $\mathsf{go}=2$ & $\mathsf{go}=3$ \\ \hline
    & 33.37 & 33.37 (0\%)     & 28.27 (-15.3\%)& 25.78 (-22.7\%) & 24.20 (-27.5\%)\\
    & 33.05 & 31.68 (-4.14\%) & 28.17 (-14.8\%)& 26.07 (-21.1\%) & 24.60 (-25.6\%)\\
60  & 31.42 & 31.42 (0\%)     & 25.98 (-17.3\%)& 24.02 (-23.6\%) & 22.43 (-28.6\%)\\
    & 28.78 & 28.78 (0\%)     & 24.42 (-15.2\%)& 21.42 (-25.6\%) & 20.02 (-30.5\%)\\
    & 32.85 & 31.85 (-3.04\%) & 28.02 (-14.7\%)& 26.40 (-19.6\%) & 24.63 (-25.0\%)\\
\\
    & 34.71 & 33.91 (-2.30\%) & 28.97 (-16.5\%)& 27.57 (-20.6\%) & 26.43 (-23.9\%)\\
    & 33.66 & 33.66 (0\%)     & 28.69 (-14.8\%)& 26.53 (-21.2\%) & 26.17 (-22.2\%)\\
70  & 33.13 & 33.13 (0\%)     & 28.93 (-12.7\%)& 26.46 (-20.1\%) & 24.81 (-25.1\%)\\
    & 30.81 & 30.81 (0\%)     & 26.54 (-13.9\%)& 24.03 (-22.0\%) & 22.67 (-26.4\%)\\
    & 32.00 & 32.00 (0\%)     & 27.06 (-15.4\%)& 25.14 (-21.4\%) & 23.49 (-26.6\%)\\
    
\\
    & 35.20 & 35.20 (0\%)     & 29.64 (-15.8\%)& 27.79 (-21.1\%) & 26.85 (-23.7\%)\\
    & 32.29 & 32.28 (-0.04\%) & 27.03 (-16.3\%)& 25.25 (-21.8\%) & 23.73 (-26.5\%)\\
80  & 33.19 & 33.19 (0\%)     & 28.71 (-13.5\%)& 26.08 (-21.4\%) & 24.44 (-26.4\%)\\
    & 33.68 & 33.68 (0\%)     & 27.64 (-17.9\%)& 27.08 (-19.6\%) & 24.51 (-27.2\%)\\
    & 33.40 & 33.40 (0\%)     & 28.96 (-13.3\%)& 26.28 (-21.3\%) & 25.39 (-24.0\%)\\
\\
    & 37.19 & 37.19 (0\%)     & 31.46 (-15.4\%)& 29.41 (-20.9\%) & 27.49 (-26.1\%)\\
    & 34.70 & 34.66 (-0.13\%) & 30.44 (-12.3\%)& 28.30 (-18.4\%) & 26.83 (-22.7\%)\\
90  & 35.14 & 35.14 (0\%)     & 30.07 (-14.4\%)& 28.46 (-19.0\%) & 27.07 (-23.0\%)\\  
    & 34.24 & 34.24 (0\%)     & 29.51 (-13.8\%)& 27.36 (-20.1\%) & 25.41 (-25.8\%)\\
    & 33.48 & 33.48 (0\%)     & 28.81 (-13.9\%)& 25.90 (-22.6\%) & 24.93 (-25.5\%)\\
\\
    & 38.41 & 37.98 (-1.12\%) & 32.72 (-14.8\%)& 29.70 (-22.7\%) & 28.14 (-26.7\%)\\
    & 34.68 & 34.68 (0\%)     & 29.72 (-14.3\%)& 27.09 (-21.9\%) & 25.94 (-25.2\%)\\
100 & 35.31 & 35.31 (0\%)     & 29.81 (-15.6\%)& 27.15 (-23.1\%) & 25.90 (-26.6\%)\\
    & 36.39 & 36.39 (0\%)     & 31.44 (-13.6\%)& 29.16 (-19.9\%) & 27.38 (-24.8\%)\\
    & 34.91 & 34.91 (0\%)     & 30.38 (-13.0\%)& 28.72 (-17.7\%) & 27.31 (-21.8\%)\\
\hline\hline
\end{tabular}
\caption{Effect of the greedy algorithm as its order is increased from 0 to 3 for five different $T$ counts ($60$, $70$, $80$, $90$, and $100$) and five randomly generated circuits for each value of $t$. The reductions achieved by the application of the greedy algorithm are clear as soon as the greedy order is greater or equal to 1. For $\mathsf{go}=0$ improvements may or may not occur.}\label{tab: Large RQCs 1 shot each}
\end{table*}

Equations~\eqref{eq: aux mi=0 proof} and \eqref{eq: aux mi=1 proof} correspond to the result presented in Eq.~\eqref{eq: Final operators to be measured in single-layer}, concluding the desired proof.

\begin{remark}[Scope of applicability of the results]\label{remark: From single-layer to multi-layer}
    Equations~\eqref{eq: aux case 2} and \eqref{eq: first Pi double prime -- explicit long proof mi=0} are more general forms of Eqs.~\eqref{eq: aux mi=1 proof} and \eqref{eq: aux mi=0 proof} respectively, and are useful when we consider a computation with multiple layers whose measurements are processed according to the ordering $\mathcal{O}_3$ given in Eq.~\eqref{eq: ordering for the depth}. In that case, Eqs.~\eqref{eq: aux case 2} and \eqref{eq: first Pi double prime -- explicit long proof mi=0} remain valid with the minimal modification that the sums and products run over the elements of the neighborhood of $i$ which have previously been identified as anti-commuting Pauli operators. This subtle new imposition has important consequences. Notably, Eqs.~\eqref{eq: aux mi=0 proof} and \eqref{eq: aux mi=1 proof} are no longer valid since they were obtained by simplifications that assume that all neighbors of the computational qubit $i$ are involved in the products. While that is verified for the multiple-layer scenario when doing the back-propagation following the ordering $\mathcal{O}_2$ in Eq.~\eqref{eq: ordering for another depth result -> GMs -> CMs}, it is not the case when the ordering $\mathcal{O}_3$ is considered. This observation was used in the proof of Theorem~\ref{theorem: improved depth and weights} outlined in Appendix~\ref{app: Proof of depth and depth-weight trade-off}, leading to Eq.~\eqref{eq: Pi's in the multi-layer}.
\end{remark}

\begin{figure}[t]
    \centering
    \includegraphics[width=0.9\columnwidth]{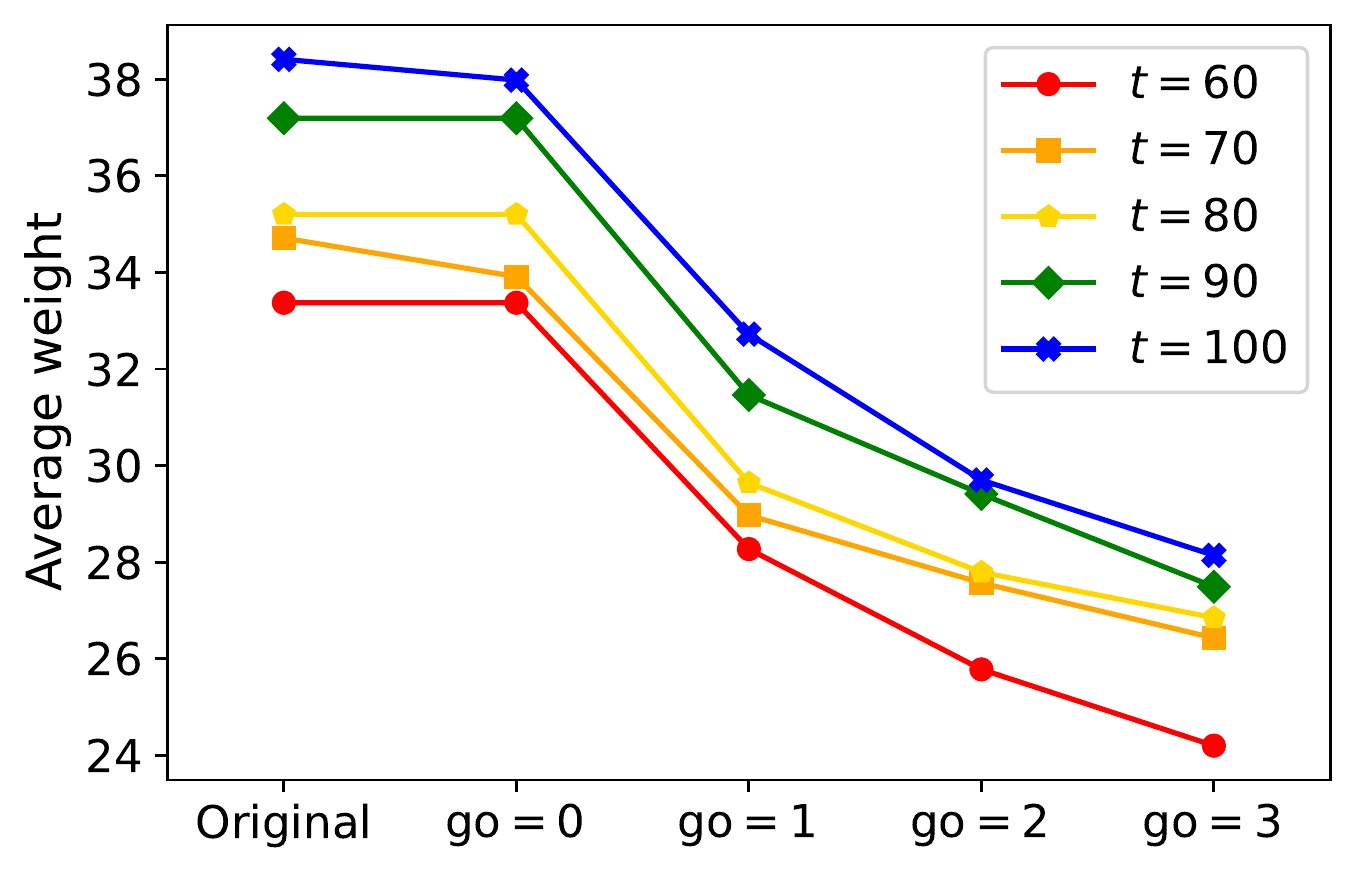}
    \caption{Evolution of the average weight for five randomly generated circuits with $T$ count $60$, $70$, $80$, $90$, and $100$ as the order of the greedy algorithm is increased from 0 until 3.}
    \label{fig: Fixed path -- greedy evolution}
\end{figure}

\begin{figure*}[t]
    \centering
    \includegraphics[width=0.6\textwidth]{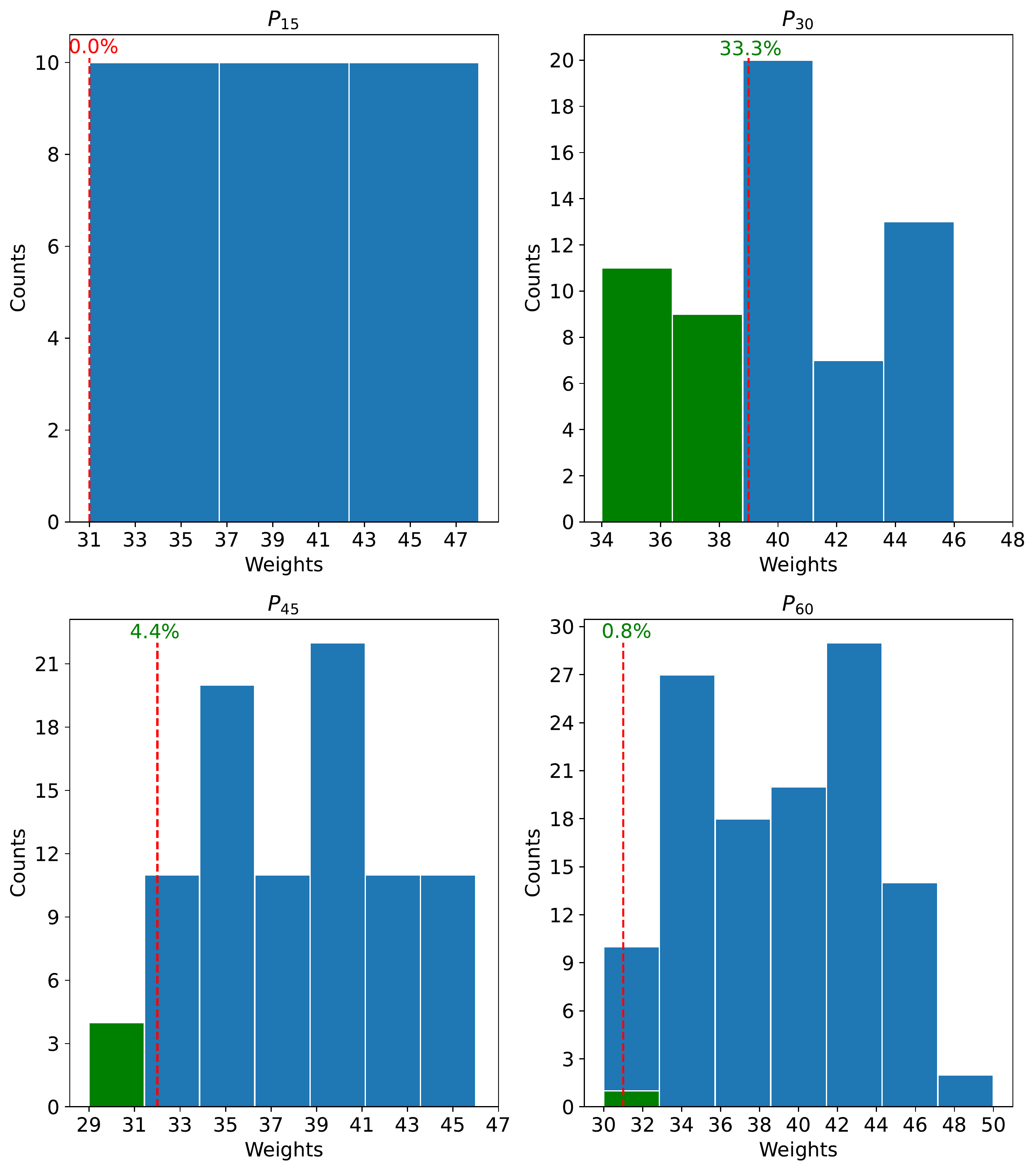}
    \caption{Distribution of the possible weights for the Pauli measurement at time-steps $r=15,$ 30, 45, and 60 of one of the random quantum circuits in Table~\ref{tab: Large RQCs 1 shot each} with $t=60$ when setting $\mathsf{go}=1$. The dashed, red line represents the weight of the Pauli measurement in the absence of the greedy algorithm and the percentage on top of it indicates the proportion of combinations (represented in green) that lead to a weight smaller than that.}
    \label{fig: Histograms -- go=1}
\end{figure*}

\section{Further results and comments concerning the greedy algorithm}\label{app: More results on greedy}

In Sec.~\ref{sec: Greedy algorithm}, we saw how Algorithm~\ref{alg: Greedy algorithm} provided important improvements to the average weight of PBCs associated with both HSCs and RQCs. The performance was analyzed both for smaller circuits using a real, Schrödinger-type simulator and also for larger circuits using a dummy simulator, where measurement outcomes are drawn from a uniform distribution rather than from the actual hard-to-simulate distribution. That the results obtained with the latter can be used to extract conclusions was demonstrated numerically in our prior work~\cite{PeresGa2023}, but also verified concretely for the results with the greedy algorithm with different orders by simulating the smaller circuits with the dummy simulator and verifying that the results obtained were statistically equivalent to those obtained with the real, Schrödinger-type simulator. 

The results presented in the main text were obtained as follows. For the smaller circuits, a total of 100 circuits were transformed into PBCs using 1024 shots/circuit. For the larger circuits, 50 circuits were considered instead, again using 1024 shots for each.

\subsection{Fixed-path analysis}

We now analyze the performance of the greedy algorithm in a slightly different manner. We consider RQCs with $n=49$ and $t=\{60,\, 70,\, 80,\, 90,\, 100\}$. For each of these $T$ counts, we generated five RQCs, each of which we compiled into a single PBC, corresponding to the PBC along the path where all outcomes yield zero. The motivation to fix the path is two-fold. First, it allows us to look into the impact of choosing different orders for the greedy along a fixed sequence of Pauli operators (and not just on average). Second, it allows us to study circuits with larger $T$ counts while still maintaining reasonable simulation times even for larger values of $\mathsf{go}$. The results obtained are presented in Table~\ref{tab: Large RQCs 1 shot each}; a visual depiction of the evolution of the average weight of the Pauli measurements for the different values of $\mathsf{go}$ can be seen in Fig.~\ref{fig: Fixed path -- greedy evolution} for the first circuit of each $T$ count value.

Importantly, we see that, for these larger PBCs involving more qubits and Pauli measurements, setting $\mathsf{go}=0$ often leads to no improvement whatsoever. Contrastingly, setting $\mathsf{go}=1,\,2,\text{ or }3$ leads, respectively, to improvements between 12.3\% and 17.9\%, 17.7\% and 25.6\%, and 21.8\% and 30.5\%. The results further corroborate the observations made in the main text that for RQCs and a fixed value of $\mathsf{go}$, increasing $t$ tends to lead to smaller improvements by Algorithm~\ref{alg: Greedy algorithm}. Nevertheless, the improvements are still substantial, even for these larger values of $t$.

\subsection{Early stopping}

\begin{figure*}
    \centering
    \includegraphics[width=0.6\textwidth]{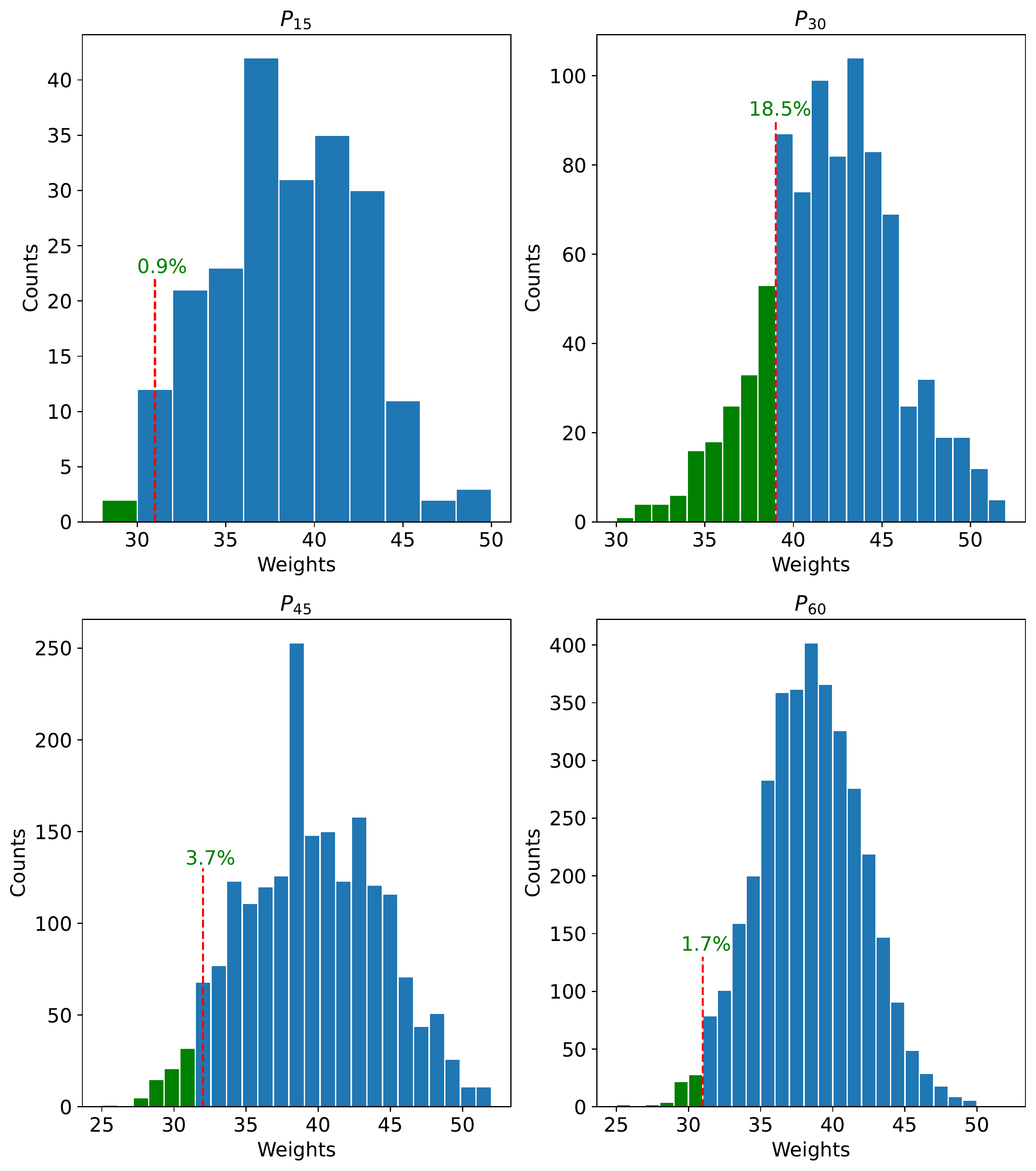}
    \caption{Distribution of the possible weights for the Pauli measurement at time-steps $r=15,$ 30, 45, and 60 of one of the random quantum circuits in Table~\ref{tab: Large RQCs 1 shot each} with $t=60$ when setting $\mathsf{go}=2$. The dashed, red line represents the weight of the Pauli measurement in the absence of the greedy algorithm and the percentage on top of it indicates the proportion of combinations (represented in green) that lead to a weight smaller than that.}
    \label{fig: Histograms -- go=2}
\end{figure*}

One may wonder whether it would be possible to reduce the overhead associated with the greedy algorithm if, rather than searching for the Pauli measurement with smallest weight among the (sub)sets $\mathcal{W} \subseteq \{1,\,\dots,\,r-1\}$ with size $r-1-a$ and $a$, with $0\leq a\leq \mathsf{go} \leq (r-1)/2$, one could stop the search after a predetermined number of attempted combinations, while still guaranteeing that, with high probability a smaller-weight (even if not the smallest-weight) Pauli measurement was found. The success of such an approach relies on understanding whether the fraction of combinations that reduce the weight with respect to that of the original Pauli measurement is small or large at each step. Taking one of the RQCs of Table~\ref{tab: Large RQCs 1 shot each} with $t=60$, we looked into the weight distribution at each step, evaluating what percentage of combinations have a weight that is smaller than the one obtained in the absence of the greedy algorithm. Figures~\ref{fig: Histograms -- go=1} and \ref{fig: Histograms -- go=2} suggest that the majority of combinations tested by the greedy algorithm increase the weight compared with that obtained naturally from the PBC procedure (in the absence of greedy). These results indicate that stopping the greedy ahead of time (say halfway through the full set of combinations) might significantly hinder the overall performance of the algorithm as presented in Table~\ref{tab: Large RQCs 1 shot each}.

\subsection{Randomizing the subsets}

Another relevant line of inquiry is whether alternative formulations of the greedy algorithm might exist that outperform Algorithm~\ref{alg: Greedy algorithm}. As we have seen, the latter works by searching for Pauli measurements with better weight among all the (sub)sets $\mathcal{W} \subseteq \{1,\,\dots,\,r-1\}$ with size $r-1-a$ and $a$, with $0\leq a\leq \mathsf{go} \leq (r-1)/2$. Rather than doing so, one could search for Pauli measurements with smaller weight by picking the size of each subset and its elements uniformly at random. Evidently, for a fair comparison with the previous formulation, the same number of subsets as associated with a fixed $\mathsf{go}$ needs to be used. This alternative procedure randomizes the search for Pauli operators with smaller weight compared with the more structured search described in Sec.~\ref{sec: Greedy algorithm}. The goal is to understand if one reaps benefits from potentially allowing a larger number of Pauli operators to be combined to yield a new Pauli measurement.

Our intuition was that the original (more structured) formulation should be beneficial in the context of more structured quantum circuits such as the HSCs. This was confirmed by running this randomized version of the greedy algorithm in the smaller HSCs ($n=10$, $t=14$), allowing at each step the same number of combinations as used with the structured version set with $\mathsf{go}=1$. The improvements obtained were roughly reduced in half by using the randomized version of the algorithm instead of the more structured approach.

For the small RQCs and a number of combinations corresponding to those used when setting $\mathsf{go}=1$, the differences in performance are not as striking. In these circuits, the more structured algorithm outperforms the randomized version by attaining reductions of the average weight that are roughly 10\% larger than the reductions obtained by the randomized version of the algorithm. This difference increases as $t$ becomes larger. For instance, for the five circuit instances with $t=100$ presented in Table~\ref{tab: Large RQCs 1 shot each}, the improvements attained by the randomized-search greedy algorithm are between 52\% and 73\% lower than those obtained by the structured search version of the algorithm. Intriguingly, for $\mathsf{go}=2$ the difference in performance is slightly diminished and, for the latter circuits, the improvements attained by the randomized-search greedy algorithm are between 27\% and 38\% smaller than those obtained with the structured search approach.


\begin{thebibliography}{51}%
	\makeatletter
	\providecommand \@ifxundefined [1]{%
		\@ifx{#1\undefined}
	}%
	\providecommand \@ifnum [1]{%
		\ifnum #1\expandafter \@firstoftwo
		\else \expandafter \@secondoftwo
		\fi
	}%
	\providecommand \@ifx [1]{%
		\ifx #1\expandafter \@firstoftwo
		\else \expandafter \@secondoftwo
		\fi
	}%
	\providecommand \natexlab [1]{#1}%
	\providecommand \enquote  [1]{``#1''}%
	\providecommand \bibnamefont  [1]{#1}%
	\providecommand \bibfnamefont [1]{#1}%
	\providecommand \citenamefont [1]{#1}%
	\providecommand \href@noop [0]{\@secondoftwo}%
	\providecommand \href [0]{\begingroup \@sanitize@url \@href}%
	\providecommand \@href[1]{\@@startlink{#1}\@@href}%
	\providecommand \@@href[1]{\endgroup#1\@@endlink}%
	\providecommand \@sanitize@url [0]{\catcode `\\12\catcode `\$12\catcode `\&12\catcode `\#12\catcode `\^12\catcode `\_12\catcode `\%12\relax}%
	\providecommand \@@startlink[1]{}%
	\providecommand \@@endlink[0]{}%
	\providecommand \url  [0]{\begingroup\@sanitize@url \@url }%
	\providecommand \@url [1]{\endgroup\@href {#1}{\urlprefix }}%
	\providecommand \urlprefix  [0]{URL }%
	\providecommand \Eprint [0]{\href }%
	\providecommand \doibase [0]{https://doi.org/}%
	\providecommand \selectlanguage [0]{\@gobble}%
	\providecommand \bibinfo  [0]{\@secondoftwo}%
	\providecommand \bibfield  [0]{\@secondoftwo}%
	\providecommand \translation [1]{[#1]}%
	\providecommand \BibitemOpen [0]{}%
	\providecommand \bibitemStop [0]{}%
	\providecommand \bibitemNoStop [0]{.\EOS\space}%
	\providecommand \EOS [0]{\spacefactor3000\relax}%
	\providecommand \BibitemShut  [1]{\csname bibitem#1\endcsname}%
	\let\auto@bib@innerbib\@empty
	\bibitem [{\citenamefont {Maslov}\ \emph {et~al.}(2008)\citenamefont {Maslov}, \citenamefont {Dueck}, \citenamefont {Miller},\ and\ \citenamefont {Negrevergne}}]{Maslov2008templates}%
	\BibitemOpen
	\bibfield  {author} {\bibinfo {author} {\bibfnamefont {D.}~\bibnamefont {Maslov}}, \bibinfo {author} {\bibfnamefont {G.~W.}\ \bibnamefont {Dueck}}, \bibinfo {author} {\bibfnamefont {D.~M.}\ \bibnamefont {Miller}},\ and\ \bibinfo {author} {\bibfnamefont {C.}~\bibnamefont {Negrevergne}},\ }\bibfield  {title} {\bibinfo {title} {{Quantum Circuit Simplification and Level Compaction}},\ }\href {https://doi.org/10.1109/TCAD.2007.911334} {\bibfield  {journal} {\bibinfo  {journal} {IEEE Transactions on Computer-Aided Design of Integrated Circuits and Systems}\ }\textbf {\bibinfo {volume} {27}},\ \bibinfo {pages} {436} (\bibinfo {year} {2008})}\BibitemShut {NoStop}%
	\bibitem [{\citenamefont {Iten}\ \emph {et~al.}(2022)\citenamefont {Iten}, \citenamefont {Moyard}, \citenamefont {Metger}, \citenamefont {Sutter},\ and\ \citenamefont {Woerner}}]{ItenMMSW2022pm}%
	\BibitemOpen
	\bibfield  {author} {\bibinfo {author} {\bibfnamefont {R.}~\bibnamefont {Iten}}, \bibinfo {author} {\bibfnamefont {R.}~\bibnamefont {Moyard}}, \bibinfo {author} {\bibfnamefont {T.}~\bibnamefont {Metger}}, \bibinfo {author} {\bibfnamefont {D.}~\bibnamefont {Sutter}},\ and\ \bibinfo {author} {\bibfnamefont {S.}~\bibnamefont {Woerner}},\ }\bibfield  {title} {\bibinfo {title} {{Exact and Practical Pattern Matching for Quantum Circuit Optimization}},\ }\href {https://doi.org/10.1145/3498325} {\bibfield  {journal} {\bibinfo  {journal} {ACM Transactions on Quantum Computing}\ }\textbf {\bibinfo {volume} {3}},\ \bibinfo {pages} {1} (\bibinfo {year} {2022})}\BibitemShut {NoStop}%
	\bibitem [{\citenamefont {Duncan}\ \emph {et~al.}(2020)\citenamefont {Duncan}, \citenamefont {Kissinger}, \citenamefont {Perdrix},\ and\ \citenamefont {van~de Wetering}}]{Duncan2020simplicationzx}%
	\BibitemOpen
	\bibfield  {author} {\bibinfo {author} {\bibfnamefont {R.}~\bibnamefont {Duncan}}, \bibinfo {author} {\bibfnamefont {A.}~\bibnamefont {Kissinger}}, \bibinfo {author} {\bibfnamefont {S.}~\bibnamefont {Perdrix}},\ and\ \bibinfo {author} {\bibfnamefont {J.}~\bibnamefont {van~de Wetering}},\ }\bibfield  {title} {\bibinfo {title} {Graph-theoretic {S}implification of {Q}uantum {C}ircuits with the {ZX}-calculus},\ }\href {https://doi.org/10.22331/q-2020-06-04-279} {\bibfield  {journal} {\bibinfo  {journal} {{Quantum}}\ }\textbf {\bibinfo {volume} {4}},\ \bibinfo {pages} {279} (\bibinfo {year} {2020})}\BibitemShut {NoStop}%
	\bibitem [{\citenamefont {Kissinger}\ and\ \citenamefont {van~de Wetering}(2020)}]{KissWet2020reduceT}%
	\BibitemOpen
	\bibfield  {author} {\bibinfo {author} {\bibfnamefont {A.}~\bibnamefont {Kissinger}}\ and\ \bibinfo {author} {\bibfnamefont {J.}~\bibnamefont {van~de Wetering}},\ }\bibfield  {title} {\bibinfo {title} {{Reducing the number of non-Clifford gates in quantum circuits}},\ }\href {https://doi.org/10.1103/PhysRevA.102.022406} {\bibfield  {journal} {\bibinfo  {journal} {Phys. Rev. A}\ }\textbf {\bibinfo {volume} {102}},\ \bibinfo {pages} {022406} (\bibinfo {year} {2020})}\BibitemShut {NoStop}%
	\bibitem [{\citenamefont {Amy}\ \emph {et~al.}(2014)\citenamefont {Amy}, \citenamefont {Maslov},\ and\ \citenamefont {Mosca}}]{AmyMM2014}%
	\BibitemOpen
	\bibfield  {author} {\bibinfo {author} {\bibfnamefont {M.}~\bibnamefont {Amy}}, \bibinfo {author} {\bibfnamefont {D.}~\bibnamefont {Maslov}},\ and\ \bibinfo {author} {\bibfnamefont {M.}~\bibnamefont {Mosca}},\ }\bibfield  {title} {\bibinfo {title} {{Polynomial-Time T-Depth Optimization of Clifford+T Circuits Via Matroid Partitioning}},\ }\href {https://doi.org/10.1109/TCAD.2014.2341953} {\bibfield  {journal} {\bibinfo  {journal} {IEEE Transactions on Computer-Aided Design of Integrated Circuits and Systems}\ }\textbf {\bibinfo {volume} {33}},\ \bibinfo {pages} {1476} (\bibinfo {year} {2014})}\BibitemShut {NoStop}%
	\bibitem [{\citenamefont {Heyfron}\ and\ \citenamefont {Campbell}(2018)}]{Heyfron2018}%
	\BibitemOpen
	\bibfield  {author} {\bibinfo {author} {\bibfnamefont {L.~E.}\ \bibnamefont {Heyfron}}\ and\ \bibinfo {author} {\bibfnamefont {E.~T.}\ \bibnamefont {Campbell}},\ }\bibfield  {title} {\bibinfo {title} {An efficient quantum compiler that reduces {T} count},\ }\href {https://doi.org/10.1088/2058-9565/aad604} {\bibfield  {journal} {\bibinfo  {journal} {Quantum Science and Technology}\ }\textbf {\bibinfo {volume} {4}},\ \bibinfo {pages} {015004} (\bibinfo {year} {2018})}\BibitemShut {NoStop}%
	\bibitem [{\citenamefont {Amy}\ and\ \citenamefont {Mosca}(2019)}]{Amy2019}%
	\BibitemOpen
	\bibfield  {author} {\bibinfo {author} {\bibfnamefont {M.}~\bibnamefont {Amy}}\ and\ \bibinfo {author} {\bibfnamefont {M.}~\bibnamefont {Mosca}},\ }\bibfield  {title} {\bibinfo {title} {{T-Count Optimization and Reed–Muller Codes}},\ }\href {https://doi.org/10.1109/tit.2019.2906374} {\bibfield  {journal} {\bibinfo  {journal} {IEEE Transactions on Information Theory}\ }\textbf {\bibinfo {volume} {65}},\ \bibinfo {pages} {4771–4784} (\bibinfo {year} {2019})}\BibitemShut {NoStop}%
	\bibitem [{\citenamefont {Nielsen}\ and\ \citenamefont {Chuang}(2011)}]{NielsenChuang}%
	\BibitemOpen
	\bibfield  {author} {\bibinfo {author} {\bibfnamefont {M.~A.}\ \bibnamefont {Nielsen}}\ and\ \bibinfo {author} {\bibfnamefont {I.~L.}\ \bibnamefont {Chuang}},\ }\href@noop {} {\emph {\bibinfo {title} {{Quantum Computation and Quantum Information: 10th Anniversary Edition}}}},\ \bibinfo {edition} {10th}\ ed.\ (\bibinfo  {publisher} {Cambridge University Press},\ \bibinfo {address} {USA},\ \bibinfo {year} {2011})\BibitemShut {NoStop}%
	\bibitem [{\citenamefont {Chi-Chih~Yao}(1993)}]{Yao1993}%
	\BibitemOpen
	\bibfield  {author} {\bibinfo {author} {\bibfnamefont {A.}~\bibnamefont {Chi-Chih~Yao}},\ }\bibfield  {title} {\bibinfo {title} {Quantum circuit complexity},\ }in\ \href {https://doi.org/10.1109/SFCS.1993.366852} {\emph {\bibinfo {booktitle} {Proceedings of 1993 IEEE 34th Annual Foundations of Computer Science}}}\ (\bibinfo {year} {1993})\ pp.\ \bibinfo {pages} {352--361}\BibitemShut {NoStop}%
	\bibitem [{\citenamefont {Bravyi}\ \emph {et~al.}(2016)\citenamefont {Bravyi}, \citenamefont {Smith},\ and\ \citenamefont {Smolin}}]{BSS2016}%
	\BibitemOpen
	\bibfield  {author} {\bibinfo {author} {\bibfnamefont {S.}~\bibnamefont {Bravyi}}, \bibinfo {author} {\bibfnamefont {G.}~\bibnamefont {Smith}},\ and\ \bibinfo {author} {\bibfnamefont {J.~A.}\ \bibnamefont {Smolin}},\ }\bibfield  {title} {\bibinfo {title} {{Trading Classical and Quantum Computational Resources}},\ }\href {https://doi.org/10.1103/PhysRevX.6.021043} {\bibfield  {journal} {\bibinfo  {journal} {Phys. Rev. X}\ }\textbf {\bibinfo {volume} {6}},\ \bibinfo {pages} {021043} (\bibinfo {year} {2016})}\BibitemShut {NoStop}%
	\bibitem [{\citenamefont {Peres}\ and\ \citenamefont {Galv{\~{a}}o}(2023)}]{PeresGa2023}%
	\BibitemOpen
	\bibfield  {author} {\bibinfo {author} {\bibfnamefont {F.~C.~R.}\ \bibnamefont {Peres}}\ and\ \bibinfo {author} {\bibfnamefont {E.~F.}\ \bibnamefont {Galv{\~{a}}o}},\ }\bibfield  {title} {\bibinfo {title} {Quantum circuit compilation and hybrid computation using {P}auli-based computation},\ }\href {https://doi.org/10.22331/q-2023-10-03-1126} {\bibfield  {journal} {\bibinfo  {journal} {{Quantum}}\ }\textbf {\bibinfo {volume} {7}},\ \bibinfo {pages} {1126} (\bibinfo {year} {2023})}\BibitemShut {NoStop}%
	\bibitem [{\citenamefont {Gottesman}(1998)}]{Gottesman1998}%
	\BibitemOpen
	\bibfield  {author} {\bibinfo {author} {\bibfnamefont {D.}~\bibnamefont {Gottesman}},\ }\bibfield  {title} {\bibinfo {title} {The {H}eisenberg {R}epresentation of {Q}uantum {C}omputers},\ }in\ \href@noop {} {\emph {\bibinfo {booktitle} {Group22: Proceedings of the XXII International Colloquium on Group Theoretical Methods in Physics}}}\ (\bibinfo {year} {1998})\ pp.\ \bibinfo {pages} {32--43},\ \Eprint {https://arxiv.org/abs/quant-ph/9807006} {arXiv:quant-ph/9807006 [quant-ph]} \BibitemShut {NoStop}%
	\bibitem [{\citenamefont {Gottesman}(1997)}]{PhDGottesman}%
	\BibitemOpen
	\bibfield  {author} {\bibinfo {author} {\bibfnamefont {D.}~\bibnamefont {Gottesman}},\ }\emph {\bibinfo {title} {Stabilizer {C}odes and {Q}uantum {E}rror {C}orrection}},\ \href@noop {} {Ph.D. thesis},\ \bibinfo  {school} {Caltech} (\bibinfo {year} {1997}),\ \Eprint {https://arxiv.org/abs/quant-ph/9705052} {arXiv:quant-ph/9705052 [quant-ph]} \BibitemShut {NoStop}%
	\bibitem [{\citenamefont {Calderbank}\ \emph {et~al.}(1997)\citenamefont {Calderbank}, \citenamefont {Rains}, \citenamefont {Shor},\ and\ \citenamefont {Sloane}}]{CalderbankRSS1997qec}%
	\BibitemOpen
	\bibfield  {author} {\bibinfo {author} {\bibfnamefont {A.~R.}\ \bibnamefont {Calderbank}}, \bibinfo {author} {\bibfnamefont {E.~M.}\ \bibnamefont {Rains}}, \bibinfo {author} {\bibfnamefont {P.~W.}\ \bibnamefont {Shor}},\ and\ \bibinfo {author} {\bibfnamefont {N.~J.~A.}\ \bibnamefont {Sloane}},\ }\bibfield  {title} {\bibinfo {title} {Quantum error correction and orthogonal geometry},\ }\href {https://doi.org/10.1103/PhysRevLett.78.405} {\bibfield  {journal} {\bibinfo  {journal} {Phys. Rev. Lett.}\ }\textbf {\bibinfo {volume} {78}},\ \bibinfo {pages} {405} (\bibinfo {year} {1997})}\BibitemShut {NoStop}%
	\bibitem [{\citenamefont {den Nest}\ \emph {et~al.}(2007)\citenamefont {den Nest}, \citenamefont {Dür}, \citenamefont {Miyake},\ and\ \citenamefont {Briegel}}]{NestDMB2007}%
	\BibitemOpen
	\bibfield  {author} {\bibinfo {author} {\bibfnamefont {M.~V.}\ \bibnamefont {den Nest}}, \bibinfo {author} {\bibfnamefont {W.}~\bibnamefont {Dür}}, \bibinfo {author} {\bibfnamefont {A.}~\bibnamefont {Miyake}},\ and\ \bibinfo {author} {\bibfnamefont {H.~J.}\ \bibnamefont {Briegel}},\ }\bibfield  {title} {\bibinfo {title} {Fundamentals of universality in one-way quantum computation},\ }\href {https://doi.org/10.1088/1367-2630/9/6/204} {\bibfield  {journal} {\bibinfo  {journal} {New Journal of Physics}\ }\textbf {\bibinfo {volume} {9}},\ \bibinfo {pages} {204} (\bibinfo {year} {2007})}\BibitemShut {NoStop}%
	\bibitem [{\citenamefont {Aharonov}(2003)}]{Aharonov2003univ}%
	\BibitemOpen
	\bibfield  {author} {\bibinfo {author} {\bibfnamefont {D.}~\bibnamefont {Aharonov}},\ }\href@noop {} {\bibinfo {title} {{A Simple Proof that Toffoli and Hadamard are Quantum Universal}}} (\bibinfo {year} {2003}),\ \Eprint {https://arxiv.org/abs/quant-ph/0301040} {arXiv:quant-ph/0301040 [quant-ph]} \BibitemShut {NoStop}%
	\bibitem [{\citenamefont {Boykin}\ \emph {et~al.}(1999)\citenamefont {Boykin}, \citenamefont {Mor}, \citenamefont {Pulver}, \citenamefont {Roychowdhury},\ and\ \citenamefont {Vatan}}]{Boykin1999clifford+t}%
	\BibitemOpen
	\bibfield  {author} {\bibinfo {author} {\bibfnamefont {P.}~\bibnamefont {Boykin}}, \bibinfo {author} {\bibfnamefont {T.}~\bibnamefont {Mor}}, \bibinfo {author} {\bibfnamefont {M.}~\bibnamefont {Pulver}}, \bibinfo {author} {\bibfnamefont {V.}~\bibnamefont {Roychowdhury}},\ and\ \bibinfo {author} {\bibfnamefont {F.}~\bibnamefont {Vatan}},\ }\bibfield  {title} {\bibinfo {title} {On universal and fault-tolerant quantum computing: a novel basis and a new constructive proof of universality for {S}hor's basis},\ }in\ \href {https://doi.org/10.1109/SFFCS.1999.814621} {\emph {\bibinfo {booktitle} {40th Annual Symposium on Foundations of Computer Science (Cat. No.99CB37039)}}}\ (\bibinfo {year} {1999})\ pp.\ \bibinfo {pages} {486--494}\BibitemShut {NoStop}%
	\bibitem [{\citenamefont {Barenco}\ \emph {et~al.}(1995)\citenamefont {Barenco}, \citenamefont {Bennett}, \citenamefont {Cleve}, \citenamefont {DiVincenzo}, \citenamefont {Margolus}, \citenamefont {Shor}, \citenamefont {Sleator}, \citenamefont {Smolin},\ and\ \citenamefont {Weinfurter}}]{Barenco+1995}%
	\BibitemOpen
	\bibfield  {author} {\bibinfo {author} {\bibfnamefont {A.}~\bibnamefont {Barenco}}, \bibinfo {author} {\bibfnamefont {C.~H.}\ \bibnamefont {Bennett}}, \bibinfo {author} {\bibfnamefont {R.}~\bibnamefont {Cleve}}, \bibinfo {author} {\bibfnamefont {D.~P.}\ \bibnamefont {DiVincenzo}}, \bibinfo {author} {\bibfnamefont {N.}~\bibnamefont {Margolus}}, \bibinfo {author} {\bibfnamefont {P.}~\bibnamefont {Shor}}, \bibinfo {author} {\bibfnamefont {T.}~\bibnamefont {Sleator}}, \bibinfo {author} {\bibfnamefont {J.~A.}\ \bibnamefont {Smolin}},\ and\ \bibinfo {author} {\bibfnamefont {H.}~\bibnamefont {Weinfurter}},\ }\bibfield  {title} {\bibinfo {title} {Elementary gates for quantum computation},\ }\href {https://doi.org/10.1103/PhysRevA.52.3457} {\bibfield  {journal} {\bibinfo  {journal} {Phys. Rev. A}\ }\textbf {\bibinfo {volume} {52}},\ \bibinfo {pages} {3457} (\bibinfo {year} {1995})}\BibitemShut {NoStop}%
	\bibitem [{\citenamefont {Shi}(2002)}]{shi2002toffoli}%
	\BibitemOpen
	\bibfield  {author} {\bibinfo {author} {\bibfnamefont {Y.}~\bibnamefont {Shi}},\ }\href@noop {} {\bibinfo {title} {Both toffoli and controlled-not need little help to do universal quantum computation}} (\bibinfo {year} {2002}),\ \Eprint {https://arxiv.org/abs/quant-ph/0205115} {arXiv:quant-ph/0205115 [quant-ph]} \BibitemShut {NoStop}%
	\bibitem [{\citenamefont {Bravyi}\ and\ \citenamefont {Kitaev}(2005)}]{BravKit2005}%
	\BibitemOpen
	\bibfield  {author} {\bibinfo {author} {\bibfnamefont {S.}~\bibnamefont {Bravyi}}\ and\ \bibinfo {author} {\bibfnamefont {A.}~\bibnamefont {Kitaev}},\ }\bibfield  {title} {\bibinfo {title} {Universal quantum computation with ideal {C}lifford gates and noisy ancillas},\ }\href {https://doi.org/10.1103/PhysRevA.71.022316} {\bibfield  {journal} {\bibinfo  {journal} {Phys. Rev. A}\ }\textbf {\bibinfo {volume} {71}},\ \bibinfo {pages} {022316} (\bibinfo {year} {2005})}\BibitemShut {NoStop}%
	\bibitem [{\citenamefont {Eastin}\ and\ \citenamefont {Knill}(2009)}]{EastinKnill2009}%
	\BibitemOpen
	\bibfield  {author} {\bibinfo {author} {\bibfnamefont {B.}~\bibnamefont {Eastin}}\ and\ \bibinfo {author} {\bibfnamefont {E.}~\bibnamefont {Knill}},\ }\bibfield  {title} {\bibinfo {title} {Restrictions on transversal encoded quantum gate sets},\ }\href {https://doi.org/10.1103/PhysRevLett.102.110502} {\bibfield  {journal} {\bibinfo  {journal} {Phys. Rev. Lett.}\ }\textbf {\bibinfo {volume} {102}},\ \bibinfo {pages} {110502} (\bibinfo {year} {2009})}\BibitemShut {NoStop}%
	\bibitem [{\citenamefont {Raussendorf}\ and\ \citenamefont {Briegel}(2001)}]{RaussBrie2001}%
	\BibitemOpen
	\bibfield  {author} {\bibinfo {author} {\bibfnamefont {R.}~\bibnamefont {Raussendorf}}\ and\ \bibinfo {author} {\bibfnamefont {H.~J.}\ \bibnamefont {Briegel}},\ }\bibfield  {title} {\bibinfo {title} {{A One-Way Quantum Computer}},\ }\href {https://doi.org/10.1103/PhysRevLett.86.5188} {\bibfield  {journal} {\bibinfo  {journal} {Phys. Rev. Lett.}\ }\textbf {\bibinfo {volume} {86}},\ \bibinfo {pages} {5188} (\bibinfo {year} {2001})}\BibitemShut {NoStop}%
	\bibitem [{\citenamefont {Hein}\ \emph {et~al.}(2004)\citenamefont {Hein}, \citenamefont {Eisert},\ and\ \citenamefont {Briegel}}]{HeinEB2004graphstates}%
	\BibitemOpen
	\bibfield  {author} {\bibinfo {author} {\bibfnamefont {M.}~\bibnamefont {Hein}}, \bibinfo {author} {\bibfnamefont {J.}~\bibnamefont {Eisert}},\ and\ \bibinfo {author} {\bibfnamefont {H.~J.}\ \bibnamefont {Briegel}},\ }\bibfield  {title} {\bibinfo {title} {Multiparty entanglement in graph states},\ }\href {https://doi.org/10.1103/PhysRevA.69.062311} {\bibfield  {journal} {\bibinfo  {journal} {Phys. Rev. A}\ }\textbf {\bibinfo {volume} {69}},\ \bibinfo {pages} {062311} (\bibinfo {year} {2004})}\BibitemShut {NoStop}%
	\bibitem [{\citenamefont {Van~den Nest}\ \emph {et~al.}(2006)\citenamefont {Van~den Nest}, \citenamefont {Miyake}, \citenamefont {D\"ur},\ and\ \citenamefont {Briegel}}]{Nest2006}%
	\BibitemOpen
	\bibfield  {author} {\bibinfo {author} {\bibfnamefont {M.}~\bibnamefont {Van~den Nest}}, \bibinfo {author} {\bibfnamefont {A.}~\bibnamefont {Miyake}}, \bibinfo {author} {\bibfnamefont {W.}~\bibnamefont {D\"ur}},\ and\ \bibinfo {author} {\bibfnamefont {H.~J.}\ \bibnamefont {Briegel}},\ }\bibfield  {title} {\bibinfo {title} {{Universal Resources for Measurement-Based Quantum Computation}},\ }\href {https://doi.org/10.1103/PhysRevLett.97.150504} {\bibfield  {journal} {\bibinfo  {journal} {Phys. Rev. Lett.}\ }\textbf {\bibinfo {volume} {97}},\ \bibinfo {pages} {150504} (\bibinfo {year} {2006})}\BibitemShut {NoStop}%
	\bibitem [{\citenamefont {Broadbent}\ \emph {et~al.}(2009)\citenamefont {Broadbent}, \citenamefont {Fitzsimons},\ and\ \citenamefont {Kashefi}}]{BroadbentFK2009blind}%
	\BibitemOpen
	\bibfield  {author} {\bibinfo {author} {\bibfnamefont {A.}~\bibnamefont {Broadbent}}, \bibinfo {author} {\bibfnamefont {J.}~\bibnamefont {Fitzsimons}},\ and\ \bibinfo {author} {\bibfnamefont {E.}~\bibnamefont {Kashefi}},\ }\bibfield  {title} {\bibinfo {title} {{U}niversal {B}lind {Q}uantum {C}omputation},\ }in\ \href {https://doi.org/10.1109/FOCS.2009.36} {\emph {\bibinfo {booktitle} {2009 50th Annual IEEE Symposium on Foundations of Computer Science}}}\ (\bibinfo {year} {2009})\ pp.\ \bibinfo {pages} {517--526}\BibitemShut {NoStop}%
	\bibitem [{\citenamefont {Mhalla}\ and\ \citenamefont {Perdrix}(2012)}]{Mhalla2012XZuni}%
	\BibitemOpen
	\bibfield  {author} {\bibinfo {author} {\bibfnamefont {M.}~\bibnamefont {Mhalla}}\ and\ \bibinfo {author} {\bibfnamefont {S.}~\bibnamefont {Perdrix}},\ }\bibfield  {title} {\bibinfo {title} {{Graph States, Pivot Minor, and Universality of (X,Z)-Measurements}},\ }\href {https://api.semanticscholar.org/CorpusID:43889959} {\bibfield  {journal} {\bibinfo  {journal} {Int. J. Unconv. Comput.}\ }\textbf {\bibinfo {volume} {9}},\ \bibinfo {pages} {153} (\bibinfo {year} {2012})}\BibitemShut {NoStop}%
	\bibitem [{\citenamefont {Mantri}\ \emph {et~al.}(2017)\citenamefont {Mantri}, \citenamefont {Demarie},\ and\ \citenamefont {Fitzsimons}}]{Mantri2017xycluster}%
	\BibitemOpen
	\bibfield  {author} {\bibinfo {author} {\bibfnamefont {A.}~\bibnamefont {Mantri}}, \bibinfo {author} {\bibfnamefont {T.~F.}\ \bibnamefont {Demarie}},\ and\ \bibinfo {author} {\bibfnamefont {J.~F.}\ \bibnamefont {Fitzsimons}},\ }\bibfield  {title} {\bibinfo {title} {Universality of quantum computation with cluster states and (x, y)-plane measurements},\ }\href {https://doi.org/10.1038/srep42861} {\bibfield  {journal} {\bibinfo  {journal} {Scientific Reports}\ }\textbf {\bibinfo {volume} {7}},\ \bibinfo {pages} {42861} (\bibinfo {year} {2017})}\BibitemShut {NoStop}%
	\bibitem [{\citenamefont {Takeuchi}\ \emph {et~al.}(2019)\citenamefont {Takeuchi}, \citenamefont {Morimae},\ and\ \citenamefont {Hayashi}}]{Takeuchi2019hypergraph}%
	\BibitemOpen
	\bibfield  {author} {\bibinfo {author} {\bibfnamefont {Y.}~\bibnamefont {Takeuchi}}, \bibinfo {author} {\bibfnamefont {T.}~\bibnamefont {Morimae}},\ and\ \bibinfo {author} {\bibfnamefont {M.}~\bibnamefont {Hayashi}},\ }\bibfield  {title} {\bibinfo {title} {Quantum computational universality of hypergraph states with pauli-{X} and {Z} basis measurements},\ }\href {https://doi.org/10.1038/s41598-019-49968-3} {\bibfield  {journal} {\bibinfo  {journal} {Scientific Reports}\ }\textbf {\bibinfo {volume} {9}},\ \bibinfo {pages} {13585} (\bibinfo {year} {2019})}\BibitemShut {NoStop}%
	\bibitem [{\citenamefont {Raussendorf}\ \emph {et~al.}(2003)\citenamefont {Raussendorf}, \citenamefont {Browne},\ and\ \citenamefont {Briegel}}]{RaussBrBr2003}%
	\BibitemOpen
	\bibfield  {author} {\bibinfo {author} {\bibfnamefont {R.}~\bibnamefont {Raussendorf}}, \bibinfo {author} {\bibfnamefont {D.~E.}\ \bibnamefont {Browne}},\ and\ \bibinfo {author} {\bibfnamefont {H.~J.}\ \bibnamefont {Briegel}},\ }\bibfield  {title} {\bibinfo {title} {Measurement-based quantum computation on cluster states},\ }\href {https://doi.org/10.1103/PhysRevA.68.022312} {\bibfield  {journal} {\bibinfo  {journal} {Phys. Rev. A}\ }\textbf {\bibinfo {volume} {68}},\ \bibinfo {pages} {022312} (\bibinfo {year} {2003})}\BibitemShut {NoStop}%
	\bibitem [{\citenamefont {Miller}\ and\ \citenamefont {Miyake}(2016)}]{MillerMiyake2016hypergraph}%
	\BibitemOpen
	\bibfield  {author} {\bibinfo {author} {\bibfnamefont {J.}~\bibnamefont {Miller}}\ and\ \bibinfo {author} {\bibfnamefont {A.}~\bibnamefont {Miyake}},\ }\bibfield  {title} {\bibinfo {title} {Hierarchy of universal entanglement in {2D} measurement-based quantum computation},\ }\href {https://doi.org/10.1038/npjqi.2016.36} {\bibfield  {journal} {\bibinfo  {journal} {npj Quantum Information}\ }\textbf {\bibinfo {volume} {2}},\ \bibinfo {pages} {16036} (\bibinfo {year} {2016})}\BibitemShut {NoStop}%
	\bibitem [{\citenamefont {Takeuchi}(2023)}]{takeuchi2023catalytic}%
	\BibitemOpen
	\bibfield  {author} {\bibinfo {author} {\bibfnamefont {Y.}~\bibnamefont {Takeuchi}},\ }\href@noop {} {\bibinfo {title} {{Catalytic Transformation from Computationally-Universal to Strictly-Universal Measurement-Based Quantum Computation}}} (\bibinfo {year} {2023}),\ \Eprint {https://arxiv.org/abs/2312.16433} {arXiv:2312.16433 [quant-ph]} \BibitemShut {NoStop}%
	\bibitem [{\citenamefont {Nielsen}(2003)}]{Nielsen2003}%
	\BibitemOpen
	\bibfield  {author} {\bibinfo {author} {\bibfnamefont {M.~A.}\ \bibnamefont {Nielsen}},\ }\bibfield  {title} {\bibinfo {title} {Quantum computation by measurement and quantum memory},\ }\href {https://doi.org/https://doi.org/10.1016/S0375-9601(02)01803-0} {\bibfield  {journal} {\bibinfo  {journal} {Physics Letters A}\ }\textbf {\bibinfo {volume} {308}},\ \bibinfo {pages} {96} (\bibinfo {year} {2003})}\BibitemShut {NoStop}%
	\bibitem [{\citenamefont {Leung}(2004)}]{Leung2004}%
	\BibitemOpen
	\bibfield  {author} {\bibinfo {author} {\bibfnamefont {D.~W.}\ \bibnamefont {Leung}},\ }\bibfield  {title} {\bibinfo {title} {Quantum computation by measurements},\ }\href {https://doi.org/10.1142/S0219749904000055} {\bibfield  {journal} {\bibinfo  {journal} {International Journal of Quantum Information}\ }\textbf {\bibinfo {volume} {02}},\ \bibinfo {pages} {33} (\bibinfo {year} {2004})}\BibitemShut {NoStop}%
	\bibitem [{\citenamefont {Gottesman}\ and\ \citenamefont {Chuang}(1999)}]{GottChuang1999}%
	\BibitemOpen
	\bibfield  {author} {\bibinfo {author} {\bibfnamefont {D.}~\bibnamefont {Gottesman}}\ and\ \bibinfo {author} {\bibfnamefont {I.~L.}\ \bibnamefont {Chuang}},\ }\bibfield  {title} {\bibinfo {title} {Demonstrating the viability of universal quantum computation using teleportation and single-qubit operations},\ }\href {https://doi.org/10.1038/46503} {\bibfield  {journal} {\bibinfo  {journal} {Nature}\ }\textbf {\bibinfo {volume} {402}},\ \bibinfo {pages} {390} (\bibinfo {year} {1999})}\BibitemShut {NoStop}%
	\bibitem [{\citenamefont {Perdrix}(2005)}]{Perdrix2005}%
	\BibitemOpen
	\bibfield  {author} {\bibinfo {author} {\bibfnamefont {S.}~\bibnamefont {Perdrix}},\ }\bibfield  {title} {\bibinfo {title} {State transfer instead of teleportation in measurement-based quantum computation},\ }\href {https://doi.org/10.1142/S0219749905000785} {\bibfield  {journal} {\bibinfo  {journal} {International Journal of Quantum Information}\ }\textbf {\bibinfo {volume} {03}},\ \bibinfo {pages} {219} (\bibinfo {year} {2005})}\BibitemShut {NoStop}%
	\bibitem [{\citenamefont {Aliferis}\ and\ \citenamefont {Leung}(2004)}]{Leung2004unify}%
	\BibitemOpen
	\bibfield  {author} {\bibinfo {author} {\bibfnamefont {P.}~\bibnamefont {Aliferis}}\ and\ \bibinfo {author} {\bibfnamefont {D.~W.}\ \bibnamefont {Leung}},\ }\bibfield  {title} {\bibinfo {title} {Computation by measurements: {A} unifying picture},\ }\href {https://doi.org/10.1103/PhysRevA.70.062314} {\bibfield  {journal} {\bibinfo  {journal} {Phys. Rev. A}\ }\textbf {\bibinfo {volume} {70}},\ \bibinfo {pages} {062314} (\bibinfo {year} {2004})}\BibitemShut {NoStop}%
	\bibitem [{\citenamefont {Jorrand}\ and\ \citenamefont {Perdrix}(2005)}]{Perdrix2005unify}%
	\BibitemOpen
	\bibfield  {author} {\bibinfo {author} {\bibfnamefont {P.}~\bibnamefont {Jorrand}}\ and\ \bibinfo {author} {\bibfnamefont {S.}~\bibnamefont {Perdrix}},\ }\bibfield  {title} {\bibinfo {title} {{Unifying quantum computation with projective measurements only and one-way quantum computation}},\ }in\ \href {https://doi.org/10.1117/12.620302} {\emph {\bibinfo {booktitle} {Quantum Informatics 2004}}},\ Vol.\ \bibinfo {volume} {5833},\ \bibinfo {editor} {edited by\ \bibinfo {editor} {\bibfnamefont {Y.~I.}\ \bibnamefont {Ozhigov}}},\ \bibinfo {organization} {International Society for Optics and Photonics}\ (\bibinfo  {publisher} {SPIE},\ \bibinfo {year} {2005})\ pp.\ \bibinfo {pages} {44 -- 51}\BibitemShut {NoStop}%
	\bibitem [{\citenamefont {Childs}\ \emph {et~al.}(2005)\citenamefont {Childs}, \citenamefont {Leung},\ and\ \citenamefont {Nielsen}}]{ChildsLeung2005unify}%
	\BibitemOpen
	\bibfield  {author} {\bibinfo {author} {\bibfnamefont {A.~M.}\ \bibnamefont {Childs}}, \bibinfo {author} {\bibfnamefont {D.~W.}\ \bibnamefont {Leung}},\ and\ \bibinfo {author} {\bibfnamefont {M.~A.}\ \bibnamefont {Nielsen}},\ }\bibfield  {title} {\bibinfo {title} {Unified derivations of measurement-based schemes for quantum computation},\ }\href {https://doi.org/10.1103/PhysRevA.71.032318} {\bibfield  {journal} {\bibinfo  {journal} {Phys. Rev. A}\ }\textbf {\bibinfo {volume} {71}},\ \bibinfo {pages} {032318} (\bibinfo {year} {2005})}\BibitemShut {NoStop}%
	\bibitem [{\citenamefont {Yoganathan}\ \emph {et~al.}(2019)\citenamefont {Yoganathan}, \citenamefont {Jozsa},\ and\ \citenamefont {Strelchuk}}]{YogaJS2019}%
	\BibitemOpen
	\bibfield  {author} {\bibinfo {author} {\bibfnamefont {M.}~\bibnamefont {Yoganathan}}, \bibinfo {author} {\bibfnamefont {R.}~\bibnamefont {Jozsa}},\ and\ \bibinfo {author} {\bibfnamefont {S.}~\bibnamefont {Strelchuk}},\ }\bibfield  {title} {\bibinfo {title} {Quantum advantage of unitary {C}lifford circuits with magic state inputs},\ }\href {https://doi.org/10.1098/rspa.2018.0427} {\bibfield  {journal} {\bibinfo  {journal} {Proc. R. Soc. A}\ }\textbf {\bibinfo {volume} {475}},\ \bibinfo {pages} {20180427} (\bibinfo {year} {2019})}\BibitemShut {NoStop}%
	\bibitem [{\citenamefont {Peres}(2023{\natexlab{a}})}]{Peres2023}%
	\BibitemOpen
	\bibfield  {author} {\bibinfo {author} {\bibfnamefont {F.~C.~R.}\ \bibnamefont {Peres}},\ }\bibfield  {title} {\bibinfo {title} {Pauli-based model of quantum computation with higher-dimensional systems},\ }\href {https://doi.org/10.1103/PhysRevA.108.032606} {\bibfield  {journal} {\bibinfo  {journal} {Phys. Rev. A}\ }\textbf {\bibinfo {volume} {108}},\ \bibinfo {pages} {032606} (\bibinfo {year} {2023}{\natexlab{a}})}\BibitemShut {NoStop}%
	\bibitem [{\citenamefont {Peres}(2023{\natexlab{b}})}]{GitRepoPBC}%
	\BibitemOpen
	\bibfield  {author} {\bibinfo {author} {\bibfnamefont {F.~C.~R.}\ \bibnamefont {Peres}},\ }\href@noop {} {\bibinfo {title} {{CompHybPBC}}},\ \bibinfo {howpublished} {\url{https://github.com/fcrperes/CompHybPBC}} (\bibinfo {year} {2023}{\natexlab{b}})\BibitemShut {NoStop}%
	\bibitem [{\citenamefont {S.~H. C.~du Toit}(1986)}]{Dutoit1986statistics}%
	\BibitemOpen
	\bibfield  {author} {\bibinfo {author} {\bibfnamefont {R.~H.~S.}\ \bibnamefont {S.~H. C.~du Toit}, \bibfnamefont {A.~G. W.~Steyn}},\ }\href {https://doi.org/https://doi.org/10.1007/978-1-4612-4950-4} {\emph {\bibinfo {title} {{Graphical Exploratory Data Analysis}}}}\ (\bibinfo  {publisher} {Springer New York, NY},\ \bibinfo {year} {1986})\BibitemShut {NoStop}%
	\bibitem [{\citenamefont {Danos}\ and\ \citenamefont {Kashefi}(2006)}]{DanosKashefi2006}%
	\BibitemOpen
	\bibfield  {author} {\bibinfo {author} {\bibfnamefont {V.}~\bibnamefont {Danos}}\ and\ \bibinfo {author} {\bibfnamefont {E.}~\bibnamefont {Kashefi}},\ }\bibfield  {title} {\bibinfo {title} {{Determinism in the one-way model}},\ }\href {https://doi.org/10.1103/PhysRevA.74.052310} {\bibfield  {journal} {\bibinfo  {journal} {Phys. Rev. A}\ }\textbf {\bibinfo {volume} {74}},\ \bibinfo {pages} {052310} (\bibinfo {year} {2006})}\BibitemShut {NoStop}%
	\bibitem [{\citenamefont {Browne}\ \emph {et~al.}(2007)\citenamefont {Browne}, \citenamefont {Kashefi}, \citenamefont {Mhalla},\ and\ \citenamefont {Perdrix}}]{BrowneKashefi+2007}%
	\BibitemOpen
	\bibfield  {author} {\bibinfo {author} {\bibfnamefont {D.~E.}\ \bibnamefont {Browne}}, \bibinfo {author} {\bibfnamefont {E.}~\bibnamefont {Kashefi}}, \bibinfo {author} {\bibfnamefont {M.}~\bibnamefont {Mhalla}},\ and\ \bibinfo {author} {\bibfnamefont {S.}~\bibnamefont {Perdrix}},\ }\bibfield  {title} {\bibinfo {title} {{Generalized flow and determinism in measurement-based quantum computation}},\ }\href {https://doi.org/10.1088/1367-2630/9/8/250} {\bibfield  {journal} {\bibinfo  {journal} {New Journal of Physics}\ }\textbf {\bibinfo {volume} {9}},\ \bibinfo {pages} {250} (\bibinfo {year} {2007})}\BibitemShut {NoStop}%
	\bibitem [{\citenamefont {Markov}\ and\ \citenamefont {Shi}(2008)}]{MarkovShi2008}%
	\BibitemOpen
	\bibfield  {author} {\bibinfo {author} {\bibfnamefont {I.~L.}\ \bibnamefont {Markov}}\ and\ \bibinfo {author} {\bibfnamefont {Y.}~\bibnamefont {Shi}},\ }\bibfield  {title} {\bibinfo {title} {{Simulating Quantum Computation by Contracting Tensor Networks}},\ }\href {https://doi.org/10.1137/050644756} {\bibfield  {journal} {\bibinfo  {journal} {SIAM Journal on Computing}\ }\textbf {\bibinfo {volume} {38}},\ \bibinfo {pages} {963–981} (\bibinfo {year} {2008})}\BibitemShut {NoStop}%
	\bibitem [{\citenamefont {van~de Wetering}(2020)}]{Wetering2020zxcalculus}%
	\BibitemOpen
	\bibfield  {author} {\bibinfo {author} {\bibfnamefont {J.}~\bibnamefont {van~de Wetering}},\ }\href {https://doi.org/10.48550/ARXIV.2012.13966} {\bibinfo {title} {{ZX-calculus for the working quantum computer scientist}}},\ \bibinfo {howpublished} {arXiv:2012.13966 [quant-ph]} (\bibinfo {year} {2020})\BibitemShut {NoStop}%
	\bibitem [{\citenamefont {Staudacher}\ \emph {et~al.}(2023)\citenamefont {Staudacher}, \citenamefont {Guggemos}, \citenamefont {Grundner-Culemann},\ and\ \citenamefont {Gehrke}}]{Staudacher2023}%
	\BibitemOpen
	\bibfield  {author} {\bibinfo {author} {\bibfnamefont {K.}~\bibnamefont {Staudacher}}, \bibinfo {author} {\bibfnamefont {T.}~\bibnamefont {Guggemos}}, \bibinfo {author} {\bibfnamefont {S.}~\bibnamefont {Grundner-Culemann}},\ and\ \bibinfo {author} {\bibfnamefont {W.}~\bibnamefont {Gehrke}},\ }\bibfield  {title} {\bibinfo {title} {{Reducing 2-QuBit Gate Count for ZX-Calculus based Quantum Circuit Optimization}},\ }\href {https://doi.org/10.4204/eptcs.394.3} {\bibfield  {journal} {\bibinfo  {journal} {Electronic Proceedings in Theoretical Computer Science}\ }\textbf {\bibinfo {volume} {394}},\ \bibinfo {pages} {29–45} (\bibinfo {year} {2023})}\BibitemShut {NoStop}%
	\bibitem [{\citenamefont {Broadbent}\ and\ \citenamefont {Kashefi}(2009)}]{BroadbentK2009parallelizing}%
	\BibitemOpen
	\bibfield  {author} {\bibinfo {author} {\bibfnamefont {A.}~\bibnamefont {Broadbent}}\ and\ \bibinfo {author} {\bibfnamefont {E.}~\bibnamefont {Kashefi}},\ }\bibfield  {title} {\bibinfo {title} {Parallelizing quantum circuits},\ }\href {https://doi.org/https://doi.org/10.1016/j.tcs.2008.12.046} {\bibfield  {journal} {\bibinfo  {journal} {Theoretical Computer Science}\ }\textbf {\bibinfo {volume} {410}},\ \bibinfo {pages} {2489} (\bibinfo {year} {2009})}\BibitemShut {NoStop}%
	\bibitem [{\citenamefont {Hoban}\ \emph {et~al.}(2014)\citenamefont {Hoban}, \citenamefont {Wallman}, \citenamefont {Anwar}, \citenamefont {Usher}, \citenamefont {Raussendorf},\ and\ \citenamefont {Browne}}]{HobanWA2014mbcc}%
	\BibitemOpen
	\bibfield  {author} {\bibinfo {author} {\bibfnamefont {M.~J.}\ \bibnamefont {Hoban}}, \bibinfo {author} {\bibfnamefont {J.~J.}\ \bibnamefont {Wallman}}, \bibinfo {author} {\bibfnamefont {H.}~\bibnamefont {Anwar}}, \bibinfo {author} {\bibfnamefont {N.}~\bibnamefont {Usher}}, \bibinfo {author} {\bibfnamefont {R.}~\bibnamefont {Raussendorf}},\ and\ \bibinfo {author} {\bibfnamefont {D.~E.}\ \bibnamefont {Browne}},\ }\bibfield  {title} {\bibinfo {title} {Measurement-based classical computation},\ }\href {https://doi.org/10.1103/PhysRevLett.112.140505} {\bibfield  {journal} {\bibinfo  {journal} {Phys. Rev. Lett.}\ }\textbf {\bibinfo {volume} {112}},\ \bibinfo {pages} {140505} (\bibinfo {year} {2014})}\BibitemShut {NoStop}%
	\bibitem [{\citenamefont {Bartolucci}\ \emph {et~al.}(2023)\citenamefont {Bartolucci}, \citenamefont {Birchall}, \citenamefont {Bombín}, \citenamefont {Cable}, \citenamefont {Dawson}, \citenamefont {Gimeno-Segovia}, \citenamefont {Johnston}, \citenamefont {Kieling}, \citenamefont {Nickerson}, \citenamefont {Pant}, \citenamefont {Pastawski}, \citenamefont {Rudolph},\ and\ \citenamefont {Sparrow}}]{FBQC2023}%
	\BibitemOpen
	\bibfield  {author} {\bibinfo {author} {\bibfnamefont {S.}~\bibnamefont {Bartolucci}}, \bibinfo {author} {\bibfnamefont {P.}~\bibnamefont {Birchall}}, \bibinfo {author} {\bibfnamefont {H.}~\bibnamefont {Bombín}}, \bibinfo {author} {\bibfnamefont {H.}~\bibnamefont {Cable}}, \bibinfo {author} {\bibfnamefont {C.}~\bibnamefont {Dawson}}, \bibinfo {author} {\bibfnamefont {M.}~\bibnamefont {Gimeno-Segovia}}, \bibinfo {author} {\bibfnamefont {E.}~\bibnamefont {Johnston}}, \bibinfo {author} {\bibfnamefont {K.}~\bibnamefont {Kieling}}, \bibinfo {author} {\bibfnamefont {N.}~\bibnamefont {Nickerson}}, \bibinfo {author} {\bibfnamefont {M.}~\bibnamefont {Pant}}, \bibinfo {author} {\bibfnamefont {F.}~\bibnamefont {Pastawski}}, \bibinfo {author} {\bibfnamefont {T.}~\bibnamefont {Rudolph}},\ and\ \bibinfo {author} {\bibfnamefont {C.}~\bibnamefont {Sparrow}},\ }\bibfield  {title} {\bibinfo {title} {{Fusion-based quantum computation}},\ }\href {https://doi.org/10.1038/s41467-023-36493-1} {\bibfield  {journal} {\bibinfo
			{journal} {Nature Communications}\ }\textbf {\bibinfo {volume} {14}},\ \bibinfo {pages} {912} (\bibinfo {year} {2023})}\BibitemShut {NoStop}%
	\bibitem [{\citenamefont {Browne}\ and\ \citenamefont {Rudolph}(2005)}]{BrowneRudolph2005}%
	\BibitemOpen
	\bibfield  {author} {\bibinfo {author} {\bibfnamefont {D.~E.}\ \bibnamefont {Browne}}\ and\ \bibinfo {author} {\bibfnamefont {T.}~\bibnamefont {Rudolph}},\ }\bibfield  {title} {\bibinfo {title} {{R}esource-{E}fficient {L}inear {O}ptical {Q}uantum {C}omputation},\ }\href {https://doi.org/10.1103/PhysRevLett.95.010501} {\bibfield  {journal} {\bibinfo  {journal} {Phys. Rev. Lett.}\ }\textbf {\bibinfo {volume} {95}},\ \bibinfo {pages} {010501} (\bibinfo {year} {2005})}\BibitemShut {NoStop}%
\end{thebibliography}
%

\end{document}